\pdfoutput=1
\documentclass[11pt]{article}

\usepackage{amsmath,amsthm,amssymb, tabu,fancyhdr} \usepackage{fullpage}
\usepackage{booktabs} \usepackage{dsfont} \usepackage{enumitem} \usepackage{float} \usepackage[margin=1in]{geometry} \usepackage{graphicx} \usepackage{mathtools}
\usepackage{enumitem}
\usepackage[pagebackref]{hyperref}
\usepackage{verbatim}
\usepackage[font=small,labelfont=bf]{caption}
\usepackage{makecell}
\usepackage{tabularx,booktabs,ragged2e}
\usepackage{thmtools} 
\usepackage{thm-restate}

\usepackage[utf8]{inputenc} \usepackage{mathpazo} \usepackage[T1]{fontenc}
\usepackage{soul} \usepackage{suffix} \usepackage{tikz} \usepackage{xcolor} \usepackage{xspace} \usepackage[scr=rsfso]{mathalfa}

\usepackage[linesnumbered,boxed]{algorithm2e}
\SetAlCapSkip{1em}
\SetAlCapNameFnt{\small}
\SetAlCapFnt{\small}
\SetAlgorithmName{Pseudocode}{}{}

\newcolumntype{L}[1]{>{\raggedright\arraybackslash}p{#1}}
 
\newcolumntype{C}[1]{>{\centering\arraybackslash}p{#1}}
 
\newcolumntype{R}[1]{>{\raggedleft\arraybackslash}p{#1}}

\definecolor{weborange}{rgb}{.8,.3,.3}
\definecolor{webblue}{rgb}{0,0,.8}
\definecolor{internallinkcolor}{rgb}{0,.5,0}
\definecolor{externallinkcolor}{rgb}{0,0,.5}

\usepackage{cleveref}

\definecolor{DarkBlue}{rgb}{0,0,0.8}  \definecolor{DarkOrange}{rgb}{0.8,0.4,0}  \def\mylinkcolor{DarkBlue}

\hypersetup{colorlinks=true, urlcolor=\mylinkcolor,
  linkcolor=\mylinkcolor, citecolor=\mylinkcolor}

\usetikzlibrary{calc,positioning}

\newcommand{\hynote}[1]{\textcolor{magenta}{\small (Henry: #1)}}
\newcommand{\hmnote}[1]{\textcolor{red}{\small (Hamoon: #1)}}
\newcommand{\snote}[1]{\textcolor{blue}{\small (Sajjad: #1)}}

\renewcommand{\hynote}[1]{\textcolor{magenta}{}}
\renewcommand{\hmnote}[1]{\textcolor{red}{}}
\renewcommand{\snote}[1]{\textcolor{blue}{}}

\definecolor{White}{rgb}{1,1,1} \definecolor{Black}{rgb}{0,0,0} \definecolor{LightGray}{rgb}{.8,.8,.8} \colorlet{ChannelColor}{LightGray} \colorlet{ChannelTextColor}{Black} \colorlet{ReadoutColor}{White}

\renewcommand{\cal}[1]{\mathcal{#1}}

\newcommand{\N}{\mathbb{N}}

\numberwithin{equation}{section}
\newtheorem{theorem}{Theorem}[section]
\newtheorem{lemma}[theorem]{Lemma}
\newtheorem{claim}[theorem]{Claim}
\newtheorem{conjecture}[theorem]{Conjecture}
\newtheorem{corollary}[theorem]{Corollary}
\newtheorem{proposition}[theorem]{Proposition}
\newtheorem{definition}[theorem]{Definition}
\newtheorem{remark}[theorem]{Remark}

 \newcommand{\complex}{\mathbb{C}} \let\epsilon=\varepsilon

\newcommand{\microspace}{\mspace{.5mu}} \newcommand{\ket}[1]{\ensuremath{\lvert\microspace #1
    \microspace\rangle}} \newcommand{\bra}[1]{\ensuremath{\langle\microspace #1
    \microspace\rvert}} \newcommand{\ketbra}[2]{\ensuremath{\lvert\microspace #1
    \microspace\rangle\! \langle \microspace #2 \microspace \rvert}}

\newcommand{\paren}[1]{(#1)}

\newcommand{\Bigparen}[1]{\Big(#1\Big)}

\newcommand{\class}[1]{\mathsf{#1}} 

     \newcommand{\QMIP}{\class{QMIP}} \WithSuffix\newcommand\QMIP*{\ensuremath{\class{QMIP}^*}} \newcommand{\PSPACE}{\class{PSPACE}}  \newcommand{\MIP}{\class{MIP}} \WithSuffix\newcommand\MIP*{\ensuremath{\class{MIP}^*}}      \newcommand{\RE}{\class{RE}}

\DeclareMathOperator*{\E}{\mathbf{E}}

\newcommand{\val}{\omega}

\newcommand{\valco}{\omega_{co}}

\newcommand{\TIME}{\mathsf{TIME}}
\newcommand{\alg}[1]{\texttt{#1}}
\newcommand{\strategy}{\mathscr{S}}

\DeclareMathOperator{\tr}{\mathrm{tr}}

\newcommand{\id}{\mathbb{1}}
\newcommand{\eps}{\varepsilon}
\newcommand{\poly}{\mathrm{poly}}

\newcommand{\hilb}{\cal{H}}
\newcommand{\B}{\mathrm{B}}

\newcommand{\algebra}{\mathscr{A}}
\newcommand{\algebraB}{\mathscr{B}}

\newcommand{\hx}{{\hat{x}}}

\newcommand{\cw}{{\overline{W}}}
\newcommand{\ca}{{\overline{A}}}
\newcommand{\cb}{{\overline{B}}}

\newcommand{\ans}{\mathrm{ans}}

\newcommand{\super}{\mathrm{super}}

\newcommand{\orac}{\mathrm{orac}}

\newcommand{\alice}{A}
\newcommand{\bob}{B}

\newcommand{\intro}{\mathrm{intro}}

\newcommand{\Succ}{\mathrm{succ}}

\newcommand{\qs}{\mathrm{QS}}
\newcommand{\sample}{S}
\newcommand{\erase}{E}
\newcommand{\introspect}{I}

\newcommand{\epr}{\mathrm{EPR}}

\newcommand{\verifier}{\mathscr{V}}
\newcommand{\UGS}{\mathscr{G}}

\newcommand{\coRE}{\mathsf{coRE}}
\newcommand{\R}{\mathbb{R}}
\newcommand{\C}{\mathbb{C}}

\newcommand{\Compress}{\texttt{Compress}}
\newcommand{\GapCompress}{\texttt{GappedCompress}}
\newcommand{\GaplessCompress}{\texttt{GaplessCompress}}
\newcommand{\Halt}{\texttt{Halt}}

\newcommand{\trace}[1]{\tau \left ( #1 \right )}

\allowdisplaybreaks

\title{Nonlocal Games, Compression Theorems, \\ and the Arithmetical Hierarchy}

\author{
Hamoon Mousavi\thanks{\texttt{h.mousavi@cs.columbia.edu}} \\ \small{Columbia University}
\and Seyed Sajjad Nezhadi\thanks{\texttt{sajjad@umd.edu}} \\ \small{University of Maryland} 
\and Henry Yuen\thanks{\texttt{hyuen@cs.columbia.edu}} \\ \small{Columbia University}
}

\date{} 

\begin{document}

\maketitle

\begin{abstract}

We investigate the connection between the complexity of nonlocal games and the
arithmetical hierarchy, a classification of languages according to the complexity of arithmetical
formulas defining them. It was recently shown by Ji, Natarajan, Vidick, Wright and Yuen that
deciding whether the (finite-dimensional) quantum value of a nonlocal game is $1$ or at most $\frac{1}{2}$ is complete for the class $\Sigma_1$ (i.e., $\RE$). A result of Slofstra implies that deciding whether the commuting
operator value of a nonlocal game is equal to $1$ is complete for the class $\Pi_1$ (i.e., $\mathsf{coRE}$). 

We prove that deciding whether the quantum value of a two-player nonlocal game is exactly equal to $1$ is complete for $\Pi_2$; this class is in the second level of the arithmetical hierarchy and corresponds to formulas 
of the form ``$\forall x \, \exists y \, \phi(x,y)$''. This shows that exactly computing the quantum value is strictly harder than approximating it, and also strictly harder than computing the commuting operator value (either exactly or approximately). 

We explain how results about the complexity of nonlocal games all follow in a unified manner from a technique known as \emph{compression}. At the core of our $\Pi_2$-completeness result is a new ``gapless'' compression theorem that holds for both quantum and commuting operator strategies. Our compression theorem yields as a byproduct an alternative proof of Slofstra's result that the set of quantum correlations is not closed. We also show how a ``gap-preserving'' compression theorem for commuting operator strategies would imply that approximating the commuting operator value is complete for $\Pi_1$.

\end{abstract}

\newpage
\tableofcontents

\newpage

\newpage

\section{Introduction}
\label{section:introduction}

A nonlocal game describes a scenario in which a (classical) verifier plays a game with two separated, but possibly entangled, players (who we'll call Alice and Bob). In the game, the verifier samples a pair of questions $(x,y)$ from a question distribution $\mu$, sends $x$ to Alice and $y$ to Bob, and then receives answers $a$ and $b$ from the players. The verifier then computes a decision procedure $D(x,y,a,b)$ to determine whether the players win or lose. We assume that Alice and Bob know the question distribution and decision procedure before the game starts, and cooperatively select an entangled strategy to maximize their probability of winning.

Recent results have shown that the optimal winning probability, called the \emph{value}, of a nonlocal game is uncomputable in general. Surprisingly, the study of the complexity of nonlocal games is also intimately tied to questions outside of complexity theory. For example, Slofstra's result about the undecidability of whether a nonlocal game has a perfect quantum strategy (i.e.\ a strategy that wins with probability $1$) was a byproduct of his showing that the set of quantum correlations is not closed~\cite{slofstra_tsirelsons_problem_and_an_embedding_theorem,slofstra_set_of_quantum_correlations}. As another example, the complexity-theoretic result $\MIP^* = \RE$~\cite{ji_mip_re} (which implies that there is no algorithm to even \emph{approximate} the quantum value of a nonlocal game) yields negative answers to both Tsirelson's Problem from quantum information theory and Connes' Embedding Problem from operator algebras~\cite{connes1976classification,ozawa_connes_embedding}. 

These uncomputability results for nonlocal games demonstrate that the space of quantum strategies is terribly complex --- no algorithm can optimize over them, even approximately! This is already quite striking, but a closer look at these results indicates that more can be said: different computational problems for nonlocal games can be uncomputable in \emph{incomparable ways}. To explain this we need to define two relevant models of entangled strategies.

\paragraph{Strategies for nonlocal games.} The most general model we consider is the class of \emph{commuting operator} strategies. Let $G = (\cal{X},\cal{A},\mu,D)$ denote a nonlocal game with question alphabet $\cal{X}$, answer alphabet $\cal{A}$, question distribution $\mu$, and decision procedure $D: \cal{X} \times \cal{X} \times \cal{A} \times \cal{A} \to \{0,1\}$.  A commuting operator strategy $\strategy$ for a game $G$ is specified by the following data: a separable Hilbert space $\hilb$, a unit vector $\ket{\psi} \in \hilb$ (called the \emph{state}), and sets of \emph{measurements} $A = \{ A^x \}_{x \in \cal{X}}$ and $B = \{ B^y \}_{y \in \cal{X}}$ acting on $\hilb$ satisfying the following: 
\begin{itemize}
	\item For all $x,y$, the measurements $A^x = \{A^x_a\}_{a \in \cal{A}}$ and $B^y = \{B^y_b\}_{b \in \cal{A}}$ are sets of bounded positive operators on $\hilb$, with each set summing to the identity, and
	\item For all $x,y,a,b$, the operators $A^x_a$ and $B^y_b$ commute.
\end{itemize}
Given questions $(x,y)$, the probability that the players respond with answers $(a,b)$ is given by $\bra{\psi}A^x_a \, B^y_b \ket{\psi}$. The two conditions on the measurement operators above ensure that this is a valid probability distribution over $\cal{A} \times \cal{A}$, and furthermore the commutation condition ensures that the strategy is \emph{non-signaling}, meaning that the marginal probability that a player responds with an answer only depends on their question (and not the other player's question).

The \emph{value} of a commuting operator strategy $\strategy = (\ket{\psi},A,B)$ in a game $G$ is given by 
\[
	\val(G,\strategy) := \sum_{x,y,a,b} \mu(x,y) \cdot \bra{\psi}A^x_a B^y_b \ket{\psi} \cdot D(x,y,a,b)~.
\]
The \emph{commuting operator} value of a game $G$ is defined as 
\[
	\val_{co}(G) := \sup_{\text{\tiny commuting operator } \strategy} \, \val(G,\strategy).
\]
Intuitively, the commuting operator value of a game represents the players' maximum success probability allowed under quantum mechanics.

An important subclass of commuting operator strategies are the \emph{finite-dimensional} ones, i.e. where the underlying Hilbert space $\hilb$ is equal to $\C^d$ for some integer $d$. We define the \emph{quantum value}\footnote{The reason for this name, as opposed to ``finite-dimensional value'', is historical: the study of nonlocal games has largely focused on the setting of finite-dimensional strategies.} of a game $G$ to be
\[
	\val_q(G) := \sup_{\text{\tiny finite-dimensional } \strategy} \, \val(G,\strategy).
\] 
In the finite-dimensional setting, commuting operator strategies coincide with strategies in the \emph{tensor product model}: one can find two finite-dimensional Hilbert spaces $\hilb_A, \hilb_B$, a bipartite state $\ket{\tilde{\psi}} \in \hilb_A \otimes \hilb_B$, and measurements $\{\tilde{A}^x_a \}$ on $\hilb_A$ and $\{\tilde{B}^y_b\}$ on $\hilb_B$ such that
\[
	\bra{\psi}A^x_a \, B^y_b \ket{\psi} = \bra{\tilde{\psi}} \tilde{A}^x_a \otimes \tilde{B}^y_b \ket{\tilde{\psi}}~.
\]
For a proof, see~\cite[Theorem 1]{scholz2008tsirelson}. Tensor product strategies give a natural way to model the behavior of spatially separated players, and this is perhaps the most commonly studied model of strategies for nonlocal games. General commuting operator strategies, on the other hand, do not assume that there is an \emph{a priori} tensor product decomposition of the Hilbert space, but only that the non-signaling property is enforced via commutativity of the players' measurements. The commuting operator model of quantum correlations arise naturally in algebraic formulations of quantum field theory~\cite{scholz2008tsirelson,fritz2012tsirelson}.

It is easy to see that $\val_q(G) \leq \val_{co}(G)$. Tsirelson's Problem is essentially a question about whether $\val_q(G) = \val_{co}(G)$ for all games $G$; in other words, can all commuting operator strategies (which might be infinite dimensional) be approximated arbitrarily well by finite-dimensional ones~\cite{scholz2008tsirelson}? Furthermore, it was shown that Tsirelson's Problem is equivalent to Connes' Embedding Problem, which was a long-standing question in operator algebras about the approximability of von Neumann algebras by finite-dimensional matrix algebras~\cite{connes1976classification,junge_connes_embedding_problem,fritz2012tsirelson,ozawa_connes_embedding}. As previously mentioned, these questions about finite-dimensional approximability of infinite-dimensional objects are intimately connected to questions about computability of the value of nonlocal games.

\paragraph{Computability of nonlocal games.} We now define computational problems associated with computing the value of nonlocal games. Fix $0 \leq \eps < 1$ and a value type $t \in \{q,co\}$. Define two sets of nonlocal games
\[
	L_t^{yes} := \{ G : \omega_t(G) = 1 \} \qquad \text{and} \qquad L_{t,\eps}^{no} := \{ G : \omega_t(G) < 1 - \eps \}~.
\]
These two sets are disjoint, and when $\eps = 0$, the union of these two sets is all nonlocal games. These two sets give rise to a decision problem: given a nonlocal game $G$ in the union $L_t^{yes} \cup L_{t,\eps}^{no}$, decide whether $G$ is a ``yes'' instance or a ``no'' instance. 

When $\eps = 0$, this decision problem corresponds to \emph{exactly} computing either the quantum or commuting operator value. When $\eps > 0$, this problem corresponds to \emph{approximating} the value, because being able to compute $\val_t(G)$ up to additive $\pm \frac{\eps}{2}$ error allows one to decide whether $G \in L_t^{yes}$ or $G \in L_{t,\eps}^{no}$. Thus we call deciding between $L_t^{yes}$ and $L_{t,0}^{no}$ the \emph{exact $t$-value problem}, and deciding between $L_t^{yes}$ and $L_{t,\eps}^{no}$ for $\eps > 0$ the \emph{approximate $t$-value problem} (we usually think of $\eps$ as $1/2$, but the specific value is immaterial, as long as it is strictly positive). 

We summarize the results known so far about the computability of nonlocal games:
\begin{enumerate}

\item In~\cite{slofstra_tsirelsons_problem_and_an_embedding_theorem}, Slofstra showed that the exact $co$-value problem is hard for the class $\coRE$, which is the complement of $\RE$, the set of recursively enumerable languages. In other words, there exists a computable reduction from Turing machines $M$ to nonlocal games $G$ such that $\val_{co}(G) = 1$ if and only if $M$ does \emph{not} halt.

Furthermore, the exact $co$-value problem is contained in $\coRE$ due to the existence of a semidefinite programming hierarchy that converges from above to the commuting operator value of a given nonlocal game~\cite{navascues2008convergent,doherty2008quantum}. Thus the exact $co$-value problem is complete for $\coRE$. 

\item In~\cite{slofstra_set_of_quantum_correlations}, Slofstra showed that the exact $q$-value problem is also hard for $\coRE$. However, no upper bound on the complexity of the exact $q$-value problem was given.

\item In~\cite{ji_mip_re}, Ji, Natarajan, Vidick, Wright and Yuen showed that the approximate $q$-value problem is hard for $\RE$. In other words, there exists a computable reduction from Turing machines $M$ to nonlocal games $G$ such that if $M$ halts then $\val_q(G) = 1$, otherwise $\val_q(G) \leq \frac{1}{2}$. 

Furthermore, the approximate $q$-value problem is contained in $\RE$ due to the fact that a brute-force enumeration algorithm can find a finite-dimensional strategy that succeeds with probability arbitrarily close to $1$, provided that $\val_q(G) = 1$. Thus, the approximate $q$-value problem is complete for $\RE$. 
\end{enumerate}
While these results show that the exact $q$-value, exact $co$-value, and approximate $q$-value problems are all undecidable, they are undecidable in different ways. For example, a basic result in computability theory is that the classes $\RE$ and $\coRE$ are incomparable (i.e.\ they do not contain each other). Thus the approximate $q$-value problem cannot be reduced to the exact $co$-value problem and vice versa.\footnote{The notion of reduction that we consider here are \emph{many-one reductions}, i.e., yes instances are mapped to yes instances, and no instances are mapped to no instances.} Similarly, because both $\RE$ and $\coRE$ can be reduced to it, the exact $q$-value problem must be \emph{strictly} harder than both the approximate $q$-value and exact $co$-value problem (in the sense that a Turing machine equipped with the ability to compute the exact $co$-value of a game provably cannot solve the exact $q$-value problem). 

We note that (a) since the complexities of the $q$-value and $co$-value problems are different, but (b) a positive answer to Tsirelson's Problem implies that they are the same, it must be that Tsirelson's Problem (and thus Connes' Embedding Problem) has a negative answer.

These results still leave two main open questions about the complexity of nonlocal games:
\begin{enumerate}
	\item What is the complexity of the exact $q$-value problem (i.e.\ deciding whether $\val_q(G) \stackrel{?}{=} 1$).
	\item What is the complexity of the approximate $co$-value problem (i.e.\ deciding whether $\val_{co}(G) = 1$ or $\val_{co}(G) < \frac{1}{2}$)? 
\end{enumerate}

In this paper we resolve the first open question by characterizing the complexity of the exact $q$-value problem:

\begin{theorem}
\label{thm:main-intro}
The problem of deciding whether $\val_q(G) = 1$ for nonlocal games $G$ is complete for $\Pi_2$. 
\end{theorem}

The class $\Pi_2$ is in the second level of the \emph{arithmetical hierarchy}, which is an infinite hierarchy of complexity classes\footnote{In computability theory these classes are usually denoted as $\Sigma_k^0$ and $\Pi_k^0$. For simplicity we have dropped the superscripts.} $\bigcup_{k=0}^\infty \Sigma_k$ and $\bigcup_{k=0}^\infty\Pi_k$ that characterize the complexity of languages according to \emph{arithmetical formulas} that define them. The class $\Sigma_k$ consists of all languages reducible to deciding whether a given \emph{$\Sigma_k$-sentence} is true. A $\Sigma_k$-sentence $S$ is of the form $\exists x_1 \, \forall x_2 \, \exists \cdots \, \phi(x_1,\ldots,x_k)$ for some computable predicate $\phi$. Similarly, the class $\Pi_k$ consists of all languages reducible to deciding a given $\Pi_k$-sentence is true; these are sentences of the form $\forall x_1 \, \exists x_2 \, \forall \cdots \, \phi(x_1,\ldots,x_k)$. 

At the zeroth ($k=0$) level, the classes $\Sigma_0 = \Pi_0$ correspond to the set of decidable languages, and the first level classes $\Sigma_1$ and $\Pi_1$ are simply $\RE$ and $\coRE$ respectively. The class $\Pi_2$ is in the second level of the arithmetical hierarchy, and contains both $\Sigma_1$ and $\Pi_1$. It is a well-known fact from computability theory that the levels of the arithmetical hierarchy are all distinct, and furthermore $\Sigma_k \neq \Pi_k$ for all $k \geq 1$.

Although we do not resolve the second open question, it is conjectured that the approximate $co$-value problem is complete for $\coRE = \Pi_1$. A positive resolution of this conjecture would complete the picture of the computability landscape of nonlocal games, depicted in \Cref{fig:table}, and give a pleasing correspondence between different nonlocal game problems and classes in the arithmetical hierarchy.

\begin{figure}[H]
\begin{center}
{\renewcommand{\arraystretch}{1.5}
\begin{tabular}{  c | l l  } 

	& $\eps = 0$ & $\eps > 0$ \\ 
\hline
$\omega_q(G) \pm \eps $ & $\Pi_2$ (this paper) \qquad \qquad & $\Sigma_1$~\cite{ji_mip_re} \\ 
$\omega_{co}(G) \pm \eps $ & $\Pi_1$~\cite{slofstra_tsirelsons_problem_and_an_embedding_theorem} & $\Pi_1$ (conjectured) \\ 

\end{tabular}
}
\end{center}
\caption{A characterization of the complexity of computing the value of a nonlocal game in terms of the arithmetical hierarchy, depending on whether the quantum or commuting operator value is being considered, and whether the value is being computed exactly or approximately. The top left entry is the main result of this paper, and the lower right entry is conjectured.}
\label{fig:table}
\end{figure}

We mention that the approximate and exact $q$- and $co$-value problems are used in defining the four complexity classes $\MIP^*$, $\MIP^*_0$, $\MIP^{co}$ and $\MIP^{co}_0$, respectively. In particular, the above figure corresponds to the results $\MIP^*=\RE = \Sigma_1$, $\MIP^*_0 = \Pi_2$ and $\MIP^{co} \subseteq \MIP^{co}_0 = \coRE = \Pi_1$.

\emph{A priori}, this tight correspondence between nonlocal games and the arithmetical hierarchy seems quite surprising. On one hand, computing the value of a nonlocal game corresponds to a continuous optimization problem over a space of quantum states and quantum measurements, possibly in infinite dimensions. On the other hand, deciding whether a quantified sentence is true is a discrete problem in symbolic logic ostensibly having nothing to do with quantum physics. Furthermore, the reader may notice that there are several interesting asymmetries in \Cref{fig:table}, illustrating that this correspondence has rich and unexpected behavior: if we assume the conjecture about the approximate $co$-value problem, then both exact and approximate computation of the commuting operator value are equivalent to deciding $\Pi_1$-sentences, whereas for the quantum value, the complexity splits depending on whether we are considering exact or approximate computation. 

\paragraph{Connections with noncommutative polynomial optimization.} We also point out that the aforementioned complexity results can be viewed as characterizations of the complexity of \emph{noncommutative polynomial optimization}, an important subject in mathematics, physics and computer science~\cite{navascues2008convergent,doherty2008quantum,pironio2010random,netzer2014hyperbolic}. The general formulation of noncommutative polynomial optimization (ncPO for short) is the following: given Hermitian polynomials $p,q_1,\ldots,q_m$ in $2n$-noncommutative variables $(x_1,\ldots,x_n,x_1^*,\ldots,x_n^*)$ over $\C$, compute the value of the following optimization program:
\begin{align*}
	\sup \qquad & \bra{\phi} p(X) \ket{\phi} \\
		\text{s.t.} \qquad &q_i(X) \succeq 0 \qquad \text{for $i=1,\ldots,m$}
\end{align*}
The supremum is over choices of tuples $(\cal{H},X,\phi)$ where $\cal{H}$ is a Hilbert space, $X$ is an $n$-tuple of bounded operators acting on $\cal{H}$, and $\ket{\phi}$ is a unit vector on $\cal{H}$. The notation $p(X)$ and $q_i(X)$ indicates that we evaluate each of the indeterminates $x_i$ with the operator $X_i$ and $x_i^*$ with the adjoint $X_i^*$, respectively. We consider two different variations of a ncPO program $P$; if we restrict the supremum to vary only over finite -- but unbounded -- dimensional Hilbert spaces then we call the program \emph{finite-dimensional} and let $\omega_{\mathrm{fin}}(P)$ denote the value of the program. Otherwise we call the program \emph{infinite-dimensional} and let $\omega_{\infty}(P)$ denote the value.

The complexity results in \Cref{fig:table} can be recast as the following. Given a ncPO program $P$ and a real number $c \in R$, deciding whether
\begin{enumerate}
	\item $\omega_{\mathrm{fin}}(P) \geq c$ is complete for $\Pi_2$. 
	\item $\omega_{\infty}(P) \geq c$ is complete for $\Pi_1$.
	\item $|\omega_{\mathrm{fin}}(P) - c | \leq \eps$ for fixed $\eps > 0$ is complete for $\Sigma_1$. 
\end{enumerate}
The reason for this is because on one hand we can encode the $t$-value of a nonlocal game for $t \in \{q,co\}$ as a ncPO program that is finite-dimensional if $t = q$ and infinite-dimensional if $t = co$; on the other hand the complexity of solving an ncPO program is upper-bounded by $\Pi_2$, $\Pi_1$, or $\Sigma_1$ depending on the variant of the problem. Although this connection is fairly straightforward, for completeness we provide the details in \Cref{app:polyopt}. 

We note that, by comparison, the analogous problems for \emph{commutative polynomial optimization} over $\R$ are decidable; this is because deciding whether a semialgebraic set defined by polynomial equalities/inequalities over $\R$ is empty is contained in $\PSPACE$~\cite{existentialtheoryreal}.

\medskip
\vspace{20pt}

The main conceptual result of our paper is that all of the complexity statements about nonlocal games expressed in \Cref{fig:table} can be established in a unified manner via a technique called \emph{nonlocal game compression}. At the heart of the proof of $\MIP^* = \RE$ is a \emph{gap-preserving} compression theorem for the $q$-value of games. The centerpiece of the present paper is a \emph{gapless} compression theorem that holds for both the $q$- and $co$-value of games. 
First we show that this gapless compression theorem directly gives an alternate proof of the $\Pi_1$-completeness of the exact $co$-value problem~\cite{slofstra_tsirelsons_problem_and_an_embedding_theorem}, as well as an alternate proof of Slofstra's result that the set of quantum correlations is not closed (i.e.\ there is a nonlocal game $G$ with $\val_q(G) = 1$, but there is no finite-dimensional strategy with success probability $1$)~\cite{slofstra_set_of_quantum_correlations}.

We then combine our gapless compression theorem with the gap-preserving one of~\cite{ji_mip_re} to obtain the $\Pi_2$-hardness of the exact $q$-value problem, establishing \Cref{thm:main-intro}. Finally, we also show how a gap-preserving compression theorem for the $co$-value of games would imply that the approximate $co$-value problem is complete for $\mathsf{coRE} = \Pi_1$. 

Another goal of this paper is to give a self-contained proof of a compression theorem that (a) illustrates the key ideas of the gap-preserving compression results of~\cite{natarajan_neexp,ji_mip_re}, (b) generalizes these ideas to the infinite-dimensional commuting operator setting, and (c) is presented in a language that is more accessible to researchers coming from operator algebras and related areas of mathematics. The proofs of the gap-preserving compression theorems of~\cite{natarajan_neexp,ji_mip_re} are quite involved and rely on sophisticated results ranging from self-testing~\cite{wu2016device,natarajan2018low} to the quantum soundness of the low-degree test~\cite{ito2012multi,ji2020quantum} to gap amplification methods~\cite{bavarian2017hardness}. These components are needed for the gap-preserving aspect of their compression theorem. Working in the ``gapless regime'' allows us to work with much simpler versions of these components (or circumventing them entirely). 

In \Cref{sec:compression-intro} we give an overview of how compression of nonlocal games yields the complexity characterization shown in \Cref{fig:table}. In \Cref{sec:tutorial} we give an overview of how our gapless compression theorem is proved. In \Cref{sec:synchronous} we explain the \emph{synchronous strategies framework}, which our results are expressed in. This framework gives an elegant way to work with both $q$- and $co$-type strategies in a unified manner, and brings out the connection between nonlocal games and operator algebras.

\subsection{The compression paradigm}
\label{sec:compression-intro}

Intuitively speaking, a nonlocal game compression procedure for $t$-type strategies (where $t \in \{q,co\}$) is a computable map $\Compress$ that takes an infinite sequence $\UGS = (G_n)_{n \in \N}$ of polynomial-complexity nonlocal games to another infinite sequence $\UGS' = (G_n')_{n \in \N}$ such that for every $n\in \N$, 
\begin{itemize}
	\item The optimal success probability of $t$-strategies in $G_n'$ is related in a predictable way to the optimal success probability of $t$-strategies in $G_n$, and
	\item The \emph{complexity} of the game $G_n'$ is much smaller than that of the original game $G_n$, where we measure the complexity of a game based on the number of time steps required by the verifier to compute the decision procedure. 
\end{itemize}
This second item is what motivates the name ``compression''. 

The ``polynomial-complexity'' condition on the input sequence $\UGS = (G_n)_{n \in \N}$ of games means that the complexity of each game $G_n$ is bounded by $O(n^c)$ for some constant $c > 0$, and the compression procedure $\Compress$ will depend on this constant. Furthermore, $\UGS$ and $\UGS'$ are specified via \emph{Turing machines} which play the role of the verifier for the games in the sequences. Thus the map $\Compress$ is a map from Turing machines to Turing machines. Importantly, the map $\Compress$ itself is also computable by a Turing machine.

Depending on which value type $t \in \{q,co\}$ we consider, how the optimal $t$-strategies of $G_n'$ and $G_n$ are related to each other, and how much smaller the complexity of $G_n'$ is than of $G_n$, we obtain different compression procedures. The different compression procedures, in turn, allow us to establish the different entries of the correspondence outlined in \Cref{fig:table}.

We now give a high-level sketch of this connection. 

\paragraph{Gapped compression for $q$-type strategies.} 

The $\MIP^* = \RE$ result of~\cite{ji_mip_re} relies on the following \emph{gap-preserving} (or \emph{gapped} for short) compression procedure for $q$-type strategies (i.e.\ finite-dimensional strategies).

\begin{theorem}[Gap-preserving compression, informally stated~\cite{ji_mip_re}]
\label{thm:gap-comp-informal}
There exists a computable map $\GapCompress_q$ that, given a sequence of games $\UGS = (G_n)_{n \in \N}$, outputs a sequence of games $\UGS' = (G_n')_{n \in \N}$ such that the complexity of $\UGS'$ is $O(\log n)$, and furthermore if the complexity of $\UGS$ is at most $\poly(n)$, then for all $n \in \N$,
\begin{itemize}
	\item If $\val_q(G_n) = 1$, then $\val_q(G_n') = 1$. 
	\item $\cal{E}(G_n',\frac{1}{2}) \geq \max \Big \{ \cal{E}(G_n,\frac{1}{2}) \, , \, 2^n \Big \}$. 
\end{itemize}
\end{theorem}
Here, for a nonlocal game $G$ and real number $0 \leq p \leq 1$, the quantity $\cal{E}(G,p)$ is defined to be the minimum dimension of a strategy $\strategy$ such that $\omega(G,\strategy) \geq p$. If there is no strategy that achieves winning probability $p$, then $\cal{E}(G,p)$ is defined to be $\infty$.

The reason $\GapCompress_q$ is called ``gap-preserving'' is because if $\val_q(G_n) = 1$, then $\val_q(G_n') = 1$, and otherwise if $\val_q(G_n) \leq \frac{1}{2}$, then $\val_q(G_n') \leq \frac{1}{2}$. In other words; the gap between $1$ versus $1/2$ in the two different possibilities for $\val_q(G_n)$ is preserved for $\val_q(G_n')$. The second ``if'' follows from the second item of \Cref{thm:gap-comp-informal}: if there are no finite-dimensional strategies for $G_n$ that succeed with probability at least $\frac{1}{2}$, then $\cal{E}(G_n,\frac{1}{2}) = \infty$, and therefore $\cal{E}(G_n',\frac{1}{2}) = \infty$, which implies that there is no finite-dimensional strategy for $G_n'$ that has value at least $\frac{1}{2}$. 

To show that every arithmetical sentence $S$ of the form $\exists x \, \phi(x)$ can be transformed into an equivalent game $G_S$ (which is essentially equivalent to the statement $\MIP^* = \RE$), the compression procedure of \Cref{thm:gap-comp-informal} is used to construct an infinite sequence of games $\UGS = (G_n)_{n \in \N}$ that depends on the sentence $S$. If $\phi(x)$ is true for some $x \leq n$ (meaning that $S$ is true), then the game $G_n$ has the property that $\val_q(G_n) = 1$; otherwise $G_n$ is designed to be equivalent to the game $G_{n+1}'$, the compression of $G_{n+1}$ through the gap-preserving transformation $\GapCompress_q$. In other words, the sequence of games $\UGS$ is effectively a \emph{self-compressing} sequence of games. By inductively utilizing the guarantees of the gapped compression procedure, we get that in the case that $S$ is true, we have $\val_q(G_n) = 1$ for all $n$, and if $S$ is false, $\val_q(G_n) \leq \frac{1}{2}$ for all $n$.\footnote{The choice of $\frac{1}{2}$ is inconsequential here; everything stated here holds true for any constant that's strictly less than $1$.}  Finally, the game $G_S$ is then chosen to be the first member $G_1$ of the sequence $\UGS$. 

Where does the $\poly(n)$-complexity assumption on $\UGS$ and the $O(\log n)$-complexity of $\UGS'$ consequence of \Cref{thm:gap-comp-informal} come in? We can imagine that the behavior of the verifier in the game $G_n$ is specified by the following pseudocode:

\vspace{10pt}
\IncMargin{1em}
\begin{algorithm}[H]
\DontPrintSemicolon

The verifier checks whether $\phi(x)$ is true for some $x \leq n$. If it is, then accept. 

Otherwise, compute $\UGS'$ by running $\GapCompress_q$ on the description of the sequence $\UGS$. 

Play the game $G_{n+1}'$, the $(n+1)$-st game of the sequence $\UGS'$.

\caption{The game $G_n$ encoding $\Sigma_1$-sentences.}\label{fig:intro-mip-re-game}
\end{algorithm}\DecMargin{1em}
\vspace{10pt}

For simplicity we assume that $\phi(n)$ is computable in time $O(n)$. Then the complexity of the game $G_n$ can be computed as $O(n^2) + O(1) + O(\log n) = \poly(n)$. The $O(n^2)$ comes from evaluating $\phi$ on $n$ different inputs; the $O(1)$ comes from the complexity of executing the compression procedure; and the $O(\log n)$ comes from the complexity of the compressed game $G_{n+1}'$. So the sequence of games $\UGS$ has complexity $\poly(n)$, and thus the consequences of the assumption (the first and second items) are satisfied.

\paragraph{Gapless compression for $q$- and $co$-type strategies.} 
We now turn to \emph{gapless} compression procedures. As suggested by the name, these are compression procedures that do not necessarily preserve any gap in the values of the ``input'' sequence of games. The main technical contribution of this paper is the following gapless compression theorem:

\begin{theorem}[Gapless compression, informally stated]
\label{thm:gapless-comp-informal}
For $t \in \{q,co\}$ there exists a computable map $\GaplessCompress_t$ that, given a sequence of games $\UGS = (G_n)_{n \in \N}$, outputs a sequence of games $\UGS' = (G_n')_{n \in \N}$ such that the complexity of $\UGS'$ is $O(\log n)$, and furthermore if the complexity of $\UGS$ is at most $\poly(n)$, then for all $n \in \N$,
\begin{itemize}
	\item If $\val_t(G_n) < 1$, then $\val_t(G_n') < 1$.
	\item $\val_t(G_n') \geq (1 - \alpha) + \alpha \val_t(G_n)$, where $0 < \alpha < 1$ is a universal constant.
	\item $\cal{E}(G_n',1) \geq \max \Big \{ \cal{E}(G_n,1) \, , \, 2^{2n} \Big \}$. 
\end{itemize}
\end{theorem}

Notice that the first and second items imply that $\val_t(G_n) =1$ if and only if $\val_t(G_n') = 1$. In the case of $t = q$, this gapless compression theorem appears to be a weaker version of \Cref{thm:gap-comp-informal}, except the second item makes it incomparable: whereas the gapped compression theorem only works on games that either have value $1$ or at most $\frac{1}{2}$, the gapless compression theorem works for all games. In fact, the compression procedure of \Cref{thm:gapless-comp-informal} is \emph{gap-shrinking}: given a game $G_n$ with value $\val_t(G_n) < 1$, the compressed game $G_n'$ has value $\val_t(G_n) < \val_t(G_n') < 1$. Intuitively, by repeatedly applying a gapless compress procedure to an initial game with value strictly less than $1$, the sequence of compressed games obtained have value that get arbitrarily close to $1$. 

Gapless compression theorems allow us to show that deciding the truth of sentences $S$ of the form $\forall x \, \phi(x)$ (i.e. $\Pi_1$-sentences) can be reduced to deciding whether the quantum (or commuting operator) value of nonlocal games is exactly $1$. Analogously to the proof sketched for $\MIP^* = \RE$, we construct a self-compressing sequence of games $\UGS = (G_n)_{n \in \N}$ that depends on the given sentence $S = \forall x \, \phi(x)$. In pseudocode, the games have the following behavior:

\vspace{10pt}
\IncMargin{1em}
\begin{algorithm}[H]
\DontPrintSemicolon

The verifier checks whether $\phi(x)$ is false for some $x \leq n$. If it is, then reject.

Otherwise, compute $\UGS'$ by running $\GaplessCompress_t$ on the description of $\UGS$. 

Play the game $G_{n+1}'$, the $(n+1)$-st game of the sequence $\UGS'$.

\caption{The game $G_n$ encoding $\Pi_1$-sentences.}\label{fig:intro-mip-0-core-game}
\end{algorithm}\DecMargin{1em}
\vspace{10pt}

Again we assume that $\phi(n)$ is computable in $O(n)$ time, implying that the games in the sequence $\UGS$ have $\poly(n)$-complexity. The difference between this construction of $G_n$ and the previous one is that instead of checking whether $\phi(x)$ is true for some $x \leq n$, the verifier now checks whether it is \emph{false} for some $x$. Using the gapless compression theorem, we get that if $\phi(x)$ is true for all $x$ (meaning $S$ is true), then we have the following chain of inequalities:
\begin{align*}
	\val_t(G_n) &= \val_t(G_{n+1}') \\
				  &\geq (1 - \alpha) + \alpha \, \val_t(G_{n+1}) \\
				  &= (1 - \alpha) + \alpha \, \val_t(G_{n+2}') \\
				  &\geq (1 - \alpha) + \alpha((1 - \alpha) + \alpha \, \val_t(G_{n+2}')) \\
				  &\cdots
\end{align*}
The equalities follow from construction of the games, and the inequalities follow from the second item of \Cref{thm:gapless-comp-informal}. Since $\alpha < 1$, this implies that $\val_t(G_n)$ is at least $(1-\alpha)(\alpha + \alpha^2 + \alpha^3 + \cdots) = 1$, and thus $\val_t(G_n) = 1$. 

On the other hand, if $S$ is false, then there is some $n$ for which $\val_t(G_n) = 0$. Working backwards, we deduce that $\val_t(G_n') < 1$ (by the first item of the gapless compression theorem), so therefore $\val_t(G_{n-1}) < 1$, which means that $\val_t(G_{n-1}') < 1$, and so on. Thus for all $k \leq n$ we have $\val_t(G_k) < 1$. 

Finally, the game $G_S$ is then chosen to be the first member $G_1$ of the sequence $\UGS$. 

Since deciding the truth of $\Pi_1$-sentences is an undecidable problem, this gives an alternate proof of the undecidability of determining whether $\val_t(G) = 1$ for $t \in \{q,co\}$, first proved by Slofstra~\cite{slofstra_tsirelsons_problem_and_an_embedding_theorem,slofstra_set_of_quantum_correlations}. His proof is based on very different techniques based on group theory and approximate representation theory. As mentioned previously, the main result of Slofstra's work is that the set of quantum correlation $C_q$ is not closed. We can also prove this separation as a corollary of our results in \cref{subsec:application_pi1}.

\paragraph{Combining gapped and gapless compression. } 
The main application of our gapless compression theorem is to combine it with the gapped compression theorem of~\cite{ji_mip_re} to prove \Cref{thm:main-intro}, which establishes the $\Pi_2$-completeness of deciding whether the quantum value of a nonlocal game is equal to $1$. The two compression theorems, interleaved together, allow us to transform sentences $S$ of the form $\forall x \, \exists y \, \phi(x,y)$ (i.e. $\Pi_2$-sentences) to an equivalent nonlocal game $G_S$ (i.e. $S$ is true if and only if $\val_q(G_S) = 1$). 

Fix a $\Pi_2$-sentence $S = \forall x \, \exists y \, \phi(x,y)$. The key idea is that $S$ can be equivalently expressed as $S = \forall m \, S_m$ where $m$ ranges over the positive integers (rather than binary strings) and $S_m$ is the $\Sigma_1$-sentence $\exists n\, \phi(m,n)$, where $n$ also ranges over the positive integers. Leveraging the $\Sigma_1$-sentences-to-nonlocal games reduction from~\cite{ji_mip_re}, we get that for all $m \in \N$ there exists a nonlocal game $H_m$ (computable from $S_m$) such that $\val_q(H_m) = 1$ if and only if $S_m$ is true. In particular $S$ is true if and only if $\forall n \, \val_q(H_n) = 1$.

Now we design a sequence of games $\UGS = (G_n)_{n \in \N}$ encoding the sentence $S$ as follows. 

\vspace{10pt}
\IncMargin{1em}
\begin{algorithm}[H]
\DontPrintSemicolon

Using the reduction from~\cite{ji_mip_re}, compute the description of the game $H_n$ corresponding to the $\Sigma_1$-sentence $S_n$. 

Compute the game sequence $\UGS' = (G_n')_{n \in \N}$ by running $\GaplessCompress_q$ on the description of $\UGS$. 

With probability $\frac{1}{2}$, play the game $G_{n+1}'$, the $(n+1)$-st game of the sequence $\UGS'$.

Otherwise with probability $\frac{1}{2}$, play the game $H_n$

\caption{The game $G_n$ encoding $\Pi_2$-sentences.}
\end{algorithm}\DecMargin{1em}
\vspace{10pt}

Since the reduction of~\cite{ji_mip_re} is polynomial-time computable, the game $H_n$ has $\poly(n)$ complexity. The compressed game $G_{n+1}'$ has $O(\log n)$ complexity, due to the guarantees of the $\alg{GaplessCompress}_q$ procedure. This implies that each game $G_n$ in the sequence $\UGS$ has $\poly(n)$ complexity.
If $S$ is true (meaning that $S_m$ is true for all $m$) then we can establish the following relationship between $\val_q(G_n)$ and $\val_q(G_{n+1})$: 

\begin{align*}
	\val_q(G_n) &= \frac{1}{2}\val_q(G_{n+1}') + \frac{1}{2}\val_q(H_{n}) & \text{(Definition of the game $G_n$)}\\
	              &= \frac{1}{2}\val_q(G_{n+1}') + \frac{1}{2} & \text{($S$ true $\Rightarrow \val_q(H_n) = 1$ for all $n$)}\\
				  &\geq \frac{1}{2} \Big( (1 - \alpha) + \alpha \, \val_q(G_{n+1}) \Big) + \frac{1}{2} & \text{(\Cref{thm:gapless-comp-informal})} \\
				  &= 1 - \frac{\alpha}{2} \Big(1 - \val_q(G_{n+1})\Big)
\end{align*}
This is equivalent to $1 - \val_q(G_n) \leq \frac{\alpha}{2} \Big( 1 - \val_q(G_{n+1}) \Big)$ and by induction this means that $1 - \val_q(G_n) \leq \Big( \frac{\alpha}{2}\Big)^k \Big( 1 - \val_q(G_{n+k}) \Big)$ for all $k \in \N$. As $k$ goes to infinity, this means that $\val_q(G_n)$ is arbitrarily close to $1$, and thus equal to $1$. 

On the other hand, if $S$ is false, then there is some $n$ for which $S_n$ is false and consequently $\val_q(H_n) < 1$. This means $\val_q(G_{n}) < 1$. By the gapless compression theorem (\Cref{thm:gapless-comp-informal}) we deduce that $\val_q(G_n') < 1$, so therefore $\val_q(G_{n-1}) < 1$, which means that $\val_q(G_{n-1}') < 1$, and so on. Thus for all $k \leq n$ we have $\val_q(G_k) < 1$. 

Finally, the desired game $G_S$ is then chosen to be the first member $G_1$ of the sequence $\UGS$. 

We observe that for this argument it did not matter that reduction from $\Sigma_1$-sentences $S_n$ to games $H_n$ is gapped (in the sense that $\val_q(H_n) = 1$ if $S_n$ is true and $\val_q(H_n) \leq \frac{1}{2}$ otherwise). All that mattered was that there was \emph{some} reduction from $\Sigma_1$-sentences to nonlocal games such that the game value reflects the truth of the sentence. This raises an interesting question for whether it is possible to prove the $\Pi_2$-hardness result using ``just'' a gapless compression theorem.

\paragraph{Gapped compression for commuting operator strategies?} It is still unknown whether the problem of approximating the commuting operator value is as hard as deciding $\Pi_1$-sentences, which would mean that exact and approximate computation of the commuting operator value are equivalent in difficulty. Once again, the question boils down to the existence of a gapped compression procedure for commuting operator strategies. Suppose the following conjecture held:

\begin{conjecture}[Gap-preserving compression for commuting operator strategies]
\label{conj:gap-comp-co}
There exists a computable map $\GapCompress_{co}$ that, given a sequence of games $\UGS = (G_n)_{n \in \N}$, outputs a sequence of games $\UGS' = (G_n')_{n \in \N}$ such that the complexity of $\UGS'$ is $O(\log n)$, and furthermore if the complexity of $\UGS$ is at most $\poly(n)$, then for all $n \in \N$,
\begin{itemize}
	\item If $\val_{co}(G_n) = 1$, then $\val_{co}(G_n') = 1$. 
	\item If $\val_{co}(G_n) \leq \frac{1}{2}$, then $\val_{co}(G_n') \leq \frac{1}{2}$. 	
\end{itemize}
\end{conjecture}

We can then design a sequence of games $\UGS$ as follows. Let $M$ denote a Turing machine that, given a description of a nonlocal game $F$ (note that this is a single game, rather than a sequence of games), halts if $\val_{co}(F) < 1$ and otherwise runs forever. The semidefinite programming hierarchies of~\cite{navascues2008convergent,doherty2008quantum}, or the procedure described by~\cite{goldbring_computability_theoretic_reformulation_of_connes}, can be used to implement $M$. 

\vspace{10pt}
\IncMargin{1em}
\begin{algorithm}[H]
\DontPrintSemicolon

The verifier checks whether $\phi(x)$ is false for some $x \leq n$. If it is, then reject.

Compute the description of the nonlocal game $G_1$, the first game of the sequence $\UGS$. 

Run $M$ on input $G_1$ for $n$ steps. If it halts, then accept.

Otherwise, compute $\UGS'$ by running $\GaplessCompress_{co}$ on the description of $\UGS$.

Play the game $G_{n+1}'$, the $(n+1)$-st game of the sequence $\UGS'$.

\caption{The game $G_n$ to decide $\Pi_1$-sentences.}\label{fig:intro-mip-gap-core-game}
\end{algorithm}\DecMargin{1em}
\vspace{10pt}

The complexity of $\UGS$ is $\poly(n)$ so the consequences of \Cref{conj:gap-comp-co} hold. Let $S$ denote the sentence $\forall x\, \phi(x)$ for some $O(n)$-time computable predicate $\phi$. Suppose $S$ were true. Then Step 1 of \Cref{fig:intro-mip-gap-core-game} would never reject. Suppose that $\val_{co}(G_1) < 1$. Then by definition, $M$ will halt in some number of steps $T$. Thus $\val_{co}(G_n) = 1$ for all $n \geq T$. For $n < T$, we have that $\val_{co}(G_n) = 1$ if and only if $\val_{co}(G_{n+1}') = 1$ (by design of $G_n)$, which is if and only if $\val_{co}(G_{n+1}) = 1$ (by \Cref{conj:gap-comp-co}). By an inductive argument we get that $\val_{co}(G_1) = 1$, which contradicts our assumption. Thus we get $\val_{co}(G_1) = 1$.

On the other hand, suppose that $S$ were false. Let $m$ denote the least integer such that $\phi(m)$ is false. First, it cannot be the case that $M$ halts in fewer than $m$ steps. If it halted in $n$ steps for $n < m$, then $\val_{co}(G_n) = 1$ by construction. However, by construction and \Cref{conj:gap-comp-co} this means that $\val_{co}(G_{n-1}) = 1$, and so on, ultimately yielding that $\val_{co}(G_1) = 1$. This is a contradiction, as the fact that $M$ halts implies that $\val_{co}(G_1) < 1$. 

Next, we see that $\val_q(G_m) = 0$ because $\phi(m)$ is false. By \Cref{conj:gap-comp-co}, this means that $\val_{co}(G_{m-1}) \leq \frac{1}{2}$, and so on, ultimately yielding that $\val_{co}(G_1) \leq \frac{1}{2}$, as desired. Letting $G_S = G_1$, this completes the reduction from the problem of deciding $\Pi_1$-sentences to approximate $co$-value problem.

We discuss a plausible approach to proving \Cref{conj:gap-comp-co} in \Cref{sec:tutorial}. 

Finally, we note that there is something bizarre about the use of the Turing machine $M$ in this construction. Regardless of whether $S$ is true or false, in \emph{both cases}, the verifier in the game $G_n$ never witnesses the Turing machine $M$ halting! Thus, it may appear that $M$'s halt/non-halt behavior is irrelevant to the decision procedures of the games $\{G_n\}$. However, if we remove line 3 from \Cref{fig:intro-mip-gap-core-game}, then it is no longer clear how to reason about the value of the game $G_1$! In particular, when $S$ is true, there is no $n$ for which we can definitively identify the value of $G_n$, because we have an ``infinite recursion'' where $G_n$ is the same game as the compression of $G_{n+1}$, which in turn is the same game as the compression of $G_{n+2}$, and so on. Thus, inserting $M$ in the description of the games seems to force the sequence of games $\{G_n\}$ to ``examine its own (commuting operator) value,'' which in turn allows us -- mathematicians looking in from the outside -- to pin down the value of $G_n$ for all $n$. We find it a fascinating question of whether it is possible to deduce the value of the games $\{G_n\}$ with line 3 removed.\footnote{This trick of inserting the Turing machine $M$ into the description of the game is also used by~\cite{ji_mip_re} to construct an explicit game whose commuting operator value differs from its quantum value.}

\paragraph{Are compression theorems necessary?} 

We have just demonstrated that, equipped with the appropriate compression procedures, we can characterize the complexity of the quantum and commuting operator value of nonlocal games. Could compression theorems be \emph{necessary}? That is, does knowing that (say) exactly computing the commuting operator value is equivalent to deciding $\Pi_1$-sentences imply the existence of a compression procedure like the one given by \Cref{thm:gapless-comp-informal}? 

In~\cite{mousavi2020complexity}, it was shown that $\MIP^* = \RE$ (i.e. the $\Sigma_1$-hardness of the approximate $q$-value problem) implies a gap-preserving compression theorem for quantum strategies (i.e., \Cref{thm:gap-comp-informal}). We show that this equivalence between compression and complexity of nonlocal games is more general:
\begin{itemize}
	\item The $\Pi_1$-hardness of the approximate $co$-value problem implies a gap-preserving compression theorem for commuting operator strategies.
	
	\item The $\Pi_1$-hardness of the exact $co$-value problem implies a gapless compression theorem for commuting operator strategies.
	
	\item The $\Pi_2$-hardness of the exact $q$-value problem implies a gapless compression theorem for quantum strategies.

\end{itemize}
We prove these equivalences in \Cref{sec:equivalence}.

\paragraph{Relation to previous work}

The idea of using compression in order to obtain complexity lower bounds for nonlocal games was first due to Ji~\cite{ji_compression_of_quantum_multi_prover}. There, he showed that the complexity of deciding between $\val_q(G) = 1$ and $\val_q(G) \leq 1 - 1/\poly(|G|)$ where $|G|$ denotes the description length of the game $G$ is at least as hard as solving $\mathsf{NEXP}$-complete problems. His result, however, only applied to games with more than two players (in fact his result applies for games with $10$ players). The techniques used to compress games use a variety of tools from quantum information theory, including quantum error correcting codes and the Feynman-Kitaev history state construction. This compression technique was further developed by \cite{fitzsimons_quantum_proof_systems_for_iterated_exponentials}, who prove a gapless compression theorem that can be \emph{recursively composed} in order to obtain arbitrarily large complexity lower bounds for nonlocal games. The lower bounds obtained by~\cite{fitzsimons_quantum_proof_systems_for_iterated_exponentials} still only apply to games with three or more players, however. This is a fundamental limitation of the compression approach of~\cite{ji_compression_of_quantum_multi_prover,fitzsimons_quantum_proof_systems_for_iterated_exponentials} because they rely on using quantum error-correcting codes to perform \emph{secret sharing}, which require $3$ or more parties.   

Obtaining complexity lower bounds for \emph{two} player games have wider implications and require new techniques. For example, the connection between Connes' Embedding Problem and the approximate $q$-value problem only hold for two player games. Compressing two-player nonlocal games was first pioneered by~\cite{natarajan_neexp} and then further developed by~\cite{ji_mip_re} to prove $\MIP^* = \RE$. 
These works use very different tools such as classical and quantum low-degree tests and probabilistically checkable proofs (PCPs).\footnote{View Section $2$ of \cite{natarajan_neexp} for a more in-depth overview of the differences.}
The gapless compression theorem of this paper is based on a simplified version of these techniques, which allows us to obtain our $\Pi_2$-hardness result for two-player games.

In~\cite{mousavi2020complexity}, 
we obtained $\Pi_2$-hardness for the exact $q$-value problem for games with three or more players. This is because we combined the gapless compression theorem of~\cite{fitzsimons_quantum_proof_systems_for_iterated_exponentials} with the gapped compressed theorem of~\cite{ji_mip_re}. However as mentioned the requirement to have games with at least three players is intrinsic to the work of~\cite{ji_compression_of_quantum_multi_prover,fitzsimons_quantum_proof_systems_for_iterated_exponentials}. Furthermore, all previous works only study the setting of finite-dimensional (i.e.\ $q$-type) strategies; ours is the first to study compression of games in the commuting operator setting.

\subsection{Overview of the gapless compression theorem} \label{sec:tutorial}

We now provide an overview of the proof of \Cref{thm:gapless-comp-informal}, our gapless compression theorem. The compression theorem technically is about a procedure for transforming a sequence of games into another, but for simplicity we discuss compression as transforming individual games. 

The high-level structure of the compression procedure follows the paradigm first established by~\cite{natarajan_neexp} and developed further by~\cite{ji_mip_re}. Let $G$ denote an ``input'' game where the question lengths, answer lengths, and complexity of the decision procedure are $\poly(n)$. The game $G$ is transformed into a ``compressed'' game $G'$ where the complexity of the decision procedure is $\poly\log(n)$. This transformation consists of two steps, the first one called \emph{Question Reduction} and the second called \emph{Answer Reduction}. We describe these two steps next. 

Fix an input game $G = (\cal{X},\cal{A},D)$. All games involved use the uniform distribution over questions; for this reason we omit mention of the question distribution when specifying a nonlocal game. Fix a value type $t \in \{q,co\}$. 

\subsubsection{Question Reduction} 

The Question Reduction step transforms $G$ into the \emph{Introspection game} $G^\intro = (\cal{X}^\intro,\cal{A}^\intro,D^\intro)$ where 
\begin{gather*}
	\log |\cal{X}^\intro| = O(\log \log |\cal{X}|) \\
	\log |\cal{A}^\intro| = \poly(\log |\cal{A}|) \\
	\text{Complexity of $D^\intro$} = \poly(\text{Complexity of $D$})~.
\end{gather*}
The Introspection game $G^\intro$ is equivalent to $G$ in the sense that the value of $\val_t(G^\intro) = 1$ if and only if $\val_t(G) = 1$.

At an intuitive level, the question lengths are reduced in $G^\intro$ by asking the players to ``ask themselves'' -- i.e., to introspect -- their own questions from $\cal{X}$. The players in $G^\intro$ are each asked to sample a question $x\in \cal{X}$ and answer with $a\in \cal{A}$ as they would have answered in the original game $G$. If the players' responses are $(x,a)$ and $(y,b)$, the decision procedure in $G^\intro$ will check that $D(x,y,a,b) = 1$. 

In order for the values of $G$ and $G^\intro$ to be meaningfully related, we need to ensure that (a) the players sample their introspected questions $x$ and $y$ from the uniform distribution (instead of, say, always picking a fixed $(x^*,y^*)$ for which they have prepared winning answers), and (b) the first player does not have any knowledge of the second player's question $y$ and the second player does not have any knowledge of the first player's question $x$. 
 
Forcing players to behave honestly according to (a) and (b) crucially relies on a property called \emph{rigidity} that holds for some nonlocal games. A nonlocal game $G$ is rigid if the state and measurement operators of any near optimal strategy for $G$ satisfy very rigid constraints. For introspection, we need a family of games, called \emph{Question Sampling games} where the $n$th member of this family is denoted by $\qs_n$. Each game has two special questions labeled by measure-standard-basis and measure-orthogonal-basis and players in $\qs_n$ are required to respond to these questions with strings in $\{0,1\}^n$. Furthermore these games exhibit rigidity in the following sense; in any near optimal strategy for $\qs_n$ the players must share $n$ EPR pairs, and the player answering the measure-standard-basis (resp. measure-orthogonal-basis) question, must measure their share of entangled state using a measurement that is close, in some metric, to the standard basis measurement (resp. orthogonal basis $\{\ket{+},\ket{-}\}$ measurement). 

For simplicity suppose that the question set for the game $G$ is $\cal{X} = \{0,1\}^n$. Then the Introspection game $G^\intro$, at its core, is the $\qs_n$ game\footnote{To be more precise the game $G^\intro$ is $\qs_n$ extended so that it has a small number of additional special questions. The cross-checks between these special questions force the players to behave ``honestly'' (i.e., to sample $(x,y)$ from the uniform distribution), or risk losing the game with some nonzero probability.
}: to introspect the verifier just asks the player the measure-standard-basis question. The verifier then takes advantage of the other special question, measure-orthogonal-basis, to ensure that the properties (a) and (b) of introspection questions are satisfied. The proof of this fact is a direct consequence of the rigidity property of the Question Sampling game as described earlier. 

There are many candidate games for Question Sampling if we only cared about the rigidity property mentioned above. One example is the \emph{parallel-repeated Magic Square game} \cite{coudron_parallel_magic_square}. What makes the search for a family of games $\qs_n$ more challenging is the additional requirement imposed by the property
\[
	\log |\cal{X}^\intro| = O(\log \log |\cal{X}|).
\]
To satisfy this requirement the Question Sampling can have at most $\poly(n)$ questions. So overall $\qs_n$ must be a game with $\poly(n)$ questions for which any optimal strategy uses $n$ EPR pairs. Any family of games satisfying this property is said to be \emph{efficiently rigid}. Efficiency is referring to the fact that games with small number of questions are certifying Hilbert spaces of large dimension ($2^n$ in the case of $\qs_n$). The family of games where the $n$th game is the $n$th parallel-repeated Magic Square game is not efficiently rigid because the number of questions grows as $2^{O(n)}$. In \Cref{sec:2-of-n-ms} we introduce a family of games called $2$-out-of-$n$ Magic Square and prove it is efficiently rigid.

Introspection first appeared in \cite{natarajan_neexp} followed by a more sophisticated version in the $\MIP^* = \RE$ result. To obtain the gapped compression in that paper, the Question Reduction step must also be gap-preserving, i.e., in addition to the above requirements for introspection, it must be that if $\omega_q(G) < 1/2$, then $\omega_q(G^\intro) < 1/2$. For gapped introspection, in addition to efficient rigidity, we need to make sure that in any strategy winning $\qs_n$ with probability at least $1-\eps$, the measurement for measure-standard-basis question is $\poly(\eps,\log n)$-close (in operator norm) to the standard-basis measurement. The crucial point is that the error function has logarithmic dependence on $n$. This is what we call an \emph{efficiently robust rigidity} result. The $2$-out-of-$n$ Magic Square game is not highly robust because the error function has a polynomial dependence on $n$. The game used in the $\MIP^*=\RE$ result that exhibits this additional robustness requirement is called the \emph{quantum low-degree-test} \cite{natarajan2018low}. The proof of rigidity for this game is considerably more complicated than the proof of rigidity for the $2$-out-of-$n$ Magic Square game.
Also, in our setting we only need to introspect games with uniform question distributions. We believe these simplifications in the gapless setting help illuminate the core ideas behind introspection.

\subsubsection{Answer Reduction}

The Answer Reduction step transforms $G$ into the game $G^\ans = (\cal{X}^\ans,\cal{A}^\ans,D^\ans)$ where 
\begin{gather*}
	\log |\cal{X}^\ans| = \poly(\log |\cal{X}|) \\
	\log |\cal{A}^\ans| = O(1) \\
	\text{Complexity of $D^\ans$} = \poly(\log \text{Complexity of $D$})~.
\end{gather*}
The game $G^\ans$ is equivalent to $G$ in the sense that the value of $\val_t(G^\ans) = 1$ if and only if $\val_t(G) = 1$.

The idea is to delegate computing the decision procedure $D(x,y,a,b)$ to the players. Then have them certify their computation using a constant sized certificate. In this paper we use the \emph{Cook-Levin reduction}: this is an efficient transformation that maps a Turing machine $M$ and input string $w$ to a 3SAT formula $\varphi_M$ and variable assignment $\pi_w$ such that $M(w) = 1$ if and only if $\pi_w$ satisifes $\varphi_M$. Furthermore, $w$ is embedded in the beginning of $\pi_w$ . Clauses of the 3SAT formula $\varphi_M$ can be computed hyper-efficiently (which allows us to exponentially reduce the verifiers runtime). 
We use this to reduce the Turing machine $D_{x,y}$, that computes the decision procedure for fixed questions $(x,y)$, and the players answers $(a,b)$ to a 3SAT formula $\varphi_{x,y}$ and assignment $\pi_{a,b}$. The verifier will now compute a random clause of this formula, and ask the players to provide the assignments specified by $\pi_{a,b}$ to the variables in the clause.

There are three immediate issues we must address in this scheme. First, in our current game no individual player has access to both questions to produce the 3SAT formula $\varphi_{x,y}$. Secondly, if we allow one of the players to have access to both questions, in order to compute $\varphi_{x,y}$, we must ensure that the answers $(a,b)$ (and certificate $\pi_{a,b}$) are produced in such way that $a$ only depends on $x$ and $b$ only depends on $y$. Lastly, we have to make sure the player in fact returns the corresponding assignments specified by $\pi_{a,b}$ and does not change this depending on the clause we query.

Fortunately, all three issues can be addressed by \emph{oracularization}. This takes our original game and transforms it to a new game $G^\orac$ where the verifier sends one player a question $x \in \cal{X}$ and the other a pair of questions $(x,y) \in \cal{X}^2$. When a player receives a single question $x$ we call them an \emph{isolated player}. When a player receives a pair $(x,y)$ we call them an \emph{oracle player}. The players win if the oracle player responds with an answer pair $(a,b) \in \cal{A}^2$ such that $D(x,y,a,b) = 1$ and the isolated player responds with answer $a$ (resp. responds with answer $b$). Intuitively, in $G^\orac$ an oracle player must ``simulate'' the behavior of the two players in $G$, and the isolated player (who only receives half of the oracle question) is used to check that the oracle player's answers $(a,b)$ are produced in a way that $a$ only depends on $x$ and $b$ only depends on $y$, solving our first two issues.

Now we can go ahead and apply the Answer Reduction protocol on the game $G^\orac$, where the oracle player responds with assignments for our clause queries as described before, but the isolated player is asked a random bit of their original answer $a$ (resp. $b$). In particular we query only from those clauses which contain at least one variable from the beginning of $\pi_{a,b}$ which embeds $a$ (resp. $b$), we make sure the two players answers match on this assignment. This allows us to continue enforcing the no communication requirement after Answer Reduction. It also ensures that the oracle player is in fact providing assignments to the clause variables from $\pi_{a,b}$.
Therefore $G^\ans$ uses constant sized answers and has exponentially more efficient verifier complexity.  

\subsubsection{From gapless to gapped compression}

We highlight the primary differences between our gapless compression theorem and the gapped compression theorem of~\cite{ji_mip_re}.
\begin{itemize}
    \item In $\MIP^*=\RE$, instead of using the Cook-Levin reduction, the Answer Reduction transformation uses \emph{probabilistically checkable proofs} (PCPs) in order to control the amount of gap shrinkage. The soundness of the PCP construction in~\cite{ji_mip_re} is based on the soundness of something called the \emph{classical low-degree test} against entangled provers~\cite{ji2020quantum}, which is a very technically challenging part of their analysis.
    
\item As explained earlier, the Question Reduction step in $\MIP^*=\RE$ uses the robust rigidity of the quantum low-degree test~\cite{natarajan2018low}. Contrast this with our gapless compression theorem that does not require a robust rigidity test.
    
    \item The proof of $\MIP^*=\RE$ uses a \emph{parallel repetition theorem}. Roughly speaking, parallel repetition theorems state that if the quantum value of a game $G$ is less than $1$, then the value of the game $G^n$, that is obtained from $G$ by playing $n$ instances of $G$ in parallel, decays exponentially with $n$. This is needed because both the Question Reduction and Answer Reduction transformations shrink the gap by some amount, and parallel repetition is used to amplify the gap back to some constant amount.
\end{itemize} 

In this paper we transfer many of the ideas from \cite{ji_mip_re} to the infinite dimensional setting, allowing us to get a gapless compression theorem for commuting operator strategies. As discussed earlier proving \Cref{conj:gap-comp-co} requires a gapped compression theorem for the commuting operator strategies. Just like in the case of $q$-strategies, we would also need to establish commuting-operator analogues of the three ingredients described above: (1) soundness of the classical low-degree test, (2) soundness of the quantum low-degree test, and (3) a parallel repetition theorem.

The first item has been resolved in a forthcoming paper~\cite{ji2021quantum}. 
The second item requires a proof that the quantum low-degree test is sound against commuting operator strategies. Finally, parallel repetition is well studied in the context of (finite-dimensional) quantum strategies \cite{parallelreptition-Yao, parallelreptition-Vidick, bavarian2017hardness} but nothing is known yet in the context of commuting operator strategies (aside from the parallel repetition result of~\cite{cleve2008perfect}, but this only holds for XOR games).

Given the commuting-operator analogues of these tools, however, the $\Pi_1$-completeness of the approximate $co$-value problem should then follow from the argument described in \Cref{sec:compression-intro}.

\subsection{The synchronous strategies framework}
\label{sec:synchronous}
As mentioned, another goal of this paper is to present the proof of the gapless compression theorem (\Cref{thm:gapless-comp-informal}) in a way that distills, into their simplest form, the techniques and conceptual components that go into establishing its much more sophisticated cousin, the gap-preserving compression theorem of~\cite{ji_mip_re}. To that end, we express and prove all our results in the framework of \emph{synchronous strategies}, a class of strategies first studied by~\cite{paulsen2016estimating}. 
Working with these strategies simplifies our arguments both notationally as well as conceptually (as compared to working with general nonlocal games and general strategies).

A synchronous strategy $\strategy$ for a game $G$ 
is specified by a separable Hilbert space $\hilb$ (which could be infinite-dimensional), a von Neumann algebra $\algebra$ on $\hilb$, a tracial state on the algebra $\algebra$, \footnote{A \emph{von Neumann algebra} $\algebra$ on a Hilbert space $\hilb$ is a $*$-subalgebra of $\B(\hilb)$ (the set of bounded operators on $\hilb$) that contains the identity operator and is closed under the weak operator topology. A \emph{tracial state} $\tau$ on the algebra $\algebra$ is a positive, unital linear functional that satisfies the \emph{trace property}: $\trace{AB} = \trace{BA}$ for all $A,B \in \algebra$. 
} and a set of projective measurements $\{ M^x \}_{x \in \cal{X}}$ in the algebra $\algebra$ (each $M^x$ is a set of projections $\{M^x_a\}_{a \in \cal{A}}$ summing to the identity). Given questions $(x,y)$, the probability of obtaining answers $(a,b)$ is given by $\tau(M^x_a \, M^y_b)$. Thus the probability that the strategy $\strategy$ succeeds in the game $G$ is given by
\[
	\sum_{x,y \in \cal{X}} \mu(x,y) \, \sum_{a,b \in \cal{A}} D(x,y,a,b) \, \tau \Big( M^x_a \, M^y_b \Big)~.
\]

Readers who are not familiar with von Neumann algebras and tracial states may find the finite-dimensional setting easier to understand. When $\hilb= \C^r$ for some dimension $r$, then we can without loss of generality take the algebra $\algebra$ to be the set $\B(\hilb)$ of all bounded operators on $\hilb$ (which in finite dimensions is simply the set of all linear operators). In this case there is a \emph{unique} tracial state, which is the normalized trace $\tau(X) = \frac{1}{r} \tr(X)$. In terms of strategies for nonlocal games, this corresponds to the players using the same projective measurements for each question and sharing the maximally entangled state $\ket{\Phi} = \frac{1}{\sqrt{r}} \sum_{e = 1}^r \ket{e} \ket{e}$. Such a strategy has the property that if both players receive the same question $x \in \cal{X}$, they always output the same answer $a \in \cal{A}$ (this is why these strategies are called ``synchronous'').

In the infinite-dimensional setting, synchronous strategies give rise to \emph{commuting operator} strategies: for every synchronous strategy $\strategy = (\tau,\{M^x\})$ with Hilbert space $\hilb$, there exist another Hilbert space $\hilb'$, a state $\ket{\psi} \in \hilb'$, and measurements $\{A^x\}, \{B^x\}$ on $\hilb'$ for the players respectively such that for all $x,y \in \cal{X}$ and $a,b\in \cal{A}$, the operators $A^x_a$ and $B^y_b$ commute and we have
\[
	\tau(M^x_a \, M^y_b) = \bra{\psi} A^x_a \, B^y_b \ket{\psi}~.
\]
For a proof, see \cite[Theorem 5.5]{paulsen2016estimating}.

\begin{remark}
On the need to specify a von Neumann algebra $\algebra$ as part of the strategy: unlike in the finite-dimensional setting, we cannot without loss of generality take $\algebra$ to be all of $\B(\hilb)$; this is because there may not necessarily be a tracial state on $\B(\hilb)$.
\end{remark}

Synchronous strategies arise naturally when considering \emph{synchronous games}: these are games where the players must output the same answers whenever they receive the same question (i.e.\ $D(x,x,a,b) = 0$ whenever $a \neq b$). This simple restriction on the rules of the game has the following consequences for optimal strategies:

\begin{theorem}[Adapted from Theorem 3.2 of \cite{helton2017algebras} and Theorem 3.6 of \cite{kim2018synchronous}]
\label{thm:synchronous}
Let $G = (\cal{X},\cal{A},\mu,D)$ be a synchronous game such that $\mu(x,x) > 0$ for all $x \in \cal{X}$. Then if $\omega_{co}(G) = 1$ then there exists a \emph{synchronous strategy} $\strategy = (\tau, \{M^x\})$ for $G$ that achieves value $1$. If furthermore $\omega_q(G) = 1$, then there exists a sequence $\{ \strategy_n \}_{n \in \N}$ of \emph{finite-dimensional} synchronous strategies whose values approach $1$.
\end{theorem}

Many games studied in quantum information theory and theoretical computer science are synchronous games; for example the games constructed in the proof of $\MIP^* = \RE$ are all synchronous. In this paper, we also focus exclusively on synchronous games. For this reason, we focus on analyzing the \emph{synchronous value} of games: we define
\[
	\omega_{co}^s(G) := \sup_{\text{\tiny synchronous } \strategy} \, \omega(G,\strategy) \qquad \text{and} \qquad \omega_q^s(G) := \sup_{\substack{\text{\tiny finite-dimensional} \\ \text{\tiny synchronous }  \strategy}} \, \omega(G,\strategy)~.
\]
Since synchronous strategies correspond to commuting operator strategies, we have that $\omega_{co}^s(G) \leq \omega_{co}(G)$ and similarly $\omega_q^s(G) \leq \omega_q(G)$; \Cref{thm:synchronous} implies that $\omega_{t}^s(G) = 1$ if and only if $\omega_t(G) = 1$ for $t \in \{q,co\}$. Thus we do not lose any generality by restricting our attention to synchronous strategies. 

The benefits of working within the synchronous games framework is that strategies only require specifying one set of measurements for both players (instead of having to keep track of one for Alice and one for Bob), and furthermore the state $\tau$ has the cyclic trace property. Working in the synchronous setting significantly simplified many of our proofs, in particular those of rigidity and introspection. Previous rigidity results needed to characterize the shared state upto isometry and find a concrete representation of the measurement operators as matrices. In the synchronous setting however we are able to completely sidestepped these technical issues. We need only to show that certain algebraic relations such as commutation or anticommutation are satisfied by any optimal strategy, which allows for a much cleaner argument. Furthermore, working in the synchronous games framework allows for a unified treatment of both the finite- and infinite-dimensional settings.

This paper builds upon arguments and techniques from a number of previous results. There has been great success in pinning down the algebra of optimal strategies within the synchronous games setting.
It is our hope that expressing our results in the language of synchronous games will facilitate connecting our work to the world of functional analysis and operator algebras.

\vspace{20pt}

\paragraph{Acknowledgments.} We thank Vern Paulsen and William Slofstra for helpful comments. We thank the FOCS referees for their feedback. The research presented in this paper was initiated at the University of Toronto. H.M. acknowledges the support of the Natural Sciences and Engineering Research Council of Canada (NSERC). H.Y. is supported by an NSERC Discovery Grant, a Google Research Award, and AFOSR award FA9550-21-1-0040.

\newpage
\section{Preliminaries}
\label{section:preliminaries}

For an integer $d \in \N$ we write $[d]$ to denote $\{1,2,\ldots,d\}$. For functions $f,g_1,\ldots,g_k: \N^k \to \N$, we write $f \leq \poly(g_1,\ldots,g_k)$ if there exists a constants $C, E \geq 0$ such that for all sufficiently large $a_1,\ldots,a_k$,
\[
    f(a_1,\ldots,a_k) \leq C \prod_{i=1}^\ell g_i(a_1,\ldots,a_k)^E.
\]

Let $A(x_1,\ldots,x_k)$ denote a $k$-input Turing machine, which is a Turing machine with $k$ input tapes, a single work tape, and a single output tape. Then $\TIME_A(x_1,\ldots,x_k)$ denotes the maximum of the description length of $A$, and the running time of $A$ on input $(x_1,\ldots,x_k)$ (which may be $\infty$ if $A$ never halts on that input). For an integer $n \in \N$, we let $\TIME_A(n)$ denote the maximum of $\TIME_A(n,x_2,\ldots,x_k)$ over all $x_2,\ldots,x_k \in \{0,1\}^*$ (where $n$ is provided to $A$ in binary). 

\subsection{Algebras, states, and norms}
Let $\hilb$ be a separable Hilbert space and let $\B(\hilb)$ denote the set of bounded linear operators on $\hilb$. We write $\id_\hilb$ to denote the identity operator on $\hilb$ (and simply write $\id$ when the Hilbert space is clear from context). 

A von Neumann algebra on a Hilbert space $\hilb$ is a unital $*$-subalgebra of bounded operators $\B(\hilb)$ that is closed in the \emph{weak operator topology}. Given two von Neumann algebras $\algebra$ and $\algebraB$ on Hilbert spaces $\hilb_A, \hilb_B$ respectively, the tensor product algebra $\algebra \otimes \algebraB$ is defined to be the closure under the weak operator topology of the $*$-subalgebra generated by $\{ A \otimes B \in\B(\hilb_A \otimes \hilb_B) : A \in \algebra,B \in \algebraB \}$. 

Let $\algebra \subseteq \B(\hilb)$ denote a von Neumann algebra on $\hilb$. We say that a positive linear functional $\tau: \algebra \to \C$ is
\begin{itemize}
	\item \emph{Unital} if $\tau(\id) = 1$~;
	\item \emph{Normal} if for all families $(P_i)_{i \in I}$ of pairwise orthogonal projections in $\algebra$, we have $\tau \Big( \sum_{i \in I} P_i \Big) = \sum_{i \in I} \tau(P_i)$~;
	\item \emph{Tracial} if for all $A,B \in \algebra$, we have $\tau(AB) = \tau(BA)$~;
\end{itemize}
In this paper, $\tau$ will always represent a positive linear functional that is tracial, normal, and unital. We call such functionals a \emph{normal tracial state}. For brevity we often drop the ``normal'' qualifier. For an in-depth reference to von Neumann algebras, we refer the reader to Blackadar's textbook~\cite{blackadar2006operator}. 

We record some basic properties of tracial states. First, tracial states satisfy the Cauchy-Schwarz and H\"{o}lder inequalities, i.e.
\[
    |\tau(A^* B)|^2 \leq \tau(A^* A) \, \tau(B^* B) \qquad \text{and} \qquad |\tau(A^* B)| \leq \|A\| \cdot \tau(|B|)
\]
where $\| \cdot \|$ denotes the operator norm, and $|B| = \sqrt{B^* B}$. Second, tracial states give rise to a seminorm on $\algebra$: we define the \emph{$\tau$-norm} of an operator $A \in \algebra$ to be
\[
    \| A \|_\tau = \sqrt{\tau(A^* A)} = \sqrt{\tau(AA^*)}.
\]
The $\| \cdot \|_\tau$ norm satisfies the triangle inequality: i.e., $\| A + B \|_\tau \leq \|A\|_\tau + \|B \|_\tau$. 

If $\hilb$ is finite dimensional (i.e. isomorphic to $\C^d$) then there is a unique tracial state on the algebra $\B(\hilb)$, which is the \emph{dimension-normalized trace} $\frac{1}{d} \tr(A)$. Thus in this case the $\tau$-norm is the normalized Frobenius norm.

\begin{proposition}
\label{prop:tensor-product}
If $\tau$ and $\sigma$ are tracial states on von Neumann algebras $\algebra$ and $\algebraB$ respectively, then $\tau\otimes \sigma$ is a tracial state on the von Neumann algebra $\algebra \otimes \algebraB$.
\end{proposition}

\begin{proposition}
\label{prop:holder}
Let $A,B \in \algebra$. Then $\|AB \|_\tau \leq \|A\| \cdot \| B \|_\tau$.
\end{proposition}
\begin{proof}
	\begin{align*}
		\|AB \|_\tau &= \sqrt{ \tau(BB^*A^* A) } \\
					&\leq \sqrt{ \| A^* A \| \cdot \tau(BB^*) } & \text{(H\"{o}lder)} \\
					&= \| A \| \cdot \| B\|_\tau
	\end{align*}
\end{proof}

The following proposition allows us to exchange any operator $A$ in any expression $CAD$ with a nearby operator $B$ and obtain a new expression $CBD$ close to the original expression.

\begin{proposition}\label{prop:exchange-operators}
    Let $C,D \in \algebra$ be any operators with $\|C\|,\|D\| \leq 1$. If $A,B \in \algebra$ and $\|A-B\|_\tau \leq \epsilon$, then $\|CAD - CBD\|_\tau \leq \epsilon$ and $|\tau(CAD-CBD)|\leq \epsilon$. \end{proposition}
\begin{proof}
By Proposition \ref{prop:holder}
\[
    \|C(A-B)D\|_\tau^2 \leq  \|C\|^2 \|D\|^2 \|A-B\|_\tau^2 \leq \|A-B\|_\tau^2.
\]
We also have
\begin{align*}
    |\tau(C(A-B)D)|^2 &= |\tau(DC(A-B))|^2 \\
    &\leq \tau(D C C^* D^*) \tau((A-B)^*(A-B)) & \text{(Cauchy-Schwarz)}\\
    &\leq \|A-B\|_\tau^2.
\end{align*}
In the last line we used that $\tau(D C C^* D^*) \leq 1$. Indeed, if $\|M\|\leq 1$, then by H\"{o}lder $|\tau(M)|\leq \|M^*\|\tau(I) \leq 1$.
\end{proof}
In applications of Proposition \ref{prop:exchange-operators} we usually find ourselves in a situation where $C$ and $D$ are products of projections and unitaries. Since the operator norm is submultiplicative, i.e., $\|MN\| \leq \|M\| \|N\|$, the operator norm of any product of projections and unitaries is bounded above by $1$. Thus the assumptions of the proposition are readily verified.

\begin{proposition}\label{prop:unitary-close-to-identity}
Let $U$ be any unitary. If $|\tau(\id-U)| \leq \epsilon$, then $\|\id-U\|_\tau \leq \sqrt{2\epsilon}$
\end{proposition}
\begin{proof}
\begin{align*}
    \|\id-U\|_\tau^2 = \tau((\id-U)^\ast(\id - U)) = \tau(2\id - U - U^\ast) \leq 2|\tau(\id - U)|.
\end{align*}
\end{proof}

\subsection{Measurements and distance measures on them} 
\label{sec:measurements}
Let $\algebra$ denote a von Neumann algebra with a normal tracial state $\tau$. Let $M = \{M_a\}_{a \in \cal{A}}$ and $N = \{N_a\}_{a \in \cal{A}}$ denote sets of operators in $\algebra$, indexed by a finite set $\cal{A}$. Then we measure the distance between $M$ and $N$, denoted by $\| M - N \|_\tau$, as
\[
    \| M - N\|_\tau = \sqrt{\sum_{a \in \cal{A}} \left \| M_a - N_a \right \|^2_\tau }\;.
\]
We say that $M$ is \emph{$\delta$-far} from $N$, denoted by $M_a \approx_\delta N_a$, if $\|M - N \|_\tau \leq \delta$. We also occasionally use the notation $\|M\|_\tau = \sqrt{\sum_{a\in\cal{A}} \|M_a\|_\tau^2}$. \hynote{added:}

\begin{lemma}
\label{lem:triangle-inequality-tau-norm}
Let $M = \{ M_a \}_{a \in \cal{A}}$ and and $N = \{N_a\}_{a \in \cal{A}}$ denote sets of operators indexed by a finite set $\cal{A}$. Then 
\[
	\| M - N \|_\tau \leq \| M \|_\tau + \| N \|_\tau~.
\]
\end{lemma}
\begin{proof}
We compute:
	\begin{align*}
		\| M - N \|_\tau^2 &= \sum_{a \in \cal{A}} \left \| M_a - N_a \right \|^2_\tau  \\
		&\leq \Big( \sum_{a \in \cal{A}} \| M_a \|^2 \Big) + \Big( \sum_{a \in \cal{A}}  \| N_a\|^2 \Big) + 2 \Big( \sum_{a \in \cal{A}}  \| M_a \|_\tau \cdot \|N_a \|_\tau \Big) \\
		&\leq \Big( \sum_{a \in \cal{A}} \| M_a \|^2 \Big) + \Big( \sum_{a \in \cal{A}}  \| N_a\|^2 \Big) + 2 \sqrt{ \sum_{a \in \cal{A}}  \| M_a \|_\tau^2 } \cdot \sqrt{\sum_{a \in \cal{A}} \|N_a \|_\tau^2 } \\
		&= \Big ( \| M \|_\tau + \| N \|_\tau \Big)^2~.
	\end{align*}
	The first inequality follows from the triangle inequality of the $\tau$-norm, and the second inequality follows from Cauchy-Schwarz.
\end{proof}

A \emph{positive operator-valued measure (POVM) on $\hilb$ with outcomes in a finite set $\cal{A}$} is a set of positive operators $\{M_a\}_{a \in \cal{A}}$ such that $\sum_{a \in \cal{A}} M_a = \id$. A projective measurement is a POVM such that each element $M_a$ is a projection. For a projective measurement $M = \{M_a\}$ it holds that $M_a M_b = \delta_{a,b}M_a$ where $\delta_{a,b}$ is Kronecker delta. So operators belonging to the same projective measurement commute. We say two measurements $M = \{M_a\}$ and $N = \{N_b\}$ commute, if $M_a N_b = N_b M_a$ for all $a,b$. 

To denote ``data processed'' measurements, i.e., apply a function $f: \cal{A} \to \cal{B}$ to the outcome of a measurement, we use the following notation: $M_{[f]}$ denotes the POVM with elements
\[
    M_{[f \mid b]} = \sum_{a : f(a) = b} M_a
\]
for all $b \in \cal{B}$. As an example, suppose $\cal{A} = \{0,1\}^n$ and $\cal{B} = \{0,1\}$. Then we write $M_{[a \mapsto a_i]}$ to denote the processed measurement that measures a string $a$, and then returns the $i$-th bit of $a$. To refer to the element of $M_{[a \mapsto a_i]}$ corresponding to outcome $b \in \{0,1\}$, we write $M_{[a \mapsto a_i \mid b]}$. For a predicate $P: \cal{A} \to \{0,1\}$, we also use the notation
\[
	M_{[a : P(a)]} = \sum_{a : P(a) = 1} M_a ~.
\]
For example, the operator $M_{[a : f(a) \neq b]}$ denotes the sum over all $M_a$ such that $f(a) \neq b$. 

We introduce two important distance measures between POVMs that will be used throughout this paper. All operators referred to in the following are assumed to be elements of a von Neumann algebra $\algebra$ on which a tracial state $\tau$ is defined.

The first distance measure we define is called \emph{inconsistency}. Let $M,N$ denote POVMs with outcomes in a finite set $\cal{A}$ (called the \emph{answer set} or \emph{outcome set}). We say that $M$ and $N$ are \emph{$\delta$-inconsistent} if 
\[
 	\sum_{\substack{a,b \in \cal{A}: \\ a \neq b}} \tau(M_a \, N_b) \leq \delta
\]
When the answer set $\cal{A}$ is clear from context, we write $M_a \simeq_\delta N_a$ to denote that $M$ and $N$ are $\delta$-inconsistent.

The second distance measurement we introduce is called \emph{closeness}. We say that sets of POVMs  $M,N$ are \emph{$\delta$-far} if 
\[
  \| M - N \|_\tau \leq \delta.
\]
Similarly, when the answer set $\cal{A}$ is clear from context, we write $M_a \approx_\delta N_a$ to denote that $M$ and $N$ are $\delta$-far. Observe that this notion of closeness is also well-defined when the operators $M_a$, $N_a$ are not necessarily positive. Thus we will also write $M_a \approx_\delta N_a$ to denote closeness of arbitrary operator sets that are indexed by an answer set $\cal{A}$.

\subsection{Utility lemmas about measurements}

We now establish several utility lemmas concerning consistency, closeness, and measurements.

\begin{lemma}[Cauchy-Schwarz for operator sets]
	Let $M = \{M_a\}_{a \in \cal{A}}$ and $N = \{N_a\}_{a \in \cal{A}}$ denote sets of operators (not necessarily POVMs). Then 
	\[
		\Big | \sum_{a \in \cal{A}} \tau( M_a \cdot N_a) \Big |^2 \leq \Big( \sum_{a \in \cal{A}} \| M_a \|_\tau^2 \Big) \cdot  \Big( \sum_{a \in \cal{A}} \| N_a \|_\tau^2 \Big)  \;.
	\]
\end{lemma}
\begin{proof}
	\hynote{simplified proof:}
	For every $a \in \cal{A}$, we have that $|\tau(M_a \cdot N_a)| \leq \| M_a \|_\tau \cdot \| N_a \|_\tau$ by the Cauchy-Schwarz inequality for tracial states. Applying the triangle inequality and Cauchy-Schwarz again we have
	\[
		\Big | \sum_{a \in \cal{A}} \tau( M_a \cdot N_a) \Big |^2 \leq \Big( \sum_{a \in \cal{A}} \Big |  \tau( M_a \cdot N_a) \Big | \Big)^2 \leq \Big( \sum_{a \in \cal{A}} \| M_a \|_\tau \cdot \| N_a \|_\tau \Big)^2 \leq \Big( \sum_{a \in \cal{A}} \| M_a \|_\tau^2 \Big) \cdot  \Big( \sum_{a \in \cal{A}} \| N_a \|_\tau^2 \Big) ~.
	\]

\end{proof}

\begin{lemma}[Data processing inequality for consistency]
\label{lem:data-processing}
Let $M = \{M_a\}$ and $N = \{N_a\}$ be POVMs with outcomes in $\cal{A}$ such that $M_a \simeq_\delta N_a$. Let $f: \cal{A} \to \cal{B}$. Then 
\[
	M_{[f\mid b]} \simeq_\delta N_{[f\mid b]} \;.
\]
\end{lemma}
\begin{proof}
	\begin{align*}
		\sum_{b \neq b' \in \cal{B}} \tau (M_{[f \mid b]} N_{[f\mid b']}) = \sum_{\substack{b \neq b' \in \cal{B} \\ a,a' \in \cal{A} \\ f(a) = b,  f(a') = b'}} \tau(M_a N_{a'}) \leq \sum_{a \neq a' \in \cal{A}} \tau(M_a N_{a'}) \leq \delta.
	\end{align*}
\end{proof}

\begin{lemma}[Consistency to closeness]
\label{lem:consistency-consequences}
Let $M = \{M_a\}$ and $N = \{N_a\}$ be POVMs with outcomes in $\cal{A}$ such that $M_a \simeq_\delta N_a$. Then $M_a \approx_{\sqrt{2\delta}} N_a$.
\end{lemma}
\begin{proof}
	\begin{align*}
	     \sqrt{ \sum_a \|M_a - N_a\|_\tau^2} &=
		 \sqrt{ \sum_a \tau ((M_a - N_a)^2)} \\
		&\leq \sqrt{ \sum_a \tau (M_a + N_a - M_a N_a)} \\
		&=  \sqrt{ 2 - 2 \sum_a \tau(M_a N_a)} \\
		&\leq \sqrt{2 \sum_a \tau(M_a (\id - N_a))} \\
		&\leq \sqrt{2 \delta}.
	\end{align*}
	The first inequality follows because $M_a - M_a^2 \geq 0$ as $\{M_a\}$ are POVMs. The second inequality follows from Jensen's inequality.
\end{proof}

\begin{lemma}[Closeness to consistency]
\label{lem:closeness-to-consistency}
Let $M = \{M_a\}$ be a projective POVM and let $N = \{N_a\}_{a \in \cal{A}}$ be a POVM with outcomes in $\cal{A}$. Suppose that $M_a \approx_\delta N_a$. Then $M_a \simeq_{\delta} N_a$. 
\end{lemma}
\begin{proof}
Applying Cauchy-Schwarz twice, we get
	\begin{align*}
		 \sum_a \tau( M_a (\id - N_a)) &=  \sum_a \tau( M_a (M_a - N_a))\\
&\leq  \sqrt{\sum_a \tau(M_a^2)} \cdot \sqrt{\sum_a \tau( (M_a - N_a)(M_a - N_a)^*) }\\ 
		&\leq \delta 
	\end{align*}
	where we used that $\sum_a \tau(M_a^2) = 1$. 
\end{proof}

\begin{lemma}[Consistency implies similar probabilities]
\label{lem:consistency-to-prob-closeness}
Let $M = \{M_a\}$ and $N = \{N_a\}$ be POVMs with outcomes indexed by $\cal{A}$. Suppose that $M_a \simeq_\delta N_a$. Then
\[
	 \sum_{a \in \cal{A}} | \tau(M_a - N_a) | \leq 2\delta.
\]
\end{lemma}
\begin{proof}
Let $S_x = \{ a : \tau(M_a) > \tau(N_a) \}$ and $T_x = \{ a : \tau(N_a) \geq \tau(M_a) \}$. Then
\begin{align*}
 \sum_{a \in \cal{A}} |\tau(M_a - N_a)| =  \sum_{a \in S_x} \tau(M_a - N_a) +  \sum_{b \in T_x} \tau(N_a - M_a).
\end{align*}
Then, since $\tau(M_a N_a) \leq \tau(N_a)$, we have
\[
 \sum_{a \in S_x} \tau(M_a - N_a) \leq  \sum_{a \in S_x} \tau(M_a (\id - N_a)) \leq  \sum_{a \in \cal{A}} \tau(M_a (\id - N_a)) \leq \delta.
\]
Similarly $ \sum_{b \in T_x} \tau(N_a - M_a) \leq \delta$. This completes the proof.
\end{proof}

\begin{lemma}
\label{lem:add-a-measurement}
Let $M = \{M_a\}_{a \in \cal{A}},N = \{N_a\}_{a \in \cal{A}}$ be sets of operators (not necessarily POVMs), and let $R = \{R_b\}_{b \in \cal{B}}$ be a set of operators such that $\sum_b R_b^* R_b \leq \id$. Suppose that $M_a \approx_\delta N_a$. Then $R_b M_a \approx_\delta R_b N_a$  where the answer summation is over $(a,b) \in \cal{A} \times \cal{B}$. Similarly, if $\sum_b R_b R_b^* \leq \id$, we have $M_a R_b  \approx_\delta N_a R_b$. 
\end{lemma}
\begin{proof}
We prove the approximation $R_b M_a \approx_\delta R_b N_a$:
	\begin{align*}
		\sum_{a \in \cal{A}, b \in \cal{B}} \| R_b (M_a - N_a) \|_\tau^2 &= \sum_{a \in \cal{A}, b \in \cal{B}} \tau \Big( (M_a - N_a)^* R_b^* R_b (M_a - N_a) \Big) \\
		&= \sum_{a} \tau \Big( (M_a - N_a)^* \Big( \sum_b R_b^* R_b \Big) (M_a - N_a) \Big) \\
		&\leq \sum_{a} \tau \Big( (M_a - N_a)^* (M_a - N_a) \Big) \\
		&= \sum_a \| M_a - N_a \|_\tau^2 \\
		&\leq \delta^2.
	\end{align*}
	where in the first inequality we used the assumption that $\sum_b R_b^* R_b \leq \id$. The proof for the approximation $M_a R_b  \approx_\delta N_a R_b$ is similar.
\end{proof}

The following lemma states that POVMs that are almost projective (in the sense that each POVM element is close to its square) is close to a projective maesurement. A version of this was first proved in the finite-dimensional setting by~\cite{kempe2011parallel}, improved quantitatively in~\cite{ji2020quantum}, and recently extended to the setting of von Neumann algebras by de la Salle~\cite{delasalle-orthogonalization}.

\begin{lemma}[Projectivization of POVMs~\cite{delasalle-orthogonalization}]
\label{lem:projectivization}
Let $\{M_a\} \subset \algebra$ be a POVM with outcomes indexed by a finite set $\cal{A}$. Suppose that the following holds:
\[
    \sum_a \tau(M_a - M_a^2) \leq \eps.
\]
Then there exists a projective measurement $\{P_a\} \subset \algebra$ such that 
\[
    P_a \approx_{\delta_{proj}}  M_a
\]
where $\delta_{proj} = \delta_{proj}(\eps)$ is a function that depends on $\eps$ (but independent of $\cal{A}$) and goes to zero as $\eps \to 0$. 
\end{lemma}

The next lemma allows us to ``paste'' multiple approximately-commuting measurements together to form a joint projective measurement. 

\begin{restatable}[Pasting lemma]{lemma}{pasting}
\label{lem:pasting}
Let $\{M^{(1)}, M^{(2)}, \ldots, M^{(K)} \} \subset \algebra$ be a set of projective measurements with outcomes in a finite set $\cal{A}$. Suppose that for all $i \neq j$, we have that
\[
	M^{(i)}_a M^{(j)}_b \approx_\eps M^{(j)}_b M^{(i)}_a
\]
where the answer summation is over $(a,b) \in \cal{A}^2$. Then there exists a projective measurement $R = \{ R_{\vec{a}} \} \subset \algebra $ with outcomes in $\cal{A}^K$ such that for all $i \in [K]$,
\[
	R_{[\vec{a} \mapsto a_i\mid b]} \approx_{\delta_{pasting}} M^{(i)}_{b}
\]
where $\delta_{pasting} = \delta_{pasting}(K,\eps)$ is a function that goes to $0$ as $\eps \to 0$.
\end{restatable}
We prove \Cref{lem:pasting} in \Cref{sec:pasting}.

\subsection{Nonlocal games, strategies, and verifiers}
\label{sec:nonlocal-games}

\paragraph{Nonlocal games.} A \emph{nonlocal game} $G$ is a tuple $(\cal{X},\cal{A},\mu,D)$ where $\cal{X}$ is a finite \emph{question set}, $\cal{A}$ is a finite \emph{answer set}, $\mu$ is a probability distribution over $\cal{X} \times \cal{X}$, and $D: \cal{X} \times \cal{X} \times \cal{A} \times \cal{A} \to \{0,1\}$ is a function called the \emph{decision predicate}. 
A game $G$ is \emph{synchronous} if for all $x \in \cal{X}$, if $D(x,x,a,b) = 1$ if and only if $a = b$. We call a question pair $(x,y) \in \cal{X} \times \cal{X}$ \emph{trivial} if $D(x,y,a,b) = 1$ for all $(a,b) \in \cal{A} \times \cal{A}$; otherwise we call $(x,y)$ \emph{nontrivial}. 
 
 In this paper, we only consider games that are synchronous and whose question distribution is uniform over the question set; thus we denote games $G$ by tuples $(\cal{X},\cal{A},D)$.

\paragraph{Strategies.} A \emph{tracial strategy} $\strategy$ for a game $G=(\cal{X},\cal{A},\mu,D)$ is a pair $(\tau, \{M^x\}_{x \in \cal{X}})$ where there is a separable Hilbert space $\hilb$ such that $\{M^x\}$ is a set of POVMs on $\hilb$ with outcomes in $\cal{A}$, and $\tau$ is a normal tracial state on a von Neumann algebra $\algebra$ containing the set $\{M_a^x\}_{x,a}$. The \emph{value} of a tracial strategy $\strategy$ in $G$ is defined as
\[
    \val(G,\strategy) = \sum_{x,y \in \cal{X}} \mu(x,y) \, \sum_{a,b \in \cal{A}} D(x,y,a,b) \, \tau (M_a^x \, M^y_b)
\]
A tracial strategy $\strategy$ is called \emph{synchronous} if $\{M^x\}$ are projective measurements. A tracial strategy $\strategy$ is \emph{finite dimensional} if $\hilb = \C^d$ for some $d$. A tracial strategy $\strategy$ \emph{commutes on a set $C \subseteq \cal{X} \times \cal{X}$} if for all $(x,y) \in C$ measurements $M^x$ and $M^y$ commute, i.e., $M_a^x M_b^y = M_b^y M_a^x$ for all $a,b \in \cal{A}$.

The \emph{synchronous commuting operator value} of a synchronous game $G$, denoted by $\valco^s(G)$, is defined as the supremum of $\val(G,\strategy)$ over all synchronous strategies $\strategy$ for $G$. The \emph{synchronous quantum value} of $G$, denoted by $\val_q^s(G)$, is defined as the same thing except the supremum is restricted to finite-dimensional synchronous strategies.

The \emph{entanglement requirement} $\mathcal{E} \Big (G, \alpha \Big )$ for a game $G$ and $\alpha \in [0,1]$ is the minimum dimension of any finite-dimensional synchronous strategy $\strategy$ for $G$ with quantum value at least $\alpha$. If no such strategy exists then $\mathcal{E} \Big (G, \alpha \Big ) = \infty$.

We introduce the notion of an oracularizable strategy; the significance of this notion is that the answer reduction transformation (discussed in \Cref{sec:answer-reduction}) requires games to have oracularizable strategies. ``Oracularizability'' is an invariant maintained by our compression procedure (as well as the compression procedures of~\cite{natarajan_neexp,ji_mip_re}).
\begin{definition}[Oracularizable strategy]
\label{def:oracularizable-strategy}
A synchronous strategy $\strategy$ for a synchronous game $G$ is \emph{oracularizable} if the strategy commutes on the set of nontrivial questions of $G$.
\end{definition}

\paragraph{Verifiers.}

We introduce the notion of a \emph{verifier}, which gives a uniform way to describe infinite sequences of nonlocal games.

\begin{definition}[Verifiers]\label{def:verifiers}
Let $\UGS = (G_n)_{n \in \N}$ denote an infinite sequence of synchronous games where $G_n = (\cal{X}_n,\cal{A}_n,D_n)$ and the sets $\cal{X}_n = \{0,1\}^{\ell_n},\cal{A}_n \subset \{0,1\}^*$ for some polynomial-time computable function $\ell_n$ of $n$. A \emph{verifier $\verifier$ for $\UGS$} is a pair $(D,C)$ of Turing machines where $D$ is a $5$-input Turing machine and $C$ is a $3$-input Turing machine, such that for all $n \in \N$, the following hold:
\begin{enumerate}
\item $D(n,x,y,a,b) = D_n(x,y,a,b)$ for all $(x,y) \in \cal{X}_n \times \cal{X}_n$ and $(a,b) \in \cal{A}_n \times \cal{A}_n$, and
	\item $C(n,x,y) = 1$ if and only if $(x,y) \in \cal{X}_n \times \cal{X}_n$ is a nontrivial question pair for $G_n$.
\end{enumerate}
The Turing machines $C$ and $D$ are respectively called a \emph{question checker} (or simply just a \emph{checker}) and \emph{decider} for $\UGS$. When $n$ is written on the first input tape of $D$ and $C$, the Turing machines discard any string that comes after the $\ell_n$'th bit in the second and third input tapes. 
\end{definition}

Verifiers play a crucial role in the compression theorems of this paper and~\cite{ji_mip_re}, as they allow for an effective method (``effective'' in the computability sense) for encoding infinite sequences of nonlocal games.

\begin{remark} Although we have defined the games in the sequence $\UGS$ corresponding to a verifier $\verifier$ to have questions and answers consisting of binary strings, we often treat the questions and answers as sets with more structure, such as tuples. There, we implicitly assume an efficiently computable representation of set elements as binary strings is fixed.
\end{remark}

We note that the Turing machine $D$ in the definition of verifier $\verifier$ for an infinite sequence $\UGS = (G_n)_{n \in \N}$ of games already implicitly specifies the set of nontrivial questions for each $G_n$. For our compression procedure, however, it will be necessary to be able to quickly compute whether a question pair is nontrivial, and having a separate Turing machine $C$ for this is helpful for separately keeping track of the decision procedure complexity versus the complexity of deciding the set of nontrivial questions. 

\subsection{Asymptotics and approximation bounds} 
\label{subsec:asymptotics}
We end the preliminaries section with a short discussion of asymptotics in the analyses of the Rigidity, Question Reduction and Answer Reduction sections. The bounds and approximations in this paper are functions of two quantities: one is the \emph{game index} $n$, which indicates the $n$-th element of an infinite sequence $\UGS = (G_n)_{n \in \N}$ of games; we take $n$ to go to infinity and use $n$ to measure sizes of question/answer alphabets, as well as the time complexity of the deciders. The other quantity is $\eps$ where $1 - \eps$ is a lower bound on the synchronous quantum or synchronous commuting operator value of a nonlocal game $G$ under consideration. We treat $\eps$ as a quantity that goes to $0$. 

All of our approximations in this paper will generally depend on both $n$ and $\eps$. From the assumption that the value of the game is at least $1 - \eps$ we will derive consequences for a pair of measurements $\{M_a\},\{N_a\}$. For example we may prove that $M_a \approx_{\delta(n,\eps)} N_a$ where $\delta: \N\times \R^+ \to \R^+$ is any function that is continuous in the second argument and is such that $\delta(n,0) = 0$ for all $n$. We call such functions \emph{proper error functions}. We usually let the dependence on $n$ to be implicit and simply write $\delta(\epsilon)$ for proper error functions.

Every instance of $\delta$ in this paper should be understood as a function that is different from all the previous instances of $\delta$ except for the aforementioned two properties. For example if $M_a \approx_{\delta(\eps)} N_a$ and $N_a \approx_{\delta(\eps)} P_a$ by the triangle inequality we have
\[\sum_a \|M_a - P_a\|^2 \leq 2\sum_a \|M_a - N_a\|^2 + 2\sum_a\|N_a-P_a\|^2 \]
so we can write $M_a \approx_{\delta(\eps)} P_a$; every occurrence of $\delta(\eps)$ in these three approximations can be a different proper error function.

As such in this paper we usually do not keep track of the specific approximation bounds. For POVMs $\{M_a\}$ and $\{N_a\}$ we will often write $M_a \approx N_a$ to denote $M_a \approx_{\delta(\eps)} N_a$ for some proper error function $\delta(\eps)$. 
We also use the notation $M \approx N$, for any two operators $M,N$, to indicate that $\|M-N\|_\tau \to 0$ as $\eps \to 0$. Similarly we may write $\tau(M) \approx \tau(N)$ to indicate that $\tau(M-N) \to 0$ as $\eps \to 0$. We recommend reading the proof of \Cref{thm:rigidty-of-magic-square} carefully to get used to these conventions. The proof contains techniques that are used over and over in this paper.

\paragraph{Averaging argument.}
A simple but prevailing idea in many of the proofs in this paper is the observation that, if a strategy in a game $G$ has a value at least $1-\eps$, then the winning probability conditioned on any event that has a nonzero probability is at least $1-\delta(\eps)$ for some error function $\delta$ that has some dependence on the probability of the conditioning event (we usually ignore this dependence). So for example since the probability distribution on questions is uniform in all our games, the event that players receive a fixed question pair $(x,y)$ has probability $1/|\cal{X}|^2$ where $\cal{X}$ is the question set of the game. Then the probability of winning conditioned on players receiving question pair $(x,y)$ is at least $1 - |\cal{X}|^2\eps = 1 - \delta(\eps)$. We usually abbreviate this by simply saying ``by an averaging argument, the probability of winning conditioned on players receiving question pair $(x,y)$ is $1-\delta(\eps)$.'' Since we are working in the gapless regime, we do not need to keep track of the dependence of $\delta$ on $|\cal{X}|$ which allows us to just simply write $\delta(\eps)$.

\paragraph{The implication of cross-checks between nontrivial question pairs.}
We explain another proof technique that appears repeatedly in the following sections of the paper. Suppose $\{q,r,qr\} \in \cal{X}$ are three questions in a game $G$ ($qr$ is a single question different from $q$ and $r$). The answer to questions $q,r,qr$ are expected to be in three sets $\cal{A},\cal{B},\cal{A}\times \cal{B}$, respectively. Furthermore suppose that the winning condition dictates that $D(q,qr,a,(a',b')) = 1$ iff $a = a'$ and that $D(r,qr,b,(a',b')) = 1$ iff $b = b'$. Clearly $(q,qr)$ and $(r,qr)$ are nontrivial question pairs in this game.

Now one very useful observation is that if $(\tau,\{N^x\}_{x \in \cal{X}})$ is any strategy that wins this game with probability at least $1- \epsilon$, then it must be that \[N^q_a N^r_b \approx_{\delta(\eps)} N^r_b N^q_a,\] or in other words the measurements $N^q$ and $N^r$ approximately commute. To see this, first note that by an averaging argument the probability of winning conditioned on receiving question pair $(q,qr)$ is $1-\delta(\eps)$. This fact can be stated as follows
\begin{align*}
    1 - \delta(\eps) \leq \sum_{a\in \cal{A},b\in \cal{B}} \tau(N^q_a N^{qr}_{a,b}) = \sum_{a\in \cal{A}} \tau(N^q_a N^{qr}_{a,\cdot})
\end{align*}
where $N^{qr}_{a,\cdot}$ is the marginal measurement projection $\sum_{b\in \cal{B}} N^{qr}_{a,b}$. We can rewrite this as
\[N^q_a \simeq_{\delta(\eps)} N^{qr}_{a,\cdot}~.\]
By an application of \Cref{lem:consistency-consequences} we get
\[N^q_a \approx_{\delta(\eps)} N^{qr}_{a,\cdot}~.\]
By the symmetry we similarly get
\[N^r_b \approx_{\delta(\eps)} N^{qr}_{\cdot,b}~.\]
where $N^{qr}_{\cdot,b}$ is the marginal measurement projection $\sum_{a\in \cal{A}} N^{qr}_{a,b}$.

Using \Cref{prop:exchange-operators}, we get \[N^q_a N^r_b \approx_{\delta(\eps)} N^{qr}_{a,\cdot} N^{r}_b~.\] With another application of \Cref{prop:exchange-operators}, we get \[N^{qr}_{a,\cdot} N^{r}_b \approx_{\delta(\eps)} N^{qr}_{a,\cdot} N^{qr}_{\cdot,b}~.\] By the triangle inequality we can combine these to get 
\[N^q_a N^r_b \approx_{\delta(\eps)} N^{qr}_{a,\cdot} N^{qr}_{\cdot,b}~.\]
Since projection operators belonging to the same projective measurement commute, we have
\[N^{qr}_{a,\cdot} N^{qr}_{\cdot,b} = N^{qr}_{\cdot,b} N^{qr}_{a,\cdot}~.\]
Finally by two more applications of \Cref{prop:exchange-operators} and the triangle inequality, we get the desired result
\[N^q_a N^r_b \approx_{\delta(\eps)} N^r_b N^q_a~.\]
 \newpage
\section{Nonlocal game rigidity}
\label{sec:rigidity}

\newcommand{\equations}{\mathrm{eqs}}
\newcommand{\variables}{\mathrm{vars}}
\newcommand{\ms}{{\textsc{MS}}}

A fundamental component of compression theorems are the use of nonlocal games with specific \emph{rigidity} properties. Informally speaking, a nonlocal game $G$ is rigid if the state and measurement operators of an optimal strategy for $G$ must satisfy very rigid constraints -- even to the point of being uniquely specified up to conjugation by isometries. 

The most well-known example of a rigid game is the CHSH game~\cite{clauser1969proposed}, named after physicists Clauser, Horne, Shimony and Holt. In this game Alice and Bob receive questions $x,y \in \{0,1\}$ and answer with bits $a,b \in \{0,1\}$. They win if and only if $a + b = xy \mod 2$. 

It is well-known that the CHSH game satisfies $\val_q(CHSH) = \val_{co}(CHSH) = \frac{1}{2} + \frac{1}{2\sqrt{2}}$, and the optimum is achieved by a simple two-dimensional strategy (that we call the \emph{canonical strategy}) where the players share the entangled state $\ket{\epr} = (\ket{0} \otimes \ket{0} + \ket{1} \otimes \ket{1})/\sqrt{2}$, and Alice and Bob's measurement operators are defined to be the following: for all $a,b \in \{0,1\}$,
\begin{enumerate}
	\item $A^0_a$ is the projection onto the eigenspace of $Z = \begin{pmatrix} 1 & 0 \\ 0 & -1 \end{pmatrix}$ with eigenvalue $(-1)^a$.
	\item $A^1_a$ is the projection onto the eigenspace of $X = \begin{pmatrix} 0 & 1 \\ 1 & 0 \end{pmatrix}$ with eigenvalue $(-1)^a$.
	\item $B^0_b$ is the projection onto the eigenspace of $(Z + X)/\sqrt{2}$ with eigenvalue $(-1)^b$.
	\item $B^1_b$ is the projection onto the eigenspace of $(Z - X)/\sqrt{2}$ with eigenvalue $(-1)^b$.
\end{enumerate}
(The CHSH game is not a synchronous game and optimal strategies for CHSH are not synchronous, so in general Alice and Bob will have different measurement operators for each question).

It turns out that \emph{any} finite-dimensional strategy achieving the optimum value for CHSH must be \emph{equivalent} to the canonical strategy just described: if the state $\ket{\psi}$ belongs to $\hilb_A \otimes \hilb_B$ for finite-dimensional Hilbert spaces $\hilb_A, \hilb_B$,\footnote{A standard result in the theory of nonlocal games is that any finite-dimensional strategy can be expressed as a tensor-product strategy~\cite[Theorem 1]{scholz2008tsirelson}.} then there exist isometries $V_A, V_B$ acting on $\hilb_A,\hilb_B$ respectively such that $(V_A \otimes V_B) \ket{\psi} = \ket{EPR} \otimes \ket{\phi}$ for some auxiliary state $\ket{\phi}$, and furthermore under the isometries the players' measurement operators are equal to the canonical measurements described above. Since we can only characterize quantum strategies up to local isometries (i.e. applying local isometries to a strategy cannot change its success probability), this shows that the canonical strategy is essentially the unique strategy achieving the optimum winning probability for CHSH.

Furthermore, the rigidity of the CHSH game is \emph{robust}: strategies that are approximately optimal for CHSH must be approximately equivalent, up to local isometries, to the canonical strategy. The rigidity of the CHSH game has been studied extensively in quantum information theory and has found applications to quantum cryptography and quantum complexity theory; see~\cite{vsupic2020self} for a survey of self-testing and its applications.

In this paper, we propose a more abstract formulation of nonlocal game rigidity: we say that a game $G$ is rigid if there is a set of \emph{algebraic relations} that are (approximately) satisfied by the measurement operators in any strategy $\strategy$ for $G$ that (approximately) attains the optimal value. We no longer worry about characterizing the state vector or finding a concrete representation of the measurement operators as matrices.

For example, the rigidity of the CHSH game can be formulated as follows: any quantum strategy where their shared state is $\ket{\psi}$ and Alice's and Bob's projective measurements are $\{A^x_a\}$ and $\{B^y_b\}$ respectively that achieves value $\val_{co}(CHSH)$ in the CHSH game must generate \emph{anti-commuting observables}: defining the self-adjoint unitary operators $U^0 = A^0_0 - A^0_1$ and $U^1 = A^1_0 - A^1_1$, we must have that $U^0 U^1 \ket{\psi} = -U^1 U^0 \ket{\psi}$; the same holds with Bob's operators. Furthermore, this anti-commutation relation establishes that the Hilbert space must have dimension at least $2$. 

Establishing anti-commutation relations between the observables induced by an optimal strategy is usually the first step in ``traditional'' proofs of CHSH rigidity; this step is key to proving that the state and measurements are isometric to $\ket{EPR}$ and the Pauli $Z$ and $X$ observables, respectively. In this paper, however, we solely focus on the algebraic relations between the measurement operators -- these are the only properties that are needed for our applications. This allows us to shortcut some of the complexity of typical arguments for nonlocal game rigidity.

Aside from providing simplifications, we believe that this algebraic perspective on rigidity will be beneficial for studying nonlocal games and their connections to subjects such as approximate representation theory and operator algebras.

\subsection{The Magic Square game} 

We illustrate how rigidity results can be formulated in the synchronous games framework using the \emph{Mermin-Peres Magic Square game} (often called \emph{Magic Square game} for short)~\cite{mermin1990simple,peres1990incompatible,aravind2002simple}. Rigidity of Magic Square is first proved in \cite{Wu_2016}. The Magic Square is a game where the players' goal is to convince the verifier that they can assign values to the cells of a $3 \times 3$ grid such that the sum of cells within a row or column is even, except in the last column, where the sum should be odd. Of course, it is impossible to deterministically assign values satisfying these constraints, but when the players use a quantum strategy it appears as if they are performing the impossible.

We can view the Magic Square game as corresponding to a system of linear equations over $\mathbb{Z}_2$: let $s_{11},\ldots,s_{33}$ denote variables for the nine squares of the $3 \times 3$ grid, as depicted below:
\begin{center}
\begin{tabular}{ |c|c|c| } 
 \hline
 $s_{11}$ & $s_{12}$ & $s_{13}$ \\ 
 \hline
 $s_{21}$ & $s_{22}$ & $s_{23}$ \\ 
 \hline
 $s_{31}$ & $s_{32}$ & $s_{33}$ \\ 
 \hline
\end{tabular}
\end{center}
There are three constraints for the rows and three constraints for the columns:
\begin{align*}
	&s_{11} + s_{12} + s_{13} = 0		\quad	&s_{11} + s_{21} + s_{31} = 0 \\
	&s_{21} + s_{22} + s_{23} = 0		\quad &s_{12} + s_{22} + s_{32} = 0 \\
	&s_{31} + s_{32} + s_{33} = 0		\quad &s_{13} + s_{23} + s_{33} = 1
\end{align*}
In the standard formulation of the Magic Square game, one player is chosen to be a \emph{constraint player}, meaning that they receive a random equation $e = \{s_{i_1 j_1},s_{i_2 j_2},s_{i_3 j_3}\}$ from this linear system. The other player is chosen to be the \emph{variable player}, meaning that they receive a random variable $s_{ij}$ from the equation $e$. The constraint player is supposed to respond with an assignment from $\{0,1\}$ to each of the variables in their received equation, and the variable player is supposed to respond with an assignment to their variable. The players win if the constraint players' assignment satisfies the given equation and if the variable player's assignment is consistent with the constraint player's answers (i.e. the constraint player's assignment for the other player's received variable must match the variable player's response). 

We only deal with games with uniform question distributions in this paper, so the variant of the Magic Square game (which we abbreviate as $\ms$) that we consider is where the questions to Alice and Bob are uniformly and independently chosen from $\cal{X}_\ms = \cal{X}_\equations \cup \cal{X}_\variables$ where 
\begin{align*}
\cal{X}_\equations &= \{r_1,r_2,r_3,c_1,c_2,c_3\},\\
\cal{X}_\variables &= \{s_{11},s_{12},s_{13},s_{21},s_{22},s_{23},s_{31},s_{32},s_{33}\}.
\end{align*}
Here $r_i$ (resp. $c_j$) stands for the equation associated with the $i$th row $\{s_{i 1},s_{i 2},s_{i 3}\}$ (resp. $j$th column $\{s_{1 j},s_{2 j},s_{3 j}\}$). For every constraint $e$ in the Magic Square linear system, let $\cal{A}_e$ denote the set of functions $f_e$ that map variables in $e$ to $\{0,1\}$. The answer set is $\cal{A}_\ms = \cal{A}_\equations \cup \cal{A}_\variables$ where $\cal{A}_\equations$ is the the union of $\cal{A}_e$ over all constraints $e$, and $\cal{A}_\variables = \{0,1\}$. The decision procedure $D_\ms(x,y,a,b)$ for the Magic Square game is described by the following table: if $(x,y)$ (resp. $(y,x)$, as the game is symmetric) is one of the nontrivial question pairs listed, then the players win if and only if the winning condition for the answers $(a,b)$ (resp. $(b,a)$) is satisfied. Otherwise, if the question pair is nontrivial, the players automatically win.

\begin{table}[H]
\centering
\begin{tabularx}{\textwidth}{ L{8cm} L{8cm} } 
 \toprule
 \textbf{Nontrivial Question Pair $(x,y)$} & \textbf{Winning Condition on Answers $(a,b)$} \\
 \midrule
  $x = y$ &  $ a=b$ \\
 \midrule 
 $x \in \cal{X}_\equations, y \in \cal{X}_\variables$ and $y$ is a variable in equation $x$ 
 & $a \in \cal{A}_\equations$ satisfies equation $x$
  and $a(y) = b$ \\
  \bottomrule  
\end{tabularx}
\caption{The nontrivial question pairs and winning conditions for the Magic Square game.}
\label{tab:magic-square}
\end{table}

We now define a value-$1$ synchronous strategy for the Magic Square game. Let $\hilb$ be a Hilbert space and for each variable $s_{ij}$ let $O^{ij}$ denote a self-adjoint unitary operator (called an \emph{observable}) acting on $\hilb$. Suppose that by arranging them into a $3 \times 3$ grid, the observables satisfy the following algebraic relations:
\begin{enumerate}
	\item (\textbf{R1}) The product of observables in a row or column multiply to $\id$, except in the last column, where they multiply to $-\id$.
	\item (\textbf{R2}) Two observables in the same row or column commute with each other;
	\item (\textbf{R3}) Two observables not in the same row or column anti-commute with each other.
\end{enumerate}
First, we note that it is possible to find such a set of observables satisfying these algebraic relations (see \Cref{fig:magic-square-operators} for an example of unitary operators acting on $\C^2 \otimes \C^2$).
\vspace{10pt}
\begin{figure}[H]
\begin{center}
  	 \begin{tabular}{|c | c | c |}
	 \hline
       $Z \otimes \id$ & $\id \otimes Z$ & $Z \otimes Z$ \\
       \hline
         $\id \otimes X$ & $X \otimes \id$ & $X \otimes X$ \\
         \hline
    $Z \otimes X$ & $X \otimes Z$ & $XZ \otimes ZX$ \\
    \hline
    \end{tabular}
 \end{center}
 \caption{An example of optimal observables for the Magic Square game, where the $X$ and $Z$ operators are the same as in the canonical CHSH strategy.}
 \label{fig:magic-square-operators}
\end{figure} 
Second, we note that relation \textbf{R3} is actually a consequence of relations \textbf{R1} and \textbf{R2}. For example to obtain $O^{11}O^{22} = -O^{22}O^{11}$ one could repeatedly apply \textbf{R1} and \textbf{R2} in the following order
\begin{align}
    (O^{11}\,\,O^{22})^2 &= (O^{12}\,\,O^{13})(O^{23}\,\,O^{21})(O^{21}\,\,O^{31})(O^{32}\,\,O^{12})\nonumber\\
    &= O^{12}(O^{13}\,\,O^{23})(O^{21}\,\,O^{21})(O^{31}\,\,O^{32})O^{12}\nonumber\\
    &= -O^{12}\,\,O^{33}\,\,O^{33}\,\,O^{12}= -\id.\label{eq:R1-R2-imply-R3}
\end{align}
However we include \textbf{R3} because the anti-commutation relation turns out to be the most important one in our applications of rigidity.

Given a set $\cal{O} = \{O^{ij}\}$ of observables satisfying relations \textbf{R1}, \textbf{R2}, and \textbf{R3}, we can define the synchronous strategy $\strategy = (\tau,\{M^x\})$ where $\tau$ is a tracial state on the von Neumann algebra generated by the observables $\cal{O}$. For a variable question $s_{ij}$, define the measurement operator $M^{s_{ij}}_b$ to be the projection onto the eigenspace of $O^{ij}$ with eigenvalue $(-1)^b$. To aid notation we abbreviate $M^{s_{ij}}_b$ as $M^{ij}_b$. The operator $M^e_a$ corresponding to a constraint question $e \in \cal{X}_\equations$ is the product
	\begin{align}\label{eq:projection-for-equations-magic-square}
		\prod_{s_{ij} \in e} M^{ij}_{a(s_{ij})}
	\end{align}
	where the product is over variables $s_{ij}$ occurring in equation $e$, and $a$ is an assignment to variables in $e$. Notice that because of relation \textbf{R2}, if $s_{i_1 j_1},s_{i_2 j_2} \in e$ then 
	\begin{align*}M^{i_1 j_1}_{b_1} \,\, M^{i_2 j_2}_{b_2} &= 1/4(\id + (-1)^{b_1}O^{i_1 j_1})(\id + (-1)^{b_2}O^{i_2 j_2})\\ &= 1/4(\id + (-1)^{b_2}O^{i_2 j_2})(\id + (-1)^{b_1}O^{i_1 j_1})\\ &= M^{i_2 j_2}_{b_2} \,\, M^{i_1 j_1}_{b_1}
	\end{align*} for every $b_1,b_2 \in \{0,1\}$. So the order of the product in \Cref{eq:projection-for-equations-magic-square}  doesn't matter, and thus $M^e_a$ is also a projection. 

	It is easy to verify that this strategy for the Magic Square game attains winning probability $1$; this relies on the relations \textbf{R1} and \textbf{R2}. Let us verify this in a few simple steps. Conditioned on players receiving a trivial question pair, the players winning probability is $1$ (as in this case players win regardless of their answers). Conditioned on receiving the same question, the players respond with the same answer with probability $1$ because $\strategy$ is a projective strategy. Indeed conditioned on receiving question pair $(s_{ij},s_{ij})$, the probability of winning is \[\tau(M^{ij}_0 \,\, M^{ij}_0) + \tau(M^{ij}_1 \,\,M^{ij}_1) = \tau(M^{ij}_0 + M^{ij}_1) = \tau(\id) = 1.\] Similarly conditioned on question pair $(e,e) \in \cal{X}_\equations \times \cal{X}_\equations$, the probability of winning is $$\sum_{a \in \cal{A}_e} \tau(M^e_a M^e_a) = \sum_{a \in \cal{A}_e} \tau(M^e_a) = \tau(\id) = 1.$$
	Finally, conditioned on receiving question pair $(r_i,s_{ij})$, the probability that the constraint player's assignment for $s_{ij}$ matches the variable player's answer to $s_{ij}$ is 
	\begin{align*}
	\sum_{a \in \cal{A}_{r_i}}\tau(M^{r_i}_a \, M^{ij}_{a(s_{ij})}) &= \sum_{b \in \cal{A}_\variables}\sum_{\substack{a\in \cal{A}_{r_i}\\ a(s_{ij})=b}} \tau(M^{r_i}_a \, M^{ij}_{b})\\ &=  \sum_{b \in \cal{A}_\variables}\tau(M^{ij}_b \, M^{ij}_b)= \sum_{b \in \cal{A}_\variables}\tau(M^{ij}_b)= \tau(\id)= 1
	\end{align*}
	and the probability that the constraint player's assignment satisfies equation $r_i$ is 
	\begin{align*}
	\sum_{\substack{a \in \cal{A}_{r_i}\\ a(s_{i1}) + a(s_{i2}) + a(s_{i3})=0}}\tau(M^{r_i}_a) &\geq \sum_{\substack{a \in \cal{A}_{r_i}}} (-1)^{a(s_{i1})+a(s_{i2})+a(s_{i3})}\tau(M^{r_i}_a)\\&=\sum_{a\in \cal{A}_{r_i}} (-1)^{a(s_{i1})+a(s_{i2})+a(s_{i3})}\tau(M^{i1}_{a(s_{i1})}\,M^{i2}_{a(s_{i2})}\,M^{i3}_{a(s_{i3})})\\ &=\tau(O^{i1} \, O^{i2} \, O^{i3}) = \tau(\id) = 1.
	\end{align*}
	A similar calculation holds for question pairs $(c_j,s_{ij})$. Since conditioned on any question pair the winning probability is $1$, we conclude that $\val(\ms,\strategy) = 1$. It should also be clear that this strategy is oracularizable, meaning that measurements corresponding to nontrivial question pairs commute. Finally, letting $O^{ij}$ be the Pauli observables in \Cref{fig:magic-square-operators}, we obtain a finite dimensional oracularizable perfect synchronous strategy for the Magic Square game defined over the Hilbert space $\complex^4$. 

We now establish the rigidity of the Magic Square game. Let $\strategy = (\tau,\{M^x\})$ denote a synchronous strategy for the Magic Square game. Each $\{M^{ij}_b\}_{b\in \cal{A}_\ms}$ is a projective measurement with outcomes $b \in \cal{A}_\ms$. Without loss of generality, we assume that the measurements corresponding to variable questions $s_{ij}$ only produce either $0$ or $1$ as answers, i.e.,
\begin{equation}
\label{eq:var-meas-completeness}
M^{ij}_0 + M^{ij}_1 = \id~.
\end{equation}
This is because for variable questions we can always define $M^{ij}_1$ to be the orthogonal projection $\id - M^{ij}_0$, and this cannot decrease the winning probability. Similarly, without loss of generality, we assume that the projective measurement $\{M^e_a\}_{a\in \cal{A}_\ms}$ corresponding to constraint question $e$ only produces assignments in $\cal{A}_e$, that is $\sum_{a \in \cal{A}_e} M^e_a = \id$. 

For every variable $s_{ij} \in \cal{X}_\variables$, define the observable
\[
	O^{ij} = M_0^{ij} - M_1^{ij}~.
\]
Note that $O^{ij}$ is a self-adjoint unitary operator (because of the assumption in~\cref{eq:var-meas-completeness}) and that $M^{ij}_b$ is a projection onto an eigenspace of $O^{ij}$. 

The rigidity of the Magic Square game is expressed in the following way: if $\strategy$ is an (approximately) optimal strategy for the Magic Square game, then the observables must (approximately) satisfy the algebraic relations \textbf{R1}, \textbf{R2}, and \textbf{R3}. 

\begin{theorem}[Rigidity of Magic Square]\label{thm:rigidty-of-magic-square}
Let $\strategy = (\tau,\{M^x\})$ be a synchronous strategy such that $\val(\ms, \strategy) \geq 1-\eps$. Let $\{O^{ij} \} $ denote the observables associated to the strategy. Then
\begin{enumerate}
	\item (\textbf{R1}) The product of observables in a row or column approximately multiply to $\id$, except in the last column, where they approximately multiply to $-\id$:
	\begin{align*}
        O^{i1} \, O^{i2} \, O^{i3} &\approx_{\delta(\eps)} \id \quad \text{for $i = 1,2,3$},\\
        O^{1j} \, O^{2j} \, O^{3j} &\approx_{\delta(\eps)} \id \quad \text{for $j = 1,2$},\\
        O^{13} \, O^{23} \, O^{33} &\approx_{\delta(\eps)} -\id~.
    \end{align*}
	\item (\textbf{R2}) Two observables in the same row or column approximately commute with each other, that is for all $i,j,k \in [3]$
	\begin{align*}
	    O^{ij}\, O^{ik} \approx_{\delta(\eps)} O^{ik}\, O^{ij}~,\\
	    O^{ji}\, O^{ki} \approx_{\delta(\eps)} O^{ki}\, O^{ji}~.
	\end{align*}
	\item (\textbf{R3}) Two observables not in the same row or column anti-commute with each other, so for example
	\[
	    O^{11}\, O^{22} \approx_{\delta(\eps)} - O^{22}\, O^{11}~,
	    O^{12}\, O^{21} \approx_{\delta(\eps)} - O^{21}\, O^{12}~,
    \]
\end{enumerate}
In all of these approximations $\delta$ is some proper error function such that $\delta(\eps) \leq 32|\cal{X}_\ms| \sqrt{\eps}$.
\end{theorem}
\begin{proof}
We saw earlier that \textbf{R3} is implied by \textbf{R1} and \textbf{R2}. This is also the main idea behind the proof here. We first show that $\{O^{ij}\}$ approximately satisfies \textbf{R1} and \textbf{R2}, then we use a derivation similar to (\ref{eq:R1-R2-imply-R3}), to conclude that \textbf{R3} is approximately satisfied.

We can deduce a number of consistency conditions from the fact that the strategy $\strategy$ succeeds with probability at least $1 - \eps$. First, by a simple averaging argument, since every question pair $(x,y) \in \cal{X}_\ms \times \cal{X}_\ms$ is sampled uniformly at random, the winning probability conditioned on players receiving any fixed question pair $(x,y)$ is at least $1 - |\cal{X}_\ms|^2$.

As a notation aid, let $R^i_a = M^{r_i}_a$ denote a row measurement operator and $C^j_a = M^{c_j}_a$ denote a column measurement operator. By the winning conditions in \Cref{tab:magic-square}, the constraint and variable players' answers must be consistent with high probability. In other words $\sum_{a\in \cal{A}_{r_i}} \trace{R^i_a \, M^{ij}_{a(s_{ij})}}$ is at least as large as the probability of winning conditioned on players receiving question pair $(r_i,s_{ij})$ for every $i,j \in [3]$. So from our remark earlier, we have
\begin{equation}
\label{eq:magic-square-winning-implication0}
	\sum_{a \in \cal{A}_{r_i}} \trace{R^i_a \, M^{ij}_{a(s_{ij})}} \geq 1 - |\cal{X}_\ms|^2 \eps~.
\end{equation}
For every row measurement operator $R^i_a$ we define marginal projection operators: for $j \in [3]$ and $b \in \{0,1\}$ define
\[
	R^{ij}_b = \sum_{a \in \cal{A}_{r_i} :\,\, a(s_{ij}) = b} R^i_a
\]
where the summation is over assignments $a$ that assigns value $b$ to variable $s_{ij}$. This is a projection and notice that for all assignments $a$ to variables in $r_i$, we have
\[
	R^i_a = R^{i1}_{a(s_{i1})} \cdot R^{i2}_{a(s_{i2})} \cdot R^{i3}_{a(s_{i3})}~.
\]
It is also clear that $\{R^{ij}_b\}_{b \in \{0,1\}}$ forms a projective measurement. We can similarly define, for all columns $j$ and variables $s_{ij}$, projective measurement $\{C^{ji}_b\}$ consisting of operators
\[
	C^{ji}_b = \sum_{a \in \cal{A}_{c_j} :\,\, a(s_{ij}) = b} C^j_a.
\]

We can rewrite \eqref{eq:magic-square-winning-implication0} in terms of projective measurements $\{R^{ij}_b\}_{b \in \{0,1\}}$ as follows
\begin{align*}
    1 - |\cal{X}_\ms|^2 \eps \leq \sum_{a \in \cal{A}_{r_i}} \trace{R^i_a \, M^{ij}_{a(s_{ij})}} = \sum_{b\in \cal{A}_\variables}\sum_{\substack{a \in \cal{A}_{r_i}:\\a(s_{ij}) = b}} \trace{R^i_a \, M^{ij}_{b}} = \sum_{b\in \cal{A}_\variables} \trace{R^{ij}_b \, M^{ij}_{b}}. 
\end{align*}
Using the notation for consistency between measurements, we can equivalently express this as 
\[
	R^{ij}_b \simeq_{|\cal{X}_\ms|^2\eps} M^{ij}_b~,
\]
where the answer set is $\cal{A}_\variables = \{0,1\}$. By \Cref{lem:consistency-consequences}, we convert consistency to closeness to obtain
\[
	R^{ij}_b \approx_{|\cal{X}_\ms|\sqrt{2\eps}} M^{ij}_b~,
\]
and with a similar argument for columns we get that \[C^{ji}_b \approx_{|\cal{X}_\ms|\sqrt{2\eps}} M^{ij}_b~.\]

At this point it will be more convenient for us to work with observables, rather than projection operators. We have already defined observable $O^{ij}$ for each variable $s_{ij}$; we now define observables corresponding to the (marginal) constraint operators: for all $i,j \in [3]$, define
\[
	R^{ij} = R^{ij}_0 - R^{ij}_1 \qquad \text{and} \qquad C^{ji} = C^{ji}_0 - C^{ji}_1~.
\]
The closeness between constraints and variable projective measurements can be expressed also in terms of observables using the triangle inequality
\[
\|O^{ij} - R^{ij}\|_\tau^2 \leq 2\|M^{ij}_0 - R^{ij}_0\|_\tau^2 + 2\|M^{ij}_1 - R^{ij}_1\|_\tau^2 \leq 4|\cal{X}_\ms|^2\eps.
\]
The same holds for columns, therefore overall we have proved that
\begin{gather}
    O^{ij} \approx_{2|\cal{X}_\ms|\sqrt{\eps}} R^{ij}~,\label{eq:magic-square-winning-implication1-robust}\\
    O^{ij} \approx_{2|\cal{X}_\ms|\sqrt{\eps}} C^{ji}\label{eq:magic-square-winning-implication2-robust}.
\end{gather}

Now using these relations, we can prove that variable observables in the same row or column approximately commute. This follows from a few simple steps. First, by the triangle inequality, for every $i,j,k \in [3]$ we can write
\begin{align}
    \|O^{ij}\, O^{ik} -  O^{ik} \, O^{ij} \|_\tau^2 &\leq 2\|O^{ij}\, O^{ik} - R^{ij}\, R^{ik}\|_\tau^2 + 2\|R^{ij}\, R^{ik} - R^{ik} \, R^{ij}\|_\tau^2 + 2\|R^{ik} \, R^{ij} - O^{ik} \, O^{ij}\|_\tau^2\nonumber\\
    &= 2\|O^{ij}\, O^{ik} - R^{ij}\, R^{ik}\|_\tau^2 + 2\|R^{ik} \, R^{ij} - O^{ik} \, O^{ij}\|_\tau^2~. \label{eq:magic-square-triangle-inequality}
\end{align}
where we used the equality $R^{ij}\, R^{ik} = R^{ik}\, R^{ij}$ which follows from the fact that projections $R^{ij}_b$ and $R^{ik}_c$ are marginals of the same projective measurement $\{R^i_a\}_{a\in\cal{A}_{r_i}}$ and projections belonging to the same projective measurement commute. By \Cref{prop:exchange-operators}, from (\ref{eq:magic-square-winning-implication1-robust}), we get that $O^{ij}O^{ik} \approx_{2|\cal{X}_\ms|\sqrt{\eps}} R^{ij}O^{ik}$. Again by \Cref{prop:exchange-operators}, from (\ref{eq:magic-square-winning-implication1-robust}), we get that $R^{ij}O^{ik} \approx_{2|\cal{X}_\ms|\sqrt{\eps}} R^{ij}R^{ik}$. So by triangle inequality we have
\begin{align*}
    \|O^{ij}\, O^{ik} - R^{ij}\, R^{ik}\|_\tau^2 \leq 2\|O^{ij}O^{ik} - R^{ij}O^{ik}\|_\tau^2 + 2\|O^{ij}R^{ik} - R^{ij}R^{ik}\|_\tau^2 \leq 16|\cal{X}_\ms|^2\eps .
\end{align*}
This is true for all $i,j,k \in [3]$, so in particular it also holds that
\begin{align*}
    \|R^{ik} \, R^{ij} - O^{ik} \, O^{ij}\|_\tau^2 \leq 16|\cal{X}_\ms|^2\eps.
\end{align*}
Now plugging these in (\ref{eq:magic-square-triangle-inequality}) we get that
\[\|O^{ij}\, O^{ik} -  O^{ik} \, O^{ij} \|_\tau^2 \leq 32|\cal{X}_\ms|^2\eps.\]
An identical argument can be applied to columns, so overall we proved
\begin{gather}
    O^{ij}\, O^{ik} \approx_{4|\cal{X}_\ms|\sqrt{2\eps}}  O^{ik} \, O^{ij}~,\label{eq:magic-square-commutation1}\\
    O^{ji}\, O^{ki} \approx_{4|\cal{X}_\ms|\sqrt{2\eps}} O^{ki}\, O^{ji}~,\label{eq:magic-square-commutation2}
\end{gather}
for every $i,j,k \in [3]$.

As mentioned in \Cref{subsec:asymptotics}, in this paper we do not need to keep track of the specific approximation bounds. As such, instead of carrying around subscripts like $4|\cal{X}_\ms|\sqrt{2\eps}$ in our approximations, we opt to instead write $O^{ij} \approx_{\delta(\eps)} R^{ij}$ where $\delta$ is some error function such that $\delta(\eps) \to 0$ as $\eps \to 0$. For example in the rest of this paper the argument above will be abbreviated as follows: From $O^{ij} \approx_{\delta(\epsilon)} R^{ij}$ for all $i,j \in [3]$ and repeated applications of \Cref{prop:exchange-operators}, we obtain \[O^{ij}\, O^{ik} \approx_{\delta(\epsilon)} R^{ij}\, R^{ik} = R^{ik}\, R^{ij} \approx_{\delta(\epsilon)} O^{ik}\, O^{ij}~,\] so by the triangle inequality \[O^{ij}\, O^{ik} \approx_{\delta(\epsilon)} O^{ik}\, O^{ij}~,\] where $\delta(\eps)$ are proper error functions. It is only in this proof that, for the benefit of the reader who sees these approximations for the first time, we tried to give the arguments in full details and kept track of all the error functions. 

So far we obtained consequences of the fact that in a strategy with large winning probability the constraint and variable players' answers are consistent with high probability. There are some other relations that must hold in any approximately optimal strategy. For instance, with high probability, the measurement outcome of a constraint measurement $\{M^e_a\}_{a \in \cal{A}_e}$ must be a satisfying assignment for the constraint $e$. Let us make this more precise. The probability of winning conditioned on players receiving question pair $(r_i,s_{ij})$ is at least $1 - |\cal{X}_\ms|^2 \epsilon$. By winning conditions in \Cref{tab:magic-square}, if players win on question pair $(r_i,s_{ij})$, then the assignment by the player receiving question $r_i$ must satisfy constraint $r_i$. So we can write
\[\sum_{\substack{a \in \cal{A}_{r_i}\\a(s_{i1}) + a(s_{i2}) + a(s_{i3}))=0}} \trace{R^i_a} \geq 1- |\cal{X}_\ms|^2 \epsilon.\]
Now from the fact that $\{R^i_a\}_{a\in \cal{A}_{r_i}}$ is a projective measurement, we get that
\[\sum_{a \in \cal{A}_{r_i}} (-1)^{a(s_{i1}) + a(s_{i2}) + a(s_{i3})} \trace{R^i_a} \geq 1- 2|\cal{X}_\ms|^2 \epsilon,\]
and in terms of observables this can be equivalently written as
\[\trace{R^{i1}R^{i2}R^{i3}} \geq 1- 2|\cal{X}_\ms|^2 \eps~.\]
By \Cref{prop:unitary-close-to-identity}, we get that
\begin{equation}
    R^{i1} \, R^{i2} \, R^{i3} \approx_{2|\cal{X}_\ms|\sqrt{\eps}} \id \quad \text{for $i = 1,2,3$}~.\label{eq:magic-square-row-rigidity}
\end{equation}
Doing the same for columns we get
\[
    C^{j1} \, C^{j2} \, C^{j3}  \approx_{2|\cal{X}_\ms|\sqrt{\eps}} \id \quad \text{for $j = 1,2$}
\] and
\[ C^{31} \, C^{32} \, C^{33}  \approx_{2|\cal{X}_\ms|\sqrt{\eps}} -\id ~\] 
Now by (\ref{eq:magic-square-winning-implication1-robust}) and (\ref{eq:magic-square-row-rigidity}), and repeated applications of \Cref{prop:exchange-operators} and the triangle inequality, for every $i \in [3]$, we obtain
\begin{align*}
    \|O^{i1} \, O^{i2} \, O^{i3}\|_\tau^2 &\leq 2\|O^{i1} \, O^{i2} \, O^{i3} - R^{i1} \, O^{i2} \, O^{i3}\|_\tau^2 + 2\|R^{i1} \, O^{i2} \, O^{i3} - R^{i1} \, R^{i2} \, O^{i3}\|_\tau^2\\ &\qquad+ 2\|R^{i1} \, R^{i2} \, O^{i3} - R^{i1} \, R^{i2} \, R^{i3}\|_\tau^2 + 2\|R^{i1} \, R^{i2} \, R^{i3} - \id\|_\tau^2\\
    &\leq 32 |\cal{X}_\ms|^2\eps.
\end{align*}
Therefore we have
\begin{align}
    O^{i1} \, O^{i2} \, O^{i3} &\approx_{4|\cal{X}_\ms|\sqrt{2\eps}} \id \quad \text{for $i = 1,2,3$}, \label{eq:magic-square-rel1}
\end{align}
and following the same argument for columns
\begin{align}
    O^{1j} \, O^{2j} \, O^{3j} &\approx_{4|\cal{X}_\ms|\sqrt{2\eps}} \id \quad \text{for $j = 1,2$},\label{eq:magic-square-rel2}\\
    O^{13} \, O^{23} \, O^{33} &\approx_{4|\cal{X}_\ms|\sqrt{2\eps}} -\id~.\label{eq:magic-square-rel3}
\end{align}

Finally to prove the approximate anticommutation $O^{11}\,\,O^{22} \approx - O^{22}O^{11}$, we follow the idea in the derivation \ref{eq:R1-R2-imply-R3}: We start with $(O^{11}\,\,O^{22})^2$ and step by step, using relations (\ref{eq:magic-square-rel1})-(\ref{eq:magic-square-rel3}), substitute $O^{11}$ and $O^{22}$ by unitaries that are nearby. By repeated applications of triangle inequality and \Cref{prop:exchange-operators} and the approximate relations we established so far, we can write
\begin{align*}
    (O^{11}\,\,O^{22})^2 &\approx_{16|\cal{X}_\ms|\sqrt{\epsilon}} (O^{12}\,\,O^{13})(O^{23}\,\,O^{21})(O^{21}\,\,O^{31})(O^{32}\,\,O^{12})\\
    &= O^{12}(O^{13}\,\,O^{23})(O^{21}\,\,O^{21})(O^{31}\,\,O^{32})O^{12}\\
    &=O^{12}(O^{13}\,\,O^{23})(O^{31}\,\,O^{32})O^{12}\\
    &\approx_{8|\cal{X}_\ms|\sqrt{2\epsilon}} -O^{12}\,\,O^{33}\,\,O^{33}\,\,O^{12}\\
    &= -\id,
\end{align*}
So altogether, with another application of triangle inequality, we obtain \[\|(O^{11}\,O^{22})^2 + \id \|_\tau \leq 32|\cal{X}_\ms|\sqrt{\epsilon}.\]
Now since $O^{11}O^{22}$ is a unitary and the $\tau$-norm is unitarily invariant, we conclude that \[\|O^{11}O^{22} + O^{22}O^{11}\|_\tau\leq 32|\cal{X}_\ms|\sqrt{\epsilon}.\] By symmetry, an almost identical argument can be applied to prove anticommutation relations for all other pairs of observables not in the same row or column.
\end{proof}
As mentioned, the rigidity of the Magic Square and CHSH games are important stepping stones for a number of results in quantum complexity theory and quantum cryptography. A crucial component of obtaining strong lower bounds on the complexity of approximating the value of nonlocal games has been through developing nonlocal games with \emph{highly efficient} rigidity properties. 

We measure efficiency via the tradeoff between the complexity of the game versus the complexity of the algebraic relations that (approximately) optimal strategies must satisfy. For example, the Magic Square game has $|\cal{X}_\ms|^2 = 15^2$ question pairs and a similar number of answer pairs, and (approximately) optimal strategies must give rise to two pairs of (approximately) anti-commuting observables $\{O^{11}, O^{22} \}$ and $\{O^{21}, O^{12}\}$, and furthermore these pairs must be \emph{independent} in the sense that they (approximately) commute with each other. This implies that when the probability of winning is sufficiently close to $1$, the dimension of the Hilbert space must be at least $4$. We say that the Magic Square game \emph{certifies} the existence of two independent anti-commuting observables and certifies a Hilbert space of dimension at least $4$. This is a consequence of the following general statement:

\begin{proposition}
\label{prop:dim-bound}
Let $\algebra$ denote a von Neumann algebra on a separable Hilbert space $\hilb$ with a tracial state $\tau$, and let $A^{(1)},\ldots,A^{(n)},B^{(1)},\ldots,B^{(n)} \in \algebra$ denote self-adjoint unitary operators (i.e. observables). Suppose for some $\eps \geq 0$ the following approximate commutation and anticommutation relations hold:
\begin{gather*}
	\forall \, i, \qquad A^{(i)} B^{(i)} \approx_\eps - B^{(i)} A^{(i)} \\
	\forall \, i \neq j, \qquad A^{(i)} A^{(j)} \approx_\eps A^{(j)} A^{(i)}~, \qquad B^{(i)} B^{(j)} \approx_\eps B^{(j)} B^{(i)} ~, \qquad  A^{(i)} B^{(j)} \approx_\eps B^{(j)} A^{(i)}~.
\end{gather*}
Then, for all sufficiently small $\epsilon$, it holds that $\dim \hilb \geq (1 - \delta(\eps)) 2^{n}$ where $\delta(\epsilon)$ is some proper error function.
\end{proposition}

\begin{proof}
There is nothing to prove when $\hilb$ is infinite dimensional. So assume that $\hilb$ is finite dimensional. By Theorem 4.4.1 in \cite{jones-lecture-notes}, every finite dimensional von Neumann algebra is a direct sum of $B(\hilb^i)$ where $\hilb^i$ are finite dimensional Hilbert spaces. So without loss of generality we may assume $\algebra = B(\hilb)$ and that $\tau(\cdot) = \tr(\cdot)/\dim \hilb$ is the dimension-normalized trace.

Let $\Pi^{(i)}_b$ be the projection onto $(-1)^b$-eigenspace of $A^{(i)}$. For every $s \in \{0,1\}^n$ let \[M_s \coloneqq \Bigparen{\prod_{i=1}^{n} \Pi^{(i)}_{s_i}} \Bigparen{\prod_{i=1}^{n} \Pi^{(i)}_{s_i}}^\ast~.\] These operators are clearly positive semidefinite and a simple inductive argument shows that $\sum_{s\in \{0,1\}^n} M_s = \id$. Therefore $\{M_s\}_{s \in \{0,1\}^n}$ is a POVM.

From approximate commutation relations between $A^{(i)}$s we get that any pair $\Pi^{(i)}_a$ and $\Pi^{(j)}_b$ must approximately commute. Therefore by repeated applications of \Cref{prop:exchange-operators}, we get that \[M_s^2 \approx_{\delta(\eps)} M_s.\] By \Cref{prop:exchange-operators} again, we obtain that $\tau(M_s - M_s^2) \leq \delta(\eps)$ for every $s$. So by \Cref{lem:projectivization}, there exists a projective measurement $\{P_s\}_{s \in \{0,1\}^n} \subset \algebra$ such that $P_s \approx_{\delta(\eps)} M_s$. 

By approximate anticommutation, we get $B^{(i)} \, A^{(i)} \, B^{(i)} \approx_{\delta(\eps)} -A^{(i)}$. We can express this in terms of projective measurement $\{\Pi_0^{(i)},\Pi_1^{(i)}\}$ \[B^{(i)} \, \Pi^{(i)}_0 \, B^{(i)} - B^{(i)} \, \Pi^{(i)}_1 \, B^{(i)} \approx_{\delta(\eps)} \Pi^{(i)}_1 - \Pi^{(i)}_0.\]
Using the relation $\Pi^{(i)}_0 + \Pi^{(i)}_1 = \id$, we conclude that 
\begin{equation}
B^{(i)} \, \Pi^{(i)}_0 \, B^{(i)} \approx_{\delta(\eps)} \Pi^{(i)}_1.\label{eq:unitaries-act-as-permutation}
\end{equation}

Now if we define unitary operators $U_{s,t} \coloneqq \prod_{i=1}^n \paren{B^{(i)}}^{s_i + t_i},$ it is straightforward to show that \[U_{s,t} \, M_s \, U_{s,t}^* \approx_{\delta(\eps)} M_t\] for every $s,t \in \{0,1\}^n$ using (\ref{eq:unitaries-act-as-permutation}) and approximate commutation and anticommutations between $A$ and $B$ operators. This immediately implies that \[\tau(M_t) \approx_{\delta(\eps)} \tau(U_{s,t} \, M_s \, U_{s,t}^*) = \tau(M_s).\] Now since projections $\{P_s\}$ are close to operators $\{M_s\}$ we also have $\tau(P_s) \approx_{\delta(\eps)} \tau(P_t)$ for every $s,t$.

From $\tau(\sum_s P_s) = \tau(\id) = 1$ and the fact that $\tau(P_s) \approx \tau(P_{t})$ for every $s,t \in \{0,1\}^n$, we get that $\tau(P_s) \approx_{\delta(\eps)} 2^{-n}$. In other words we have \[(1-\delta(\eps))2^{-n} \leq \tau(P_s) \leq (1+\delta(\eps))2^{-n}\] for every $s$. For all $\eps$ sufficiently small, we have $\delta(\eps) < 1$, and thus $\tau(P_s) > 0$. Since $P_s$ is a projection and it is nonzero it must be that $\tr(P_s) \geq 1$ so $\tau(P_s) = \tr(P_s)/ \dim \hilb \geq 1/\dim \hilb$. We can write \[1/\dim \hilb \leq \tau(P_s) \leq (1+\delta(\eps))2^{-n}\] from which we conclude that \[\dim \hilb \geq \frac{2^n}{1+\delta(\eps)} \geq (1-\delta(\eps))2^n.\]
\end{proof}

It is possible to construct games that certify a larger Hilbert space. An example is the \emph{$n$-fold parallel repetition} of the Magic Square game, which is a nonlocal game where the verifier plays $n$ independent instances of the Magic Square game simultaneously with the two players. This game is also rigid, and it certifies $2n$ pairs of independent anti-commuting observables and consequently, by the proposition we just proved, certifies a Hilbert space of dimension $2^{2n}$. However the complexity of the game also scales commensurately with the dimension: the number of questions and answers grows as $2^{O(n)}$. 

Are there games that certify a $d$-dimensional Hilbert space using much fewer than $d$ questions/answer pairs? Chao, Reichardt, Sutherland and Vidick~\cite{chao2018test} and Natarajan and Vidick \cite{natarajan2018low} showed that there exist families of games $\{G_n\}$ where the $n$-th game $G_n$ certifies a $2^n$-dimensional space using $\poly(n)$ question/answer pairs. The rigidity result of~\cite{natarajan2018low} is also highly \emph{robust}, in the sense that strategies for $G_n$ that succeed with probability $1 - \eps$ must be $\delta(\eps)$-close to satisfying the target algebraic relations, for some function $\delta(\eps)$ that has a mild (e.g., logarithmic) dependence on $n$. The existence of games with efficient and robust rigidity properties is a key component of the gap-preserving compression theorem of~\cite{ji_mip_re}.\footnote{In fact, the result of~\cite{ji_mip_re} implies that one can construct games with $m$ questions/answers that certify $d$-dimensional Hilbert spaces, and $d$ can be an arbitrarily large (computable) function of $m$!}

For our gapless compression result, we only need games with efficient rigidity properties (i.e., small game certifying a large Hilbert space), not necessarily highly robust ones. In this paper we use a family of games that we call \emph{$2$-out-of-$n$ Magic Square}, which is inspired by the family of games introduced in~\cite{chao2018test}, which we call $2$-out-of-$n$ CHSH. We describe the $2$-out-of-$n$ Magic Square games next.

\subsection{The $2$-out-of-$n$ Magic Square game}
\label{sec:2-of-n-ms}

\newcommand{\twoofnms}{\textsc{$2$-of-$n$-MS}}

Fix an integer $n > 0$. The basic idea behind the $2$-out-of-$n$ Magic Square game, abbreviated $\twoofnms$, is that the players are asked to play $n$ simultaneous instances of the Magic Square game, but the verifier only asks the players for their responses for $2$ instances. Define the question set $\cal{X}_{\twoofnms} = \{(i,j) \in [n]^2: i \neq j\} \times \cal{X}_{\ms}^2$, and the answer set $\cal{A}_{\twoofnms} = \cal{A}_\ms^2$. The decision predicate $D_\twoofnms(q,r,a,b)$ is specified as follows, via its nontrivial question pairs and the corresponding winning conditions for the answers.

\begin{table}[H]
\centering
\begin{tabularx}{\textwidth}{ L{7cm} L{10cm} } 
 \toprule
 \textbf{Nontrivial Question Pair $(q,r)$} & \textbf{Winning Condition on Answers $(a,b)$} \\
 \midrule
  $q = r$ &  $ a=b$ \\
 \midrule 
 $q = (i,j,x_i,x_j), r = (k,\ell,y_k,y_\ell)$ & $D_\ms(x_w,y_w,u_w,v_w) = 1$ for all $w \in \{i,j\} \cap \{k,\ell \}$ \\
 where $\{i,j\} \cap \{k,\ell \} \neq \emptyset$, and for all $w$ in the intersection, $(x_w,y_w)$ is a nontrivial question pair for $\ms$  & where $a = (u_i,u_j), b = (v_k,v_\ell)$ \\
  \bottomrule  
\end{tabularx}
\caption{The nontrivial question pairs and winning conditions for the $\twoofnms$.}
\label{tab:two-of-n-ms}
\end{table}
In other words, each player gets asked to generate answers for two instances of the Magic Square game, but do not know what instances the other player is asked about. If there is an instance $i$ that is asked to both players, then their questions and answers for instance $i$ must satisfy the Magic Square decision predicate.

It is easy to see that the $\twoofnms$ has a perfect synchronous strategy: let $\strategy_\ms = (\tau,\{M^x\})$, where $\tau$ is a tracial state on some von Neumann algebra $\algebra$ on a Hilbert space $\hilb$,  denote the perfect strategy for the Magic Square game described above. Then define the synchronous strategy $\strategy_\twoofnms = (\tau^{\otimes n}, \{M^{i,j,x,y}\})$, where  $M^{i,j,x,y} = \{M^{i,j,x,y}_{a,b}\}_{a,b \in \cal{A}_\ms}$ is the projective measurement defined such that \[M^{i,j,x,y}_{a,b} \coloneqq \id\otimes\cdots\otimes\id \otimes M^{x}_{a}\otimes \id \otimes \cdots \otimes \id \otimes M^{y}_{b}\otimes\id\otimes\cdots\otimes\id \in \algebra^{\otimes n}\] in which $M^{x}_{a}$ and $M^{y}_{b}$ are acting on the $i$th and $j$th copy of $\hilb$, respectively. Intuitively if a player receives the question $(i,j,x,y)$ they perform independent Magic Square measurements corresponding to questions $x$ and $y$ on the $i$-th and $j$-th copy of $\hilb$, respectively, and respond with their measurement outcomes. Clearly, the players’ will win the instances that are shared between them. The oracularizability of this strategy follows from the oracularizablity of the honest strategy of the Magic Square game and the construction above: for example if $(x_i,y_i)$ is a nontrivial question pair in the Magic Square game, then measurements $M^{i,j,x_i,x_j}$ and $M^{i,k,y_i,y_k}$ commute for all $j\neq k$ since measurements $M^{x_i}$ and $M^{y_i}$ commute by the oracularizability of the honest Magic Square strategy from the previous section.

The next lemma expresses the rigidity properties of the $\twoofnms$. Let $\{ M^{i,j,x,y}_{a,b} \}_{a,b \in \cal{A}_\ms}$ denote a measurement corresponding to a question $(i,j,x,y) \in \cal{X}_{\twoofnms}$. Define the marginal measurement operator
\begin{align*}
	M^{i,x}_{a} = \sum_b M^{i,\Succ(i),x,x}_{a,b}
\end{align*}
where the sum is over answers $b \in \cal{A}_\ms$ and $\Succ(i) = \begin{cases}i+1, &\quad i < n,\\ 1, &\quad i = n.\end{cases}$

Note that for all $(i,x) \in [n] \times \cal{X}_\ms$, the set $\{M^{i,x}_{a}\}_{a \in \cal{A}_\ms}$ forms a projective measurement. Just like with strategies for the Magic Square game, when $x$ is a variable question in the Magic Square game (i.e. it is $s_{cd}$ for some $c,d \in [3]$), we assume without loss of generality that 
\[
	M^{i,s_{cd}}_0 + M^{i,s_{cd}}_1 = \id
\]
for all $i \in [n], c,d \in [3]$. For each variable $s_{cd}$ define the corresponding observable
\[
	O^{i,c,d} = M^{i,s_{cd}}_0 - M^{i,s_{cd}}_1~.
\]

\begin{lemma}[Rigidity of the $\twoofnms$]
\label{lem:2-of-n-ms}
Let $\strategy = (\tau,\{M^x\})$ be a synchronous strategy such that $\val(\twoofnms,\strategy) \geq 1 - \eps$. For all $i \in [n]$ define
\begin{gather*}
    A^{(2i-1)} = O^{i,1,1}~, B^{(2i-1)} = O^{i,2,2}~,\\
    A^{(2i)} = O^{i,1,2}~, B^{(2i)} = O^{i,2,1}~.
\end{gather*}
Then 
\begin{gather*}
	\forall \, k \in [2n], \qquad A^{(k)} B^{(k)} \approx_\delta - B^{(k)} A^{(k)} \\
	\forall \, k, l \in [2n] \text{ and }k\neq l, \qquad A^{(k)} A^{(l)} \approx_\delta A^{(l)} A^{(k)}~, \qquad B^{(k)} B^{(l)} \approx_\delta B^{(l)} B^{(k)} ~, \qquad  A^{(k)} B^{(l)} \approx_\delta B^{(l)} A^{(k)}
\end{gather*}
where $\delta(n,\eps) = \poly(n) \cdot \poly(\eps)$ is a proper error function.
\end{lemma}

\begin{proof}
Fixing $i \in [n]$ and $x,y \in \cal{X}_\ms$, the probability of winning the instance $i$ Magic Square game, conditioned on players receiving questions $(i,\Succ(i),x,x)$ and $(i,\Succ(i),y,y)$ is at least $1-|\cal{X}_{\twoofnms}|^2\eps$, thus
\[\sum_{a,b} \tau(M^{i,x}_a \,\, M^{i,y}_b) D_\ms(x,y,a,b) \geq 1- |\cal{X}_{\twoofnms}|^2\eps.\]
So conditioned on every question pair $(x,y)$, the strategy $(\tau,\{M^{i,x}\}_{x\in \cal{\ms}})$ wins in the Magic Square game with probability at least \[1-|\cal{X}_{\twoofnms}|^2\eps = 1 - \poly(n,\eps).\] Therefore by \Cref{thm:rigidty-of-magic-square}, for every $i \in [n]$, we have
\begin{gather*}
    A^{(2i-1)} \, B^{(2i-1)} \approx_{\poly(n,\eps)} -B^{(2i-1)} \, A^{(2i-1)}~, A^{(2i)} \, B^{(2i)} \approx_{\poly(n,\eps)} -B^{(2i)} \, A^{(2i)}~,\\
    A^{(2i-1)} \, A^{(2i)} \approx_{\poly(n,\eps)} A^{(2i)} \, A^{(2i-1)}~, B^{(2i-1)} \, B^{(2i)} \approx_{\poly(n,\eps)} B^{(2i)} \, B^{(2i-1)}~,\\ A^{(2i-1)} \, B^{(2i)} \approx_{\poly(n,\eps)} B^{(2i)} \, A^{(2i-1)}~, B^{(2i-1)} \, A^{(2i)} \approx_{\poly(n,\eps)} A^{(2i)} \, B^{(2i-1)}~.
\end{gather*}
It is only left to prove that when $k,l \in [2n]$ and $|k-l| > 1$, it holds that
\[\qquad A^{(k)} A^{(l)} \approx_\delta A^{(l)} A^{(k)}~, \qquad B^{(k)} B^{(l)} \approx_\delta B^{(l)} B^{(k)} ~, \qquad  A^{(k)} B^{(l)} \approx_\delta B^{(l)} A^{(k)}.\]
We prove the stronger statement that $M^{i,x}_a \, M^{j,y}_b \approx_\delta M^{j,x}_b \, M^{i,y}_a$ for all $i,j \in [n], i\neq j, x,y \in \cal{X}_\ms, a,b \in \cal{A}_\ms$. 

We give the proof for the case where $j \neq \Succ(i)$ and $i \neq \Succ(j)$. The proof for the other cases follow the same idea. The proof is based on the cross-check between nontrivial question pair $(i,\Succ(i),x,x)$ and $(i,j,x,y)$ on one hand and the cross-check between nontrivial question pair $(i,j,x,y)$ and $(j,\Succ(j),y,y)$ on the other hand. We derive consequences of the fact that, conditioned on players receiving questions $(i,\Succ(i),x,x)$ and $(i,j,x,y)$, they win instance $i$ of the Magic Square with high probability. Similarly we derive consequences of the fact that, conditioned on players receiving questions $(j,\Succ(j),y,y)$ and $(i,j,x,y)$, they win instance $j$ of Magic Square with high probability. The consequences we derive are then used to prove the desired approximate commutation relations. 

Recall that by the winning conditions of the Magic Square game, if players win (in the Magic Square game) when receiving the same question, then they must have responded with the same answer. This can be expressed as
\begin{align*}
    \sum_{a\in \cal{A}_\ms} \sum_{b,c \in \cal{A}_\ms} \tau(M^{i,\Succ(i),x,x}_{a,b} \, M^{i,j,x,y}_{a,c}) \geq 1-|\cal{X}_{\twoofnms}|^2\eps~,
\end{align*}
or in other words
\begin{align*}
    \sum_{a\in \cal{A}_\ms} \tau(M^{i,x}_{a} \, \sum_c M^{i,j,x,y}_{a,c}) \geq 1-|\cal{X}_{\twoofnms}|^2\eps~.
\end{align*}
In terms of consistency relations this can be expressed as $M^{i,x}_a \simeq_\delta \sum_c M^{i,j,x,y}_{a,c}$. 

Similarly we have
\begin{align*}
    \sum_{b\in \cal{A}_\ms} \sum_{c,d \in \cal{A}_\ms} \tau(M^{j,\Succ(j),y,y}_{b,c} \, M^{i,j,x,y}_{d,b}) \geq 1-|\cal{X}_{\twoofnms}|^2\eps~,
\end{align*}
or in other words
\begin{align*}
    \sum_{a\in \cal{A}_\ms} \tau(M^{j,y}_{b} \, \sum_d M^{i,j,x,y}_{d,b}) \geq 1-|\cal{X}_{\twoofnms}|^2\eps~.
\end{align*}
In terms of consistency relations this can be expressed as $M^{j,y}_b \simeq_\delta \sum_c M^{i,j,x,y}_{c,b}$. 

Using \Cref{lem:consistency-consequences} we turn the consistency relations to the following closeness relations
\[
    M^{i,x}_a \approx_\delta \sum_c M^{i,j,x,y}_{a,c}~, M^{j,y}_b \approx_\delta \sum_d M^{i,j,x,y}_{d,b}~,
\]
where $\delta$ is some proper error function. Now using \Cref{prop:exchange-operators}, we can write
\begin{align*}
    M^{i,x}_a \, M^{j,y}_b &\approx  \Bigparen{\sum_c M^{i,j,x,y}_{a,c}}\Bigparen{\sum_d M^{i,j,x,y}_{d,b}}\\
    &= \Bigparen{\sum_d M^{i,j,x,y}_{d,b}}\Bigparen{\sum_c M^{i,j,x,y}_{a,c}}\\
    &\approx  M^{j,y}_b \, M^{i,x}_a,
\end{align*}
where the equality follows from the fact that projection operators belonging to the same projective measurement commute.
\end{proof}

\Cref{prop:dim-bound} immediately implies that any strategy that succeeds for the $\twoofnms$ with probability $1 - \eps$ must be on a Hilbert space of dimension at least $(1 - \poly(n)\poly(\delta)) 2^{2n}$, which is nontrivial for $\delta < 1/\poly(n)$. Furthermore, this game is highly efficient because the number of questions and answers grows only \emph{polynomially} with $n$. Observe that
\[
	|\cal{X}_{\twoofnms}| = n^2 \cdot |\cal{X}_\ms|^2 ~, \qquad |\cal{A}_{\twoofnms}| = |\cal{A}_\ms|^2~,
\]
which means that the total number of question and answer pairs for the $\twoofnms$ is $O(n^4)$, where we treat the question and answer sizes of the Magic Square game as constant.

\subsection{The Question Sampling game}\label{sec:question-sampling-game}

For readers who are familiar with quantum information theory, the $\twoofnms$ can be understood in the following way. In the honest strategy for $\twoofnms$ the two players share the state $\ket{\epr}^{\otimes 2n}$ (i.e. $2n$ maximally entangled Bell pairs), and if we assume the perfect strategy for the Magic Square game is the one coming from \Cref{fig:magic-square-operators}, the observables $A^{(1)},\ldots,A^{(2n)},B^{(1)},\ldots,B^{(2n)}$, defined in \Cref{lem:2-of-n-ms}, are $A^{(i)} = Z_i$ and $B^{(i)} = X_i$ where $Z_i$ (resp. $X_i$) represents the $2n$-qubit operator with the $Z$ (resp. $X$) Pauli operator acting on the $i$-th qubit and identity everywhere else.
Then by the rigidity of $\twoofnms$, in any approximately optimal strategy, there are observable that are close to these Pauli operators. These Pauli operators act nontrivially only on a single qubit. However for the question reduction in \Cref{sec:question-reduction}, we need access to the measurements that simultaneously measure blocks of qubits. To achieve this goal, in this section, we extend the $\twoofnms$ by including a few additional questions. By doing so, and as it becomes clear in a moment, we guarantee that any optimal strategy for the extended game must be using these block-qubit measurement operators.

We now introduce a family of synchronous games called \emph{Question Sampling games}, denoted by $\qs = \{ \qs_n \}_{n \in \N}$. The $n$-th Question Sampling game $\qs_n$ is an extension of the $\twoofnms$ where there are four additional questions $\sample_\alice,\sample_\bob,\erase_\alice,\erase_\bob$, where $\sample$ and $\erase$ stand for \emph{sample} and \emph{erase}, respectively. The answers for these additional questions are $n$-bit strings.

In the honest strategy for the Question Sampling game (which we formally introduce in a moment), the $\sample_\alice$ (resp. $\sample_\bob$) measurement is supposed to correspond to measuring the first $n$ (resp. second $n$) EPR pairs in the standard basis, whereas the $\erase_\alice$ (resp. $\erase_\bob$) measurement is supposed to correspond to measuring the first $n$ (resp. second $n$) EPR pairs in a complementary basis. 

The rigidity of the $\twoofnms$ (\Cref{lem:2-of-n-ms}) implies that measurements of strategy with high winning probability give rise to $2n$ pairs of (approximately) anticommuting observables $(A^{(i)},B^{(i)})_{i \in [2n]}$, and the observables (approximately) commute across different pairs. This rigidity guarantee is also present for the Question Sampling game $\qs_n$, but furthermore the measurements corresponding to the additional questions also satisfy the following:
\begin{itemize}
	\item The measurements corresponding to $\sample_\alice$ (resp. $\sample_\bob$) are approximately consistent with ``simultaneously measuring'' the observables $A^{(1)},\ldots,A^{(n)}$ (resp. $A^{(n+1)},\ldots,A^{(2n)}$) to produce an $n$-bit string answer.
	\item The measurements corresponding to $\erase_\alice$ (resp. $\erase_\bob$) are approximately consistent with ``simultaneously measuring'' the observables $B^{(1)},\ldots,B^{(n)}$ (resp. $B^{(n+1)},\ldots,B^{(2n)}$) to produce an $n$-bit string answer.	
\end{itemize}
Here, ``approximate consistency'' is used in the sense defined in \Cref{sec:measurements}. Furthermore, since the observables referred to in each item above only approximately commute with each other, the notion of simultaneous measurement is only meant in an approximate sense; we formalize this below in \Cref{thm:rigidity-of-question-sampling}.

We now formally define the game $\qs_n = (\cal{Q}_n,\cal{X}_n,D_{\qs_n})$. Its question set is defined to be $\cal{Q}_n = \cal{X}_{\twoofnms} \cup \{ \sample_\alice, \sample_\bob, \erase_\alice, \erase_\bob \}$, and thus $|\cal{Q}_n| = \poly(n)$. Its answer set is defined to be $\cal{X}_n = \cal{A}_{\twoofnms} \cup \{0,1\}^n$, and thus $|\cal{X}_n| = O(2^n)$.

\begin{remark}
The Question Sampling game and the Introspection game, appearing in the next section, are the only games in this paper for which we use the symbol $\cal{Q}$ (instead of $\cal{X}$) to refer to the question set. In fact, for the Question Sampling game the letter $\cal{X}$ is reserved for the answer set. The reason for this convention is because, as the name suggests, the Question Sampling game is meant to sample a question pair $(x,y)$ for another game (this should become clearer in the section on Introspection games).
\end{remark}

The nontrivial questions and winning conditions of the decision procedure $D_{\qs_n}(q,r,x,y)$ are specified as follows (note that the answers are now denoted $(x,y)$). We only consider the case of even $n$. The case of odd $n$ is slightly more tedious to write down.

\begin{table}[H]
\centering
\begin{tabularx}{\textwidth}{ L{8cm} L{8cm} } 
 \toprule
 \textbf{Nontrivial Question Pair $(q,r)$} & \textbf{Winning Condition on Answers $(x,y)$} \\
 \midrule
 	$q = r$ &  $ x=y$ \\
 \midrule
 \midrule
$(q,r)$ is a nontrivial question for $\twoofnms$ & $D_{\twoofnms}(q,r,x,y) = 1$ \\
 \midrule
 \midrule
 $q = (i,j,s_{11},.) \in \cal{X}_{\twoofnms}$ where $i \leq \frac{n}{2}, j > \frac{n}{2}$, and $r = \sample_\alice$ & $x = (a_i,a_j) \in \cal{A}_\ms^2$, $y \in \{0,1\}^n$, and $y_{2i-1} = a_i$ \\
 \midrule
 $q = (i,j,s_{12},.) \in \cal{X}_{\twoofnms}$ where $i \leq \frac{n}{2}, j > \frac{n}{2}$, and $r = \sample_\alice$ & $x = (a_i,a_j) \in \cal{A}_\ms^2$, $y \in \{0,1\}^n$, and $y_{2i} = a_i$ \\
  \midrule
 $q = (i,j,s_{11},.) \in \cal{X}_{\twoofnms}$ where $i > \frac{n}{2}, j \leq \frac{n}{2}$, and $r = \sample_\bob$ & $x = (a_i,a_j) \in \cal{A}_\ms^2$, $y \in \{0,1\}^n$, and $y_{2(i-\frac{n}{2})-1} = a_i$ \\
 \midrule
 $q = (i,j,s_{12},.) \in \cal{X}_{\twoofnms}$ where $i > \frac{n}{2}, j \leq \frac{n}{2}$, and $r = \sample_\bob$ & $x = (a_i,a_j) \in \cal{A}_\ms^2$, $y \in \{0,1\}^n$, and $y_{2(i-\frac{n}{2})} = a_i$ \\
 \midrule
 \midrule
 $q = (i,j,s_{22},.) \in \cal{X}_{\twoofnms}$ where $i \leq \frac{n}{2}, j > \frac{n}{2}$, and $r = \erase_\alice$ & $x = (a_i,a_j) \in \cal{A}_\ms^2$, $y \in \{0,1\}^n$, and $y_{2i-1} = a_i$ \\
 \midrule
 $q = (i,j,s_{21},.) \in \cal{X}_{\twoofnms}$ where $i \leq \frac{n}{2}, j > \frac{n}{2}$, and $r = \erase_\alice$ & $x = (a_i,a_j) \in \cal{A}_\ms^2$, $y \in \{0,1\}^n$, and $y_{2i} = a_i$ \\
  \midrule
 $q = (i,j,s_{22},.) \in \cal{X}_{\twoofnms}$ where $i > \frac{n}{2}, j \leq \frac{n}{2}$, and $r = \erase_\bob$ & $x = (a_i,a_j) \in \cal{A}_\ms^2$, $y \in \{0,1\}^n$, and $y_{2(i-\frac{n}{2})-1} = a_i$ \\
 \midrule
 $q = (i,j,s_{21},.) \in \cal{X}_{\twoofnms}$ where $i > \frac{n}{2}, j \leq \frac{n}{2}$, and $r = \erase_\bob$ & $x = (a_i,a_j) \in \cal{A}_\ms^2$, $y \in \{0,1\}^n$, and $y_{2(i-\frac{n}{2})} = a_i$ \\
 \bottomrule  
\end{tabularx}
\caption{The nontrivial question pairs and winning conditions for the $n$-th Question Sampling game. We used dot for example in $(i,j,s_{11},.) \in \cal{X}_{\twoofnms}$ to indicate that the fourth coordinate does not matter as long as the quadruple is a valid question in $\cal{X}_{\twoofnms}$.}
\label{tab:question-sampling}
\end{table}

We now to describe an oracularizable synchronous strategy for $\qs_n$ with value $1$. Let $\strategy_\ms = (\tau,\{M^q\}_{q\in \cal{X}_\ms})$ be the honest strategy for the Magic Square game on the Hilbert space $\hilb_{\ms} = \complex^4$ and let $\strategy_{\twoofnms} = (\tau^{\otimes n},\{M^q\}_{q\in \cal{X}_\twoofnms})$ be its extension to a perfect oracularizable synchronous strategy for the $\twoofnms$ as defined in \Cref{sec:2-of-n-ms}. We extend this to a perfect finite-dimensional oracularizable synchronous strategy $\strategy_{\qs_n}$ for $\qs_n$. 

For every $y \in \{0,1\}^n$ define 
\begin{align*}
M^{\sample_\alice}_y &\coloneqq M^{s_{11}}_{y_1} \, M^{s_{12}}_{y_2} \otimes M^{s_{11}}_{y_3} \, M^{s_{12}}_{y_4} \otimes \cdots \otimes M^{s_{11}}_{y_{n-1}} \, M^{s_{12}}_{y_n} \otimes \id_{\complex^{2^n}},\\
M^{\sample_\bob}_y &\coloneqq  \id_{\complex^{2^n}} \otimes M^{s_{11}}_{y_1} \, M^{s_{12}}_{y_2} \otimes M^{s_{11}}_{y_3} \, M^{s_{12}}_{y_4} \otimes \cdots \otimes M^{s_{11}}_{y_{n-1}} \, M^{s_{12}}_{y_n},\\
M^{\erase_\alice}_y &\coloneqq M^{s_{22}}_{y_1} \, M^{s_{21}}_{y_2} \otimes M^{s_{22}}_{y_3} \, M^{s_{21}}_{y_4} \otimes \cdots \otimes M^{s_{22}}_{y_{n-1}} \, M^{s_{21}}_{y_n} \otimes \id_{\complex^{2^n}},\\
M^{\erase_\bob}_y &\coloneqq \id_{\complex^{2^n}} \otimes M^{s_{22}}_{y_1} \, M^{s_{21}}_{y_2} \otimes M^{s_{22}}_{y_3} \, M^{s_{21}}_{y_4} \otimes \cdots \otimes M^{s_{22}}_{y_{n-1}} \, M^{s_{21}}_{y_n}.
\end{align*}
Note that measurements $M^{s_{11}}$ and $M^{s_{12}}$ (and similarly $M^{s_{22}}$ and $M^{s_{21}}$) of the honest Magic Square strategy commute as they belong to the same row. It is easily verified that $\{M^{\sample_\alice}_y\}, \{M^{\sample_\bob}_y\}, \{M^{\erase_\alice}_y\}, \{M^{\erase_\bob}_y\}$ are projective measurements and that $\strategy_{\qs_n} = (\tau^{\otimes n},\{M^{q}\}_{q \in \cal{Q}_{\qs_n}})$ is a synchronous strategy for $\qs_n$.\footnote{If we take the Magic Square strategy from \Cref{fig:magic-square-operators}, these formulas simplify to
\begin{align*}
M_y^{S_A} &= \ketbra{y}{y}\otimes \id,\\
M_y^{S_B} &= \id \otimes \ketbra{y}{y},\\
M_y^{\erase_A} &= H^{\otimes n} \ketbra{y}{y} H^{\otimes n} \otimes \id,\\
M_y^{\erase_B} &= \id \otimes H^{\otimes n} \ketbra{y}{y} H^{\otimes n},
\end{align*}
where $H = \frac{1}{\sqrt{2}}\begin{bmatrix}1 & 1\\ 1 & -1\end{bmatrix}$ is the Hadamard transform. } 

Next we show that $\strategy_{\qs_n}$ wins with probability $1$. Fix an $i \leq \frac{n}{2}, j > \frac{n}{2}, t \in \cal{X}_\ms$. Conditioned on players receiving the nontrivial question pair $((i,j,s_{11},t), \sample_\alice)$, which corresponds to the third row in \Cref{tab:question-sampling}, the probability of winning is
\begin{align*}
    \sum_{a \in \cal{A}_\ms} \sum_{y \in \{0,1\}^n} \tau(M^{i,j,s_{11},t}_{y_{2i-1},a} M^{\sample_\alice}_y) = \sum_{y \in \{0,1\}^n} \tau(M^{i,s_{11}}_{y_{2i-1}} M^{\sample_\alice}_y)
    = \sum_{y \in \{0,1\}^n} \tau(M^{\sample_\alice}_y)
    = 1,
\end{align*}
in which $M^{i,s_{11}}_{y_{2i-1}}$ is defined to be the marginal \[M^{i,s_{11}}_{y_{2i-1}} \coloneqq \sum_{a \in \cal{A}_\ms} M^{i,j,s_{11},t}_{y_{2i-1},a}
    = \id_{\hilb_\ms}^{i-1} \otimes M^{s_{11}}_{y_{2i-1}}\otimes \id_{\hilb_\ms}^{n-i-1}.\]
It is similarly verified that the probability of winning conditioned on any other question pair is $1$.

Since $\strategy_{\twoofnms}$ is oracularizable in $\twoofnms$, to verify the oracularizability of $\strategy_{\qs_n}$ we just need to check commutativity between measurements for $\sample_\alice,\sample_\bob,\erase_\alice,\erase_\bob$ on one hand and measurements for $(i,j,q_i,q_j)$ on the other hand. This follows very easily from the construction of the measurements \[M^{\sample_\alice},M^{\sample_\bob},M^{\erase_\alice},M^{\erase_\bob}\]

Finally we note that in the honest strategy $\tau(M^{\sample_\alice}_x M^{\sample_\bob}_y) = 2^{-2n}$ (and similarly $\tau(M^{\erase_\alice}_x M^{\erase_\bob}_y) = 2^{-2n}$) for all $x,y \in \{0,1\}^{n}$. We see in a moment that approximately optimal strategies approximately satisfy these relations. 

Let $\strategy = (\tau,\{M^q\}_{q\in \cal{Q}_{\qs_n}})$ be a synchronous strategy for the Question Sampling game. For convenience we use the notational shorthand
\begin{gather*}
	S^\alice_x = M^{\sample_\alice}_x \qquad \text{and} \qquad S^\bob_x = M^{\sample_\bob}_x \\
	E^\alice_x = M^{\erase_\alice}_x \qquad \text{and} \qquad E^\bob_x = M^{\erase_\bob}_x
\end{gather*}
for all $x \in \{0,1\}^n$. We also define a family of observables derived from these measurements as follows. For all $u \in \{0,1\}^n$, 
\begin{gather*}
	O^{\sample_\alice}_u = \sum_{x \in \{0,1\}^n} (-1)^{u \cdot x} \, S^\alice_x \qquad \text{and} \qquad O^{\sample_\bob}_u = \sum_{x \in \{0,1\}^n} (-1)^{u \cdot x} \, S^\bob_x \\
	O^{\erase_\alice}_u = \sum_{x \in \{0,1\}^n} (-1)^{u \cdot x} \, E^\alice_x \qquad \text{and} \qquad O^{\erase_\bob}_u = \sum_{x \in \{0,1\}^n} (-1)^{u \cdot x} \, E^\bob_x~.
\end{gather*}
Note that by construction these are self-adjoint unitaries, and therefore observables. 
We call $S^\alice,S^\bob$ (resp. $E^\alice,E^\bob$) \emph{sampling measurements} (resp. \emph{erasure measurements}), and $O^{\sample_\alice}, O^{\sample_\bob}$ (resp. $O^{\erase_\alice}, O^{\erase_\bob}$) \emph{sampling observables} (resp. \emph{erasure observables})~. In what follows we write $\ca = \bob, \cb = \alice$.

\begin{theorem}[Rigidity of the Question Sampling game]\label{thm:rigidity-of-question-sampling}
Let $\strategy = (\tau,\{M^q\}_{q\in\cal{Q}_n})$ be a synchronous strategy such that $\val(\qs_n,\strategy) \geq 1 - \eps$. Then for all $W \in \{\alice,\bob\}$, 
\begin{enumerate}
\item The sampling (resp. erasure) measurements almost commute with one another, that is for every $x,y \in \{0,1\}^n$
\begin{align*}
    S^\alice_{x} S^\bob_{y} \approx S^\bob_{y} S^\alice_{x} \qquad \text{and} \qquad E^\alice_{x} E^\bob_{y} \approx E^\bob_{y} E^\alice_{x}~.
\end{align*}
\item Sampling measurements $S_W$ almost commute with erasure measurements $E_\cw$, that is, for every $x,y \in \{0,1\}^n$,
\begin{align*}
    S^W_x E^{\cw}_y \approx E^{\cw}_y S^W_x.
\end{align*}

\item The erasure observables $O^{\erase_W}$ approximately permute the sampling measurements $S^W$ and vice versa. That is, for every $u,x \in \{0,1\}^n$,
\begin{align*}
    O^{\erase_W}_u S^W_x O^{\erase_W}_u \approx S^W_{x+u} \qquad \text{and} \qquad O^{\sample_W}_u E^W_x O^{\sample_W}_u \approx E^W_{x+u}~.
\end{align*}
where the arithmetic in the subscript is bitwise XOR.
\item Finally, for all $x,y \in \{0,1\}^n$, 
\begin{gather*}
	\tau(S^{W}_x) \approx 2^{-n} \qquad \text{and} \qquad \tau(S^{W}_x S^{\cw}_{y}) \approx 2^{-2n}~,\\
	\tau(E^{W}_x) \approx 2^{-n} \qquad \text{and} \qquad \tau(E^{W}_x E^{\cw}_{y}) \approx 2^{-2n}~.
\end{gather*}
\end{enumerate}
\end{theorem}
We explained the usage of $\approx$ in \Cref{subsec:asymptotics}. For a detailed example see the proof of \Cref{thm:rigidty-of-magic-square}.

\begin{proof}
By the winning conditions of the game, for all $i \leq n/2$ and $j > n/2$, we have 
\begin{align*}
1 - \delta(\eps) &\geq \sum_{b,c \in \{0,1\}}\sum_{\substack{x\in\{0,1\}^n:\\x_{2i-1} = b}} \trace{S^\alice_x  M^{i,j,s_{11},s_{11}}_{b,c}}\\
&= \sum_{b \in \{0,1\}} \trace{S^\alice_{[x\mapsto x_{2i-1}|b]}  \Bigparen{\sum_{c\in \{0,1\}} M^{i,j,s_{11},s_{11}}_{b,c}}}.
\end{align*}
By the proof of rigidity of $\twoofnms$ we have $M^{i,s_{11}}_b \approx \sum_{c\in \{0,1\}} M^{i,j,s_{11},s_{11}}_{b,c}$
where $M^{i,s_{11}}_b$ is the marginal $\sum_{c\in \{0,1\}} M^{i,\Succ(i),s_{11},s_{11}}_{b,c}$ as defined in the previous section. So we can rewrite our earlier inequality as
\[
\sum_{b \in \{0,1\}} \trace{S^\alice_{[x\mapsto x_{2i-1}|b]} M^{i,s_{11}}_b} \geq 1 - \delta(\eps)~.
\]
Using \Cref{lem:consistency-consequences} we can write this as closeness relation \[S^\alice_{[x\mapsto x_{2i-1}|b]} \approx M^{i,s_{11}}_{b}.\]  With a similar argument we obtain \[S^\alice_{[x\mapsto x_{2i}|b]} \approx M^{i,s_{12}}_{b}.\] Now using the identity \[S^\alice_x = \prod_{i=1}^n S^\alice_{[y\mapsto y_i|x_i]}\] and repeated applications of \Cref{prop:exchange-operators}, we obtain
\[S^\alice_x \approx \prod_{i=1}^{n/2} M^{i,s_{11}}_{x_{2i-1}}M^{i,s_{12}}_{x_{2i}}~.\]
With a similar argument we obtain
\begin{align*}
    S^\bob_x &\approx \prod_{i=1}^{n/2} M^{i+n/2,s_{11}}_{x_{2i-1}}M^{i+n/2,s_{12}}_{x_{2i}}~,\\
    E^\alice_x &\approx \prod_{i=1}^{n/2} M^{i,s_{22}}_{x_{2i-1}}M^{i,s_{21}}_{x_{2i}}~,\\
    E^\bob_x &\approx \prod_{i=1}^{n/2} M^{i+n/2,s_{22}}_{x_{2i-1}}M^{i+n/2,s_{21}}_{x_{2i}}~.
\end{align*}
Now by the definition of the sampling and erasure observables, we have 
\begin{align*}
    O^{\sample_\alice}_u &\approx (A^{(1)})^{u_1}(A^{(2)})^{u_2}\cdots (A^{(n)})^{u_n}~,\\
    O^{\sample_\alice}_u &\approx (A^{(n/2+1)})^{u_1}(A^{(n/2+2)})^{u_2}\cdots (A^{(n)})^{u_n}~,\\
    O^{\erase_\alice}_u &\approx (B^{(1)})^{u_1}(B^{(2)})^{u_2}\cdots (B^{(n)})^{u_n}~,\\
    O^{\erase_\alice}_u &\approx (B^{(n/2+1)})^{u_1}(B^{(n/2+2)})^{u_2}\cdots (B^{(n)})^{u_n}~,\\
\end{align*}
where $A^{(i)}$ and $B^{(j)}$ are as defined in \Cref{lem:2-of-n-ms}. Properties 1-3 now follow easily from the rigidity of $\twoofnms$ in \Cref{lem:2-of-n-ms}.

Finally, we prove 4 using 1-3. We have $O^{E_W}_x S^W_x O^{E_W}_x \approx S^{W}_{0^n}$ for every $x\in \{0,1\}^n$. Applying Proposition \ref{prop:exchange-operators}, we obtain $\tau(O^{E_W}_x S^{W}_x O^{E_W}_x) \approx \tau(S^{W}_{0^n})$. By cyclicity of tracial states we have $\tau(S^{W}_x) \approx \tau(S^{W}_{0^n})$. Now \[1 = \tau(\sum_x S^{W}_x) \approx 2^n \tau(S^{W}_{0^n}),\] from which we get that $\tau(M^{S_W}_{0^n}) \approx 2^{-n}$. Similarly $\tau(S^{W}_x) \approx 2^{-n}$ for $x\neq 0^n$.

Similar to the above line of reasoning, by repeated applications of \Cref{prop:exchange-operators} we have
\begin{align*}
    1 &= \sum_{x,y} \tau(S_x^{W} S_{y}^{\cw})\\
      &= \sum_{x,y} \tau((O_x^{\erase_W})^2 (O_{y}^{\erase_\cw})^2 S_x^{W} S_{y}^{\cw})\\
      &\approx \sum_{x,y} \tau(O_x^{\erase_W} S_x^{W} O_x^{\erase_W} O_{y}^{\erase_\cw} S_{y}^{\cw} O_{y}^{\erase_\cw})\\
      &\approx \sum_{x,y} \tau(S_{0^n}^{W} S_{0^{n}}^{\cw})\\
      &= 2^{2n} \tau(S_{0^n}^{W} S_{0^{n}}^{\cw}).
\end{align*}
In the first approximation we used the fact that $W$ operators approximately commute with $\cw$ operators. The proof for erasure measurements is identical.
\end{proof}

\begin{corollary}[Entanglement bound for Question Sampling] \label{cor:entanglement-bound-question-sampling}Let $\strategy = (\tau,\{M^q\}_{q\in \cal{Q}_n})$ be a synchronous strategy for $\qs_n$ over a von Neumann algebra $\algebra \subset B(\hilb)$. If $\val(\qs_n,\strategy)\geq 1-\eps$ for sufficiently small $\epsilon > 0$, then $\dim(\hilb) > (1-\delta(n,\eps))2^{2n}$. 

Furthermore there exists a projection $\Pi \in \algebra$ such that $\tau(\Pi) \approx 2^{-2n}$ and $\Pi \approx S^{\alice}_{0^n}S^{\bob}_{0^n}$.
\end{corollary}
\begin{proof}
The inequality $\dim(\hilb) > (1-\delta(n,\eps))2^n$ is immediate from \Cref{lem:2-of-n-ms} and \Cref{prop:dim-bound}. We now prove $\Pi$ exists. Let $M = S^\alice_{0^n}S^\bob_{0^n}S^\alice_{0^n}$ and note that $\{M,\id-M\}$ is a POVM. Indeed we have $0 \preceq S^{\alice}_{0^n}(\id-S^{\bob}_{0^n})S^{\alice}_{0^n} \preceq \id - M$ in positive semidefinite ordering. Since $S^\alice_{0^n}$ and $S^\bob_{0^n}$ approximately commute, we can write
\begin{align*}
    M^2 &= S^\alice_{0^n}S^\bob_{0^n}S^\alice_{0^n}S^\alice_{0^n}S^\bob_{0^n}S^\alice_{0^n}\\
    &\approx S^\alice_{0^n}S^\bob_{0^n}S^\alice_{0^n}\\
    &= M.
\end{align*}
Therefore we also have $(\id-M)^2 = \id - 2M + M^2 \approx \id-M$. So we can apply Lemma \ref{lem:projectivization} to obtain a projection $\Pi \in \algebra$ such that $\Pi \approx S^\alice_{0^n}S^\bob_{0^n}S^\alice_{0^n}$. Now again since $S^\alice_{0^n}$ and $S^\bob_{0^n}$ approximately commute, we get that $\Pi \approx S^\alice_{0^n}S^\bob_{0^n}$. An application of Proposition \ref{prop:exchange-operators} gives us $\tau(\Pi) \approx \tau(S^\alice_{0^n}S^\bob_{0^n})$. The result $\tau(\Pi) \approx 2^{-2n}$ now follows from item 4 in the preceding theorem.

\end{proof}

We finish this section by stating a technical lemma. The lemma holds in a more general setting but here we restricted attention only to the Question Sampling game.

\begin{lemma}\label{lem:new-algebra}
Let $\strategy = (\tau,\{M^q\}_{q\in \cal{Q}_n})$ be a synchronous strategy for $\qs_n$ over a von Neumann algebra $\algebra \subset B(\hilb)$ and suppose $\val(\qs_n,\strategy) \geq 1-\eps$. Also let $\Pi$ be the projection in the preceding corollary and let $\hat{\hilb}$ be the subspace $\Pi$ projects onto. Then the set of operators \[\hat{\algebra} = \{\Pi M \Pi : M \in \algebra\} \subset B(\hat{\hilb})\] is a von Neumann algebra with unit $\Pi$. Furthermore, the functional $\sigma: B(\hat{\hilb}) \to \C$ defined by $\sigma(N) = \frac{\tau(N)}{\tau(\Pi)}$, for every $N \in B(\hat{\hilb})$, is a tracial state on $\hat{\algebra}$.
\end{lemma}
\begin{proof}
For a proof that $\hat{\algebra}$ is a von Neumann algebra see the section on ``Elementary properties of von Neumann algebras'' in the notes by Vaughan Jones \cite{jones-lecture-notes}.
The functional $\sigma$ is a positive linear functional because $\tau$ is a positive linear functional. It is unital because $\sigma(\id_{\hat{\hilb}}) = \sigma(\Pi) = \tau(\Pi)/\tau(\Pi) =1$. It is cyclic on $\hat{\algebra}$ because $\tau$ is cyclic on $\algebra$ and $\hat{\algebra} \subset \algebra$.
\end{proof}

\newpage
\section{Question Reduction}
\label{sec:question-reduction}

In this section we present the Question Reduction transformation, whose properties are given by the following Theorem.
\begin{theorem}[Question Reduction]
\label{thm:question-reduction}
    For all $\alpha \in \N$, there exists a polynomial-time algorithm $\alg{QuestionReduction}_\alpha$ that takes as input a pair of Turing machines $(D,C)$ and outputs a pair of Turing machines $(D^\intro,C^\intro)$ such that the following holds. If $\verifier = (D,C)$ is a verifier for a sequence of games $\UGS_\verifier = (G_n)_{n \in \N}$ and $n_0 \in \N$ is an integer such that for all $n \geq n_0$, 
      \[
    	\max \Big \{ \TIME_C(n), \TIME_D(n) \Big \} \leq n^\alpha~,
    \]
	then $\verifier^\intro = (D^\intro,C^\intro)$ is a verifier corresponding to a sequence of games $\UGS_{\verifier^\intro} = (G^\intro_n)_{n \in \N}$ with the following properties. There exists $\beta = \poly(\alpha) \in \N$ and $n_0^\intro = \poly(\beta,n_0)$ such that for all $n \geq n_0^\intro$,
    \begin{enumerate}
        \item (Complexity bounds)
        \begin{align*}
        	&\text{The questions of $G_n^\intro$ have length at most $\log^\beta n$, } \\
    &\TIME_{C^{\intro}}(n) \leq \log^\beta n~, \text{ and} \\    	
	&\TIME_{D^{\intro}}(n) \leq n^\beta
		\end{align*}
        \item (Completeness) For all oracularizable synchronous strategies $\strategy$ for $G_n$, there exists an oracularizable synchronous strategy $\strategy^\intro$ for $G_n^\intro$ such that
        \[
            \val(G^{\intro}_n,\strategy^\intro) \geq \val(G_n,\strategy).
        \]
        Furthermore, if $\strategy$ is finite-dimensional, then so is $\strategy^\intro$.
        \item (Soundness) For all $t \in \{q,co\}$ we have 
        \[
        	\val_t^s(G_n) < 1 \Longrightarrow \val_t^s(G^{\intro}_n) < 1~.
		\]
		\item (Entanglement bound)
        \[
        	\mathcal{E} (G_n^{\intro},1) \geq \max \left \{ \mathcal{E} (G_n,1) , 2^{2n} \right \}~. 
        \]
    \end{enumerate}
\end{theorem}
Intuitively, the Question Reduction transformation transforms a sequence of games $(G_1,G_2,\ldots)$ to a sequence $(G_1^\intro,G_2^\intro,\ldots)$ of ``Introspection games'' such that the question lengths of the Introspection game $G_n^\intro$ is \emph{polylogarithmic} in the time complexity of the ``original game'' $G_n$ while the value of $G_n^\intro$ approximates the value of $G_n$. In particular, the value of $G_n^\intro$ is $1$ if and only if the value of $G_n$ is $1$. Furthermore, the time complexity of the Introspection game $G_n^\intro$ is polynomial in the time complexity of the original game $G_n$. The reason this is called ``Question Reduction'' is because the question lengths of the original game $G_n$ can be as large as $n^\alpha$ (because that's the time complexity of the decision procedure $D_n$) and the question lengths of the Introspection games are at most $\log^\beta n$. The core of the Question Reduction transformation is the \emph{Introspection protocol}, which is a simplification of the one developed by~\cite{natarajan_neexp,ji_mip_re}. Aside from the fact that we work in the setting of synchronous games, the two other major simplifications are that
\begin{itemize}
    \item we only need to introspect games with uniform question distributions, and
    \item the transformation does not need to be gap preserving.
\end{itemize}
The bulk of this section will be spent on analyzing the Introspection protocol, and then in \Cref{sec:question-reduction-thm-proof} we prove \Cref{thm:question-reduction}.

\subsection{Overview}
\label{sec:introspection-overview}
Let $G = (\cal{X},\cal{A},D)$ be a synchronous game with $\cal{X} = \{0,1\}^\ell, \cal{A} = \{0,1\}^m$. We present a transformation $G \mapsto G^\intro$ where $G^\intro$ is called the \emph{Introspection game} corresponding to $G$. The question lengths of $G^\intro$ will be much smaller than those of $G$, but the values of the two games will still be tightly related. 

At an intuitive level, the question lengths are reduced in $G^\intro$ by asking the players to ``ask themselves'' -- i.e., to introspect -- their own questions from $\cal{X}$. The players in $G^\intro$ are each asked to sample a question $x\in \cal{X}$ and answer with $a\in \cal{A}$ as they would have answered in the original game $G$ if they have received question $x$. The players then each respond with a tuple $(x,a)$. If the players' responses are $(x,a)$ and $(y,b)$, the decision procedure in $G^\intro$ will check that $D(x,y,a,b) = 1$. 

In order for the values of $G$ and $G^\intro$ to be meaningfully related, we need to ensure that the players sample their introspected questions $x$ and $y$ from the uniform distribution (instead of, say, always picking a fixed $(x^*,y^*)$ for which they have prepared winning answers). We ensure this by introducing a small number of special questions in the game $G^\intro$. The cross-checks between these special questions force the players to behave ``honestly'' (i.e., to sample $(x,y)$ from the uniform distribution), or risk losing the game with some nonzero probability.

The Introspection game $G^\intro$ is an extension of the Question Sampling game $\qs_\ell$ from \Cref{sec:question-sampling-game}, where $\ell$ is the bit length of questions in the original game $G$. Recall that the Question Sampling game certifies that the players have measurements for questions $\sample_\alice,\sample_\bob,\erase_\alice,\erase_\bob$ satisfying the rigidity properties detailed in \Cref{thm:rigidity-of-question-sampling}. 

In addition to these questions, the Introspection game has an additional question $\introspect$, which stands for ``introspect''. When a player receives question $\introspect$, they are expected to answer with a tuple $(x,a,y,b) \in (\cal{X} \times \cal{A})^2$, and the players win if $D(x,y,a,b) = 1$. The Introspection game certifies the measurement corresponding to $\introspect$ is consistent with the following measurement process: performing both $\sample_\alice,\sample_\bob$ measurements (which commute with each other) to produce $(x,y) \in \cal{X}^2$, and then performing measurements $N^x$ and $N^y$ (which commute with each other when $(x,y)$ is a nontrivial question pair in the original game) to produce $(a,b) \in \cal{A}^2$. Furthermore, $N^x$ commutes with the $\erase_\bob$ measurement and $N^y$ commutes with the $\erase_\alice$ measurement. 

The fact that the $\introspect$ measurement is consistent with $\sample_\alice,\sample_\bob$ ensures that the distribution of the pair $(x,y)$ is uniform over $\cal{X}^2$. The fact that the the measurements $N^x,N^y$ commute with the $\erase_\bob$ and $\erase_\alice$ measurements, respectively, ensures that the output $a$ of $N^x$ does not depend on $y$ and similarly the output $b$ of $N^y$ does not depend on $x$. Thus the measurements $\{N^x\}$ give rise to a strategy for the original game $G$, and thus the value of $G^\intro$ is related to that of $G$. 

There are several other questions that are used in the Introspection game $G^\intro$ to ensure these consistency properties. Overall, the number of questions in $G^\intro$ is $|\qs_\ell| + 7$, and thus the question lengths represented in binary is $\lceil \log (|\qs_\ell| + 7) \rceil = O(\log(\ell))$.

We formally define the Introspection game next.
\subsection{Definition of Introspection game}
\label{sec:introspection-definition}

Throughout this section, we write $W$ to denote a value from the set $\{\alice,\bob\}$, and we write 
\[
	\overline{W} = \left\{\begin{array}{ll}
		B  & \mbox{if $W = A$},  \\
		A  & \mbox{if $W=B$}.
	\end{array}\right.~.
\]
The Introspection game $G^\intro$ corresponding to $G$ is a synchronous game $(\cal{Q}^\intro, \cal{A}^\intro,D^\intro)$ with
\begin{align*}
    &\cal{Q}^\intro = \cal{Q}_{\qs_\ell} \cup \{ \,\,\, \introspect \,\,\, \} \cup \{\,\,\, \introspect_W, \,\,\, \introspect_W \sample_\cw  \,\,\,, \,\,\, \introspect_W \erase_\cw \,\, \, \}_{W \in \{\alice,\bob\}},\\
    &\cal{A}^\intro = \cal{A}_{\qs_\ell} \cup \cal{X} \cup (\cal{X} \times \cal{A}) \cup (\cal{X} \times \cal{A} \times \cal{X}) \cup (\cal{X}\times\cal{A} \times \cal{X} \times \cal{A})~.
\end{align*}
The symbol $\introspect$ stands for \emph{introspect}, and $\sample$ and $\erase$ stand for \emph{sample} and \emph{erase} as in the Question Sampling game. We emphasize that the symbols $\introspect_W \sample_\cw$ and $\introspect_W \erase_\cw$ respectively are each individual questions; for example $\introspect_\alice \sample_\bob$ is distinct from the questions $\introspect_\alice$ and $\sample_\bob$, and is also distinct from the question $\introspect_\bob \sample_\alice$.

The decision procedure $D^\intro$ is specified by \Cref{tab:introspection}. On question pair $(q,r)$ and answer pair $(\hat{a},\hat{b})$, the decision procedure checks if $(q,r)$ is nontrivial according to the table, and if so, checks the corresponding winning condition. For the sake of clarity, we omit the symmetric case where the question pair is $(r,q)$ and the answer pair is $(\hat{b},\hat{a})$. 

\begin{table}[H]
\centering
\begin{tabularx}{\textwidth}{ L{4cm} L{13cm} } 
 \toprule
 \textbf{Nontrivial Question Pair $(q,r)$} & \textbf{Winning Condition on Answers $(\hat{a},\hat{b})$} \\
 \midrule
 	$q = r$ &  $ \hat{a}=\hat{b}$ \\
 \midrule 
$(q,r)$ is nontrivial for $\qs_\ell$ & $D_{\qs_\ell}(q,r,\hat{a},\hat{b}) = 1$ \\
	 \midrule
	$q = \introspect$  & $\Big( (x_\alice,x_\bob)$ is trivial for $G$ $\Big)$ or $\Big (z = x_W \wedge c = a_W \wedge D(x_\alice,x_\bob,a_\alice,a_\bob) = 1 \Big) $  \\
		$r = \introspect_W$ & where $\hat{a}=(x_\alice,a_\alice,x_\bob,a_\bob) \in (\cal{X} \times \cal{A})^2$ and $\hat{b} = (z,c) \in \cal{X} \times \cal{A}$    \\
	 \midrule
	$q = \introspect_W$  & $z = x_W \wedge c = a_W$ \\
		$r = \introspect_W \sample_\cw$ & where $\hat{a} = (x_W,a_W) \in \cal{X} \times \cal{A}$ and $\hat{b}=(z,c,x_\cw) \in \cal{X} \times \cal{A} \times \cal{X}$ \\
	\midrule
$q = \introspect_W$ & $z = x_W$ \\
$r = \sample_W$  	& where $\hat{a} = (x_W,a_W) \in \cal{X} \times \cal{A}$ and $\hat{b}=z \in \cal{X}$ \\
	\midrule
$q = \introspect_W$  & $z = x_W \wedge c = a_W$ \\
$r = \introspect_W \erase_\cw$ & where $\hat{a} = (x_W,a_W) \in \cal{X} \times \cal{A}$ and $\hat{b}=(z,c,x_\cw) \in \cal{X} \times \cal{A} \times \cal{X}$ \\
	\midrule
$q = \introspect_W \erase_\cw$  & $z = x_\cw$  \\
$r = \erase_\cw$ & where $\hat{a}=(x_W,a_W,x_\cw) \in \cal{X} \times \cal{A} \times \cal{X}$ and $\hat{b}=z \in \cal{X}$ \\
	\midrule
$q = \introspect_W \sample_\cw$  & $z = x_\cw$ \\
$r = \sample_\cw$ &  where $\hat{a}=(x_W,a_W,x_\cw) \in \cal{X} \times \cal{A} \times \cal{X}$ and $\hat{b}=z \in \cal{X}$ \\
\bottomrule  
\end{tabularx}
\caption{The nontrivial question pairs and winning conditions for the Introspection game $G^\intro$.}
\label{tab:introspection}
\end{table}

The nontrivial question pairs of the Introspection game $G^\intro$, apart from those in the Question Sampling game $\qs_\ell$, are also depicted as a graph in \Cref{fig:graph-of-introspection-game}. The questions are connected via an edge if they form a nontrivial question pair (and self-loops are not drawn for clarity). 

\begin{figure}[h]
\centering
\begin{tikzpicture}

    \node (1) {$\introspect_\alice \sample_\bob $};
    \node (2) [right = 3cm of 1]  {$\sample_\bob $};
    
    \node (3) [below = 1cm of 1]{$\introspect_\alice$};
    \node (4) [left = 1cm of 3] {$\introspect_\alice \erase_\bob $}; 
    \node (5) [right = 1.5cm of 3]{$\introspect$};
    \node (6) [below = 1cm of 2]  {$\introspect_\bob$};
    \node (7) [right = 1cm of 6] {$\introspect_\bob \erase_\alice $};
    \node (8) [left = 1cm of 4] {$\erase_\bob$}; 
    \node (9) [right = 1cm of 7] {$\erase_\alice $};
        
    \node (10) [below = 1cm of 3]{$\sample_\alice $};
    \node (11) [below = 1cm of 6]  {$\introspect_\bob \sample_\alice $};

    \path[draw,thick]
	(1) edge node {} (2)
    (1) edge node {} (3)
    (2) edge node {} (6)
    (3) edge node {} (4)
	(3) edge node {} (5)
    (5) edge node {} (6)
    (6) edge node {} (7)
    (8) edge node {} (4)
    (7) edge node {} (9)
    (10) edge node {} (3)
	(10) edge node {} (11)
    (11) edge node {} (6)
    ;
\end{tikzpicture}
\caption{A node indicates a special question in $G^\intro$. A pair of questions are connected with an edge if the pair is a nontrivial question pair as defined in Section \ref{sec:nonlocal-games}. There should also be loops on every node (which we omitted here for clarity). 
}
\label{fig:graph-of-introspection-game}
\end{figure}
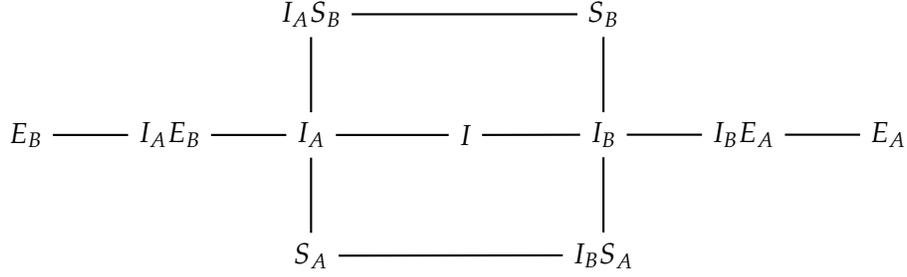
The rationale behind the questions $\introspect_W\sample_\cw$ and $\introspect_W \erase_\cw$ is the following.  A player receiving the composite question $\introspect_W\sample_\cw$, for example, is expected to answer both questions $\introspect_W$ and $\sample_\cw$. By cross-checking this player's answers against the other player (who may have received either $\introspect_W$ or $\sample_\cw$ alone), the game ensures that the measurements corresponding to $\introspect_W$ and $\sample_\cw$ \emph{commute}, and this in turn enables the ``honest'' strategy in the completeness case to be oracularizable. This and more will become clear in the next subsection.

\subsection{Completeness of Introspection}
\label{sec:introspection-completeness}
As mentioned earlier, we need to show that the value of the original game and the introspected game are tightly related. This has two directions. First we need to show that if $G$ has a perfect strategy so does $G^\intro$; this is called the \emph{completeness} property. In fact we prove the following stronger statement. 

\begin{proposition}[Completeness of Introspection]\label{prop:introspection-completeness}
For all oracularizable synchronous strategies $\strategy$ for $G$, there exists an oracularizable synchronous strategy $\strategy^\intro$ for $G^\intro$ such that 
	\[
		\omega(G^\intro,\strategy^\intro) \geq \omega(G,\strategy)~.
	\] 
	Furthermore, if $\strategy$ is finite-dimensional then so is $\strategy^\intro$.
\end{proposition}

Recall that a synchronous strategy $\strategy$ for a synchronous game $G$ is oracularizable if for every nontrivial question pair $(q,r)$, the corresponding measurement operators commute (see \Cref{def:oracularizable-strategy}).

\begin{proof}
Let $\strategy = (\sigma,\{N^x\}_{x\in\cal{X}})$ be an oracularizable synchronous strategy for $G$ and let $\strategy_{\qs_\ell} = (\tau,\{M^q\}_{q\in \cal{Q}_{\qs_\ell}})$ be the ``honest'' perfect oracularizable strategy for the Question Sampling game $\qs_\ell$ as defined in \Cref{sec:question-sampling-game}. Let $\hilb_{\qs_\ell}$, $\hilb_\strategy$ and $\algebra_{\qs_\ell} \subseteq \B(\hilb_{\qs_\ell}),\algebra_{\strategy} \subseteq \B(\hilb_\strategy)$ denote the Hilbert spaces and algebras of the two strategies, respectively. We define a synchronous strategy $\strategy^\intro = (\rho,\{P^q\}_{q \in \cal{Q}^\intro} )$, which we call the \emph{honest Introspection strategy}, for $G^\intro$ over the algebra $\algebra_{\qs_\ell}\otimes \algebra_\strategy$ with the tracial state $\rho = \tau\otimes \sigma$. In this proof we use the shorthand notation $S^W_x, E^W_x$ to denote the operators $M^{\sample_W}_x$, $M^{\erase_W}_x$ from the strategy $\strategy_{\qs_\ell}$, respectively.

The measurement operators are defined as follows. For all $q \in \cal{Q}_{\qs_\ell}$ and $x \in \cal{A}_{\qs_\ell}$, let $P^q_x = M^q_x \otimes \id$ where the $\id$ denotes the identity on the Hilbert space $\hilb_\strategy$. Since $M^q_x$ is a projection on $\hilb_{\qs_\ell}$, the operators $\{P^q_x\}$ are also projections and furthermore form a measurement. 

For all other questions $q \in \cal{Q}^\intro \setminus \cal{Q}_{\qs_\ell}$, we define
\begin{gather*}
P^{\introspect_W}_{x,a} \coloneqq S^W_x \otimes N_a^x~, \qquad \qquad P^{\introspect_W\sample_\cw}_{x,a,y} \coloneqq S^W_x \sample^\cw_y \otimes N_a^x~, \qquad \qquad P^{\introspect_W\erase_\cw}_{x,a,y} \coloneqq S^W_x \erase^\cw_y \otimes N_a^x
\end{gather*}
for all $W \in \{ \alice, \bob \}$, $x,y \in \cal{X}$, and $a\in \cal{A}$. The operator $P^{\introspect_W}_{x,a}$ is clearly a projection (because $S^W_x, N^x_a$ are projections), and forms a projective measurement. In the honest Question Sampling strategy the operators $S^W_x$ and $\sample^\cw_y$ commute (by \Cref{thm:rigidity-of-question-sampling}), therefore $P^{\introspect_W\sample_\cw}_{x,a,y}$ forms a projective measurement. Similarly $S^W_x$ and $\erase^\cw_y$ commute, therefore $P^{\introspect_W\sample_\cw}_{x,a,y}$ forms a projective measurement. 

It should be clear now why we choose the notation $\introspect_W\sample_\cw$ and $\introspect_W\erase_\cw$: in the honest Introspection strategy, we have that 
\begin{equation}
\label{eq:composite-introspection-measurements}
	P^{\introspect_W\sample_\cw}_{x,a,y} = P^{\introspect_W}_{x,a} \, \, \sample^\cw_y = \sample^\cw_y \,\,P^{\introspect_W}_{x,a}   \qquad \text{ and } \qquad 
	P^{\introspect_W\erase_\cw}_{x,a,y} = P^{\introspect_W}_{x,a} \,\, \erase^\cw_y = \erase^\cw_y \,\, P^{\introspect_W}_{x,a}~.
\end{equation}

It remains to define the projective measurement $\{P^\introspect_{x,a,y,b}\}$ for the Introspection question $\introspect$. If $(x,y) \in\cal{X}\times\cal{X}$ is a nontrivial question in $G$, we define \[P^\introspect_{x,a,y,b} \coloneqq \sample^\alice_x \,\, \sample^\bob_y \otimes N_a^{x} \,\, N_b^y.\] Since $N_a^x$ and $N_b^y$ commute when $(x,y)$ is nontrivial for $G$ (because $\strategy$ is oracularizable), we see that $P^\introspect_{x,a,y,b}$ is a projection. If on the other hand $(x,y)$ is a trivial question in $G$, we define \[P^\introspect_{x,a,y,b} \coloneqq \left\{
	\begin{array}{ll}
		\sample^\alice_x \,\, \sample^\bob_y \otimes \id  & \mbox{if $(a,b)=(0^m,0^m)$},  \\
		0 & \mbox{otherwise}.
	\end{array}
    \right.\]
This is clearly a projective measurement as well. Intuitively, when a player receives the question $\introspect$, they first perform the sampling measurements $\sample^\alice$ and $\sample^\bob$ (which can be performed simultaneously since they commute) to obtain a pair of questions $(x,y) \in \cal{X} \times \cal{X}$ for the original game $G$. If $(x,y)$ is trivial for $G$, then the player outputs $(x,0^m,y,0^m)$. Otherwise, the player then simultaneously measures $N^x$ and $N^y$ (which commute since $(x,y)$ is nontrivial for $G$) to obtain answers $(a,b) \in \cal{A} \times \cal{A}$. The player then returns $(x,a,y,b)$ as its answer. 

Clearly $\strategy^\intro$ is finite-dimensional when $\strategy$ is finite-dimensional. Next we show that $\strategy^\intro$ is oracularizable and has success probability $1$ in the Introspection game $G^\intro$.

First, if $(q,r)$ is a trivial pair of questions for $G^\intro$ then by definition the players win with probability $1$ on those questions. Assume that $(q,r)$ is a nontrivial question pair.

Suppose that $(q,r) \in \cal{Q}_{\qs_\ell}$. Since $\strategy_{\qs_\ell}$ is oracularizable and $(q,r)$ must also be nontrivial for $\qs_\ell$, the measurement operators $\{P^q_x\}$ and $\{P^r_x\}$ commute. Furthermore, by design the strategy $\strategy_{\qs_\ell}$ succeeds with probability $1$ in the game $\qs_\ell$ and thus succeeds with probability $1$ in $G^\intro$ conditioned on questions from $\cal{Q}_{\qs_\ell}$.

It remains to check the commutativity property and success probability for all question pairs that are connected via an edge in Figure \ref{fig:graph-of-introspection-game}. For self-loops (i.e, question pairs $(q,q)$), commutativity and success probability $1$ are trivially satisfied because the operators $P^q_{\hat{a}}$ are projections. We now check the other nontrivial question pairs.

\medskip

\noindent \underline{$(\introspect_W,\sample_W)$}: Commutativity follows because 
\[P^{\introspect_W}_{x,a} \,\, P^{\sample_W}_{z} = \sample^W_{x} \, \sample^W_{z} \otimes N_{a}^{x} = \sample^W_{z} \, \sample^W_{x} \otimes N_{a}^{x} = P^{\sample_W}_{z} \,\, P^{\introspect_W}_{x,a}~.
\] 
Here we used the fact that $S^W_x, S^W_z$ are elements of the same projective measurement and thus commute. 
The probability of winning conditioned on this question pair is 
\[
\sum_{x,a} \rho(P^{\introspect_W}_{x,a} \,\, P^{\sample_W}_x) = \sum_{x,a} \tau(S^W_x \,\, S^W_x ) \,\, \sigma(N_a^x) = \sum_{x} \tau(S^W_x ) = 1~.
\]

\medskip

\noindent \underline{$(\introspect_W,\introspect_W\sample_\cw)$}: Commutativity follows because 
\[
P^{\introspect_W}_{x,a} \,\, P^{\introspect_W\sample_\cw}_{z,c,y} = \sample^W_{x} \,\, \sample^W_{z} \,\, \sample^\cw_y \otimes N_{a}^{x} \,\, N_{c}^{z} = \sample^W_{z} \,\,\sample^\cw_y\,\, \sample^W_{x} \otimes N_c^z\,\, N_{a}^{x} = P^{\introspect_W \sample_\cw}_{z,c,y} \,\, P^{\introspect_W}_{x,a}.
\]
    The second equality holds because if $x \neq z$, then $\sample^W_{x}\,\, \sample^W_{z}= 0$ and the equality holds trivially. If on the other hand $x = z$, the equality holds because $S^W_x, \sample^\cw_y$ commute with each other and $N_{a}^x, N_c^x$ commute with each other.
    
The probability of winning conditioned on this question pair is 
\[
\sum_{x,a,y} \rho(P^{\introspect_W}_{x,a}\,\, P^{\introspect_W\sample_\cw}_{x,a,y}) = \sum_{x,a,y} \rho(P^{\introspect_W}_{x,a}\,\, P^{\introspect_W}_{x,a} \,\, \sample^W_y ) = \sum_{x,a} \rho(P^{\introspect_W}_{x,a} ) = 1
\]
where in the first equality we used~\eqref{eq:composite-introspection-measurements}.

\medskip

\noindent \underline{$(\introspect_W,\introspect_W\erase_\cw)$}: The argument for this is nearly identical to that for the previous question pair, except we replace the sampling measurement $\sample^\cw$ with the erasure measurement $\erase^\cw$.

\medskip
\noindent \underline{$(\introspect_W\sample_\cw,\sample_\cw)$}: Commutativity follows because 
\[
P^{\introspect_W\sample_\cw}_{x,a,y} \,\, \sample^\cw_z = P^{\introspect_W}_{x,a} \,\, \sample^\cw_y \,\, \sample^\cw_z = \sample^\cw_z \,\, P^{\introspect_W}_{x,a} \,\, \sample^\cw_y = \sample^\cw_z \,\, P^{\introspect_W\sample_\cw}_{x,a,y}
\]
where in the first equality we used~\eqref{eq:composite-introspection-measurements}, and then we used the fact that $\sample^\cw_z$ commute with $P^{\introspect_W}_{x,a}$. 

The probability of winning conditioned on this question pair is 
\[
\sum_{x,a,y} \rho(P^{\introspect_W \sample_\cw}_{x,a,y}\,\, \sample^\cw_y) = \sum_{x,a,y} \rho(P^{\introspect_W}_{x,a} \,\, \sample^\cw_y \,\, \sample^\cw_y) = \sum_{x,a} \rho(P^{\introspect_W}_{x,a} ) = 1
\]
where in the first equality we used~\eqref{eq:composite-introspection-measurements} and in the second equality we used the fact that $\sample^\cw_y$ is a projection and forms a measurement.

\medskip
\noindent \underline{$(\introspect_W\erase_\cw,\erase_\cw)$}: The argument for this is identical to that for the previous question pair, except we replace the sampling measurement $\sample^\cw$ with $\erase^\cw$.

\medskip
\noindent \underline{$(\introspect,\introspect_W)$}: Assume without loss of generality that $W = \alice$. Commutativity is due to the following. Suppose $(x,y)$ is a trivial question pair for $G$. Then 
    \[
    P^{\introspect}_{x,0,y,0} \,\, P^{\introspect_\alice}_{z,c} = \sample^\alice_x \,\, \sample^\bob_y \,\, \sample^\alice_z \otimes N^z_c = \sample^\alice_z \,\, \sample^\alice_x \,\, \sample^\bob_y \otimes N^z_c =  P^{\introspect_\alice}_{z,c} \,\,  P^{\introspect}_{x,0,y,0}
    \]
    where $0$ is shorthand for $0^m$, and for all $(a,b) \neq (0^m,0^m)$ we have
    \[
    	P^{\introspect}_{x,a,y,b} \,\, P^{\introspect_\alice}_{z,c} = 0 = P^{\introspect_\alice}_{z,c}\,\, P^{\introspect}_{x,a,y,b}~.
	\]
    If $(x,y)$ is a nontrivial question pair for $G$ then
    \[
    	P^{\introspect}_{x,a,y,b} \,\, P^{\introspect_\alice}_{z,c} = \sample^\alice_x \,\, \sample^\bob_y \,\, \sample^\alice_z \otimes N^x_a \,\, N^y_b \,\, N^z_c = \sample^\alice_z \,\, \sample^\alice_x \,\, \sample^\bob_y \otimes N^z_c \,\, N^x_a \,\, N^y_b = P^{\introspect_\alice}_{z,c} \,\, P^{\introspect}_{x,a,y,b}
	\]	
    where the second equality holds because if $x \neq z$, then $\sample^\alice_x \,\, \sample^\bob_y \,\,  \sample^\alice_z = 0$ and the equality holds trivially. If on the other hand $x=z$, the equality holds because $N^x_a, N^y_b, N^x_c$ all commute (because $(x,y)$ is a nontrivial question pair and $N^x_a, N^x_c$ are elements of the same projective measurement).
    
    We calculate the probability of success as follows. If $(x,y)$ is a nontrivial question pair in the original game $G$ we have 
    \[
    	\rho(P^{\introspect}_{x,a,y,b} \,\, P^{\introspect_\alice}_{z,c}) = \tau(\sample^\alice_x \,\, \sample^\bob_y \,\, \sample^\alice_z ) \,\, \sigma(N^x_a \,\, N^y_b \,\, N^z_c) = 2^{-2\ell} \,\, \sigma(N^x_a \,\, N^y_b) \,\, \mathbf{1}_{z = x, c = a}
	\] 
	where we used the fact that in the honest strategy $\strategy_{\qs_\ell}$ we have $\tau(\sample^\alice_x \,\, \sample^\bob_y) = 2^{-2\ell}$. Notation $\mathbf{1}_{z = x, c = a}$ denotes the indicator variable for the equalities $z = x, c = a$. If $(x,y)$ is trivial we have 
	\[
		\rho(P^{\introspect}_{x,a,y,b} \,\, P^{\introspect_\alice}_{z,c}) = 2^{-2\ell} \,\, \sigma (N^z_c) \,\, \mathbf{1}_{z = x, a = b = 0^m}~.
	\]	
    
    So the probability of winning using $\strategy^\intro$ conditioned on players receiving question pair $(\introspect,\introspect_\alice)$ is 
    \begin{align*}
    &\sum_{x,a,y,b,z,c} \rho(P^{\introspect}_{x,a,y,b} \,\, P^{\introspect_\alice}_{z,c}) \,\,  D^\intro(\introspect,\introspect_\alice, (x,a,y,b),(z,c))\\ 
    &\qquad\qquad= \frac{1}{2^{2\ell}}\sum_{\substack{(x,y)\\\text{nontrivial for $G$}}} \,\, \sum_{a,b} \sigma(N^x_a \,\, N^y_b ) \,\, D(x,y,a,b) + \frac{1}{2^{2\ell}} \sum_{\substack{(x,y)\\\text{trivial for $G$}}} \,\, \sum_c \sigma(N^x_c) \\
    &\qquad\qquad= \frac{1}{2^{2\ell}}\sum_{\substack{(x,y)\\\text{nontrivial for $G$}}} \,\, \sum_{a,b} \sigma(N^x_a \,\, N^y_b ) \,\, D(x,y,a,b) + \frac{1}{2^{2\ell}} \sum_{\substack{(x,y)\\\text{trivial for $G$}}} \,\, 1 \\    
    &\qquad\qquad= \frac{1}{2^{2\ell}}\sum_{\substack{(x,y)\\\text{nontrivial for $G$}}} \,\, \sum_{a,b} \sigma(N^x_a \,\, N^y_b ) \,\, D(x,y,a,b) + \frac{1}{2^{2\ell}} \sum_{\substack{(x,y)\\\text{trivial for $G$}}} \,\, \sum_{a,b} \sigma(N^x_a \,\, N^y_b) \,\, D(x,y,a,b) \\
    &\qquad\qquad= \val(G,\strategy)   
    \end{align*}
where in the third line we used that $\{N^x_c\}$ is a measurement, and in the fourth line we used that $D(x,y,a,b) = 1$ for all trivial $(x,y)$. 

So conditioned on any pair of questions the players win with probability $1$ using strategy $\strategy^\intro$, except when they receive question pair $(\introspect,\introspect_\alice)$ or $(\introspect,\introspect_\bob)$ in which case they win with probability $\val(G,\strategy)$. From this we conclude that $\val(G^\intro,\strategy^\intro) \geq \val(G,\strategy)$.
\end{proof}

\subsection{Soundness of Introspection}
\label{sec:introspection-soundness}
The second part of showing that the value of the original game and the introspected game are tightly related is called \emph{soundness}. Informally speaking the soundness property states that if the original game has no perfect strategy, then neither does the introspected game. 

In the soundness proposition below, we also prove a lower bound on the dimension of the Hilbert space for any perfect strategy of $G^\intro$. We show this dimension is at least as big as the maximum of $2^{2\ell}$ and the smallest dimension of a Hilbert space among all perfect strategies of $G$. Recall that $\ell$ is the bit length of questions in $G$. This dimension lower bound will be used later in the section on compression. 

\begin{proposition}[Soundness of Introspection]\label{prop:introspection-soundness}
For all $t \in \{q,co\}$
\[\val_t^s(G^\intro) = 1 \implies \val_t^s(G) = 1.\]
Furthermore it holds that
    \[\mathcal{E} (G^{\intro},1) \geq \max \left \{ \mathcal{E} (G,1) , 2^{2\ell} \right\}.\]
\end{proposition}

At a high level, the proof of \Cref{prop:introspection-soundness} proceeds by taking a synchronous strategy $\strategy^\intro = (\rho, \{P^q\}_{q \in \cal{Q}^\intro})$ for $G^\intro$ that succeeds with probability $1 - \eps$, with $\eps$ sufficiently small, and ``extracting'' from it a strategy $\strategy = (\sigma,\{N^x\}_{x \in \cal{X}})$ for the original game $G$ that has value $1 - \delta(\eps)$ where $\delta$ is a proper error function (see \Cref{subsec:asymptotics} for definition of proper error function). The error function $\delta$ also has a dependence on $\ell$, but since we do not need to carry that around, we hide it in our notation $\delta(\eps)$.

Note that $\val_q^s(G^\intro) = 1$ does not imply the existence of a finite-dimensional synchronous strategy with value $1$. All we can guarantee is that for every $\eps > 0$ there exists a finite-dimensional synchronous strategy with value at least $1-\eps$. On the other hand $\val_{co}^s(G^\intro) = 1$ means that there exists a perfect synchronous strategy for $G^\intro$.

To make the notation easier to read, we use the following abbreviations for the measurements $P^q$ corresponding to the questions $q \in \{ \,\, \introspect, \,\, \introspect_W, \,\, \introspect_W \sample_\cw, \,\, \introspect_W \erase_\cw, \,\, \sample_W, \,\, \erase_W \,\, \}_{W \in \{\alice,\bob\}} \subseteq \cal{Q}^\intro$. For all $W \in \{\alice,\bob\}$, $x,y \in \cal{X}$ and $a,b \in \cal{A}$,
\begin{gather*}
	\introspect_{x,a,y,b} = P^{\introspect}_{x,a,y,b}~, \qquad \qquad \introspect^W_{x,a} = P^{\introspect_W}_{x,a}~, \qquad \qquad (\introspect^W \sample^\cw)_{x,a,y} = P^{\introspect_W \sample_\cw}_{x,a,y} \\
	(\introspect^W \erase^\cw)_{x,a,y} = P^{\introspect_W \erase_\cw}_{x,a,y}~, \qquad \qquad S^W_x = P^{\sample_W}_x~, \qquad \qquad \erase^W_x = P^{\erase_W}_x~.
\end{gather*}
Furthermore, we define the \emph{erasure observables}
\[
	O^W_x = \sum_{y \in \cal{X}} (-1)^{x \cdot y} \, \erase^W_y 
\]
for $W \in \{\alice,\bob\}$. Unlike the section on Question Sampling, we do not need to define sampling observables for the purpose of proving the current proposition. We use $\cdot$ in the subscript to indicate the data-processed measurement that ignores part of the measurement outcome, so for example
\begin{align*}
    \introspect_{\cdot,a,y,b} &= \sum_{x\in \cal{X}} \introspect_{x,a,y,b},\\ \introspect_{x,\cdot,y,b} &= \sum_{a\in \cal{A}} \introspect_{x,a,y,b},\\ \introspect_{x,a,\cdot,\cdot} &= \sum_{y \in \cal{X},b \in \cal{A}} \introspect_{x,a,y,b}, 
\end{align*}
etc. We may sometime drop $\cdot$ when there is no risk of ambiguity, for example we may write $\introspect^W_x$ instead of $\introspect^W_{x,\cdot}$.

We first prove two key lemmas establishing that in any strategy with large value certain commutation relations are approximately satisfied and that introspected questions are almost uniformly sampled. Throughout this section, we let $\strategy^\intro = (\rho, \{P^{q}\}_{q\in \cal{Q}^\intro})$ be a fixed synchronous strategy for $G^\intro$ with value $1-\eps$.

\begin{lemma}\label{lem:approximate-commutation}
The following approximate relations hold
\begin{align*}
    \introspect^{W}_x &\approx \sample^{W}_x\\
    \introspect^{W}_{x,a} \, \sample^{W}_{y} &\approx \sample^{W}_{y} \, \introspect^{W}_{x,a}\\
    \introspect^{W}_{x,a} \, \sample^{\cw}_{y} &\approx \sample^{\cw}_{y} \, \introspect^{W}_{x,a}\\
    \introspect^{W}_{x,a} \, \erase^{\cw}_{y} &\approx \erase^{\cw}_{y} \, \introspect^{W}_{x,a}\\
    \introspect^{W}_{x,a} \, O^{\cw}_{u} &\approx O^{\cw}_{u} \, \introspect^{W}_{x,a}.
\end{align*}
\end{lemma}
\begin{proof}
As mentioned in Section \ref{subsec:asymptotics}, when we write $\introspect^{W}_x \approx \sample^{W}_x$ we mean $\introspect^{W}_x \approx_{\delta(\eps)} \sample^{W}_x$ for some function $\delta$ such that $\delta(\eps) \to 0$ as $\eps \to 0$.

Since the strategy is winning with probability $1-\eps$, the winning probability conditioned on receiving question $(\introspect_W,\sample_W)$ is at least $1 - |\cal{Q}^\intro|^2\eps$. The expression for the probability of winning conditioned on players receiving question pair $(\introspect_W,\sample_W)$ is
\begin{align*}
\sum_{x,a,y} \rho(\introspect^W_{x,a} \, \sample^W_y)D^\intro(\introspect_W,\sample_W,(x,a),y) &= \sum_{x,a} \rho(\introspect^W_{x,a} \, S^W_x)\\
&= \sum_{x} \rho(\introspect^W_{x} \, \sample^W_x).
\end{align*}
Therefore we have
\[\sum_x \rho(\introspect^{W}_x \, \sample^W_x ) \approx 1,\]
or equivalently that $\introspect^{W}_x \simeq \sample^W_x $.
By Lemma \ref{lem:consistency-consequences}, we get that $\introspect^{W}_x \approx \sample^{W}_x$. By Proposition \ref{prop:exchange-operators}, we obtain that $\introspect^{W}_{x,a} \, \sample^{W}_{y} \approx \introspect^{W}_{x,a} \, \introspect^{W}_{y}$ from which we arrive at our first approximate commutation relation
\[\introspect^{W}_{x,a} \, \sample^{W}_{y} \approx \introspect^{W}_{x,a} \, \introspect^{W}_{y} = \introspect^{W}_{y} \, \introspect^{W}_{x,a} \approx \sample^{W}_{y} \, \introspect^{W}_{x,a}\]
where the equality in the middle follows because operators belonging to the same projective measurement commute. This is the basic idea behind the proof of the remaining approximate relations.

Next we prove the approximate commutation relation $\introspect^{W}_{x,a} \, \erase^{\cw}_{y} \approx \erase^{\cw}_{y} \, \introspect^{W}_{x,a}$ (the relation $\introspect^{W}_{x,a} \, \sample^{\cw}_{y} \approx \sample^{\cw}_{y} \, \introspect^{W}_{x,a}$ is proved nearly identically). Similar to our argument above for $(\introspect_W,\sample_W)$, the players winning probability conditioned on receiving question pair $(\erase_\cw,\introspect_W\erase_\cw)$ is $1-\delta(\eps)$, that is
\[\sum_y \tau(\erase^{\cw}_y (\introspect^W\erase^\cw)_y) \approx 1\]
from which, similar to the argument above, we arrive at $\erase^{\cw}_y \approx (\introspect^W\erase^\cw)_y$. With a similar argument, this time starting from the winning probability conditioned on question pair $(\introspect_W,\introspect_W\erase_\cw)$, we get that $\introspect^{W}_{x,a} \approx (\introspect^W\erase^\cw)_{x,a}$. Putting these together we obtain
\begin{align*}
    \introspect^{W}_{x,a} \, \erase^{\cw}_{y} &\approx (\introspect^W\erase^\cw)_{x,a} (\introspect^W\erase^\cw)_{y}\\
    &= (\introspect^W\erase^\cw)_{y}(\introspect^W\erase^\cw)_{x,a}\\
    &\approx \erase^{\cw}_{y} \, \introspect^{W}_{x,a}.
\end{align*}
Finally the last approximate commutation relation follows
\begin{align*}
    \introspect^{W}_{x,a} \, O^{\cw}_{u} &= \sum_{y\in \cal{X}} (-1)^{y.u}\introspect^{W}_{x,a} \, \erase^{\cw}_{y} \\
    &\approx \sum_{y\in \cal{X}} (-1)^{y.u}\erase^{\cw}_{y} \, \introspect^{W}_{x,a}\\
    &= O^{\cw}_{u} \, \introspect^{W}_{x,a}.
\end{align*}
Switching the order of multiplication in $\introspect^{W}_{x,a} \, \erase^{\cw}_{y}$ incurs an error of $\delta(\eps)$ for each $x,a,y$. So over all the norm of $\sum_{y\in \cal{X}} (-1)^{y.u}\introspect^{W}_{x,a} \, \erase^{\cw}_{y} - \sum_{y\in \cal{X}} (-1)^{y.u}\erase^{\cw}_{y} \, \introspect^{W}_{x,a}$ is bounded above by $|\cal{X}\times \cal{A}\times \cal{X}| \delta(\eps)$ which is another error function $\delta(\eps)$.
\end{proof}

Next lemma establishes that the introspected questions are sampled almost uniformly from the question set of the original game. We then use this to justify that $\introspect_{x,a,y,b}$ is approximately $\introspect^{\alice}_{x,a}\introspect^{\bob}_{y,b}$ when $x,y$ is a nontrivial question pair in the original game. 

\begin{lemma}\label{lem:approximately-uniform-sampling}
Let $\introspect_{x,y} = \introspect_{x,\cdot,y,\cdot}$. Then the following hold 
\begin{gather*}
\introspect_{x,y} \approx \sample^{\alice }_x \, \sample^{\bob }_{y},\\
\rho(\introspect_{x,y}) \approx \frac{1}{2^{2\ell}}.
\end{gather*}
Furthermore, if $x,y$ is a nontrivial question pair in the original game, then for every $a,b \in \cal{A}$ \[\introspect_{x,a,y,b} \approx \introspect^{\alice}_{x,a} \, \introspect^{\bob}_{y,b}.\]
\end{lemma}
\begin{proof}
The players winning probability conditioned on receiving question pair $(\introspect,\introspect_\alice)$ is $1-\delta(\eps)$. So $\sum_{x} \rho(\introspect_{x,\cdot,\cdot,\cdot} \, \introspect^{\alice}_x) = 1-\delta(\eps)$ where $\introspect^\alice_x = \sum_{a} \introspect^\alice_{x,a}$. Therefore $\introspect_{x,\cdot,\cdot,\cdot} \approx \introspect^{\alice}_x$ and consequently $\introspect_{x,\cdot,\cdot,\cdot} \approx \sample^{\alice }_x$ by \Cref{lem:consistency-consequences}. Similarly $\introspect_{\cdot,y,\cdot,\cdot} \approx \introspect^{\bob}_{y} \approx \sample^{\bob }_{y}$. Thus we have $\introspect_{x,y} = \introspect_{x,\cdot,\cdot,\cdot} \, \introspect_{\cdot,y,\cdot,\cdot} \approx \sample^{\alice }_x \, \sample^{\bob }_{y}$. By \Cref{thm:rigidity-of-question-sampling} and \Cref{prop:exchange-operators}, we conclude that $\rho(\introspect_{x,y}) \approx \frac{1}{2^{2\ell}}$.

So far we established that any question pair $(x,y)$ in the answer to the Introspection question $\introspect$ occurs almost uniformly, that is with probability approximately $1/2^{2\ell}$. Fix a nontrivial question pair $x,y$ in the original game. The probability of the event that players receive question pair $(\introspect,\introspect^A)$ and respond with $(x,a,y,b)$ and $(z,c)$, respectively, for some $a,b,c \in \cal{A}$ and $z \in \cal{X}$ is at least $(1-\delta(\eps))2^{-2\ell}/|\cal{Q}^\intro|^2$. Since the overall strategy looses with probability at most $\eps$, the probability of loosing conditioned on this event is bonded above by $$2^{2\ell}|\cal{Q}^\intro|^2\eps/(1-\delta(\eps)) \leq 2^{2\ell}|\cal{Q}^\intro|^2(1+\delta(\eps))\eps =\delta(\eps)$$ or in other words the probability of winning conditioned on this event is $1- \delta(\eps)$.
It is now a simple exercise in probability theory to see that conditioned on receiving question $(\introspect,\introspect_A)$, the probability that player receiving $\introspect$ answers with introspected questions $(x,y)$ and the players win is $\approx 2^{-2\ell}$.

By the construction of the Introspection game, if the players win, then it must be that $(z,c) = (x,a)$. Therefore we have \[\sum_{a} \rho(\introspect_{x,a,y,\cdot} \, \introspect^{\alice}_{x,a}) = \sum_{a,b}\rho(\introspect_{x,a,y,b} \, \introspect^{\alice}_{x,a}) \approx 2^{-2\ell}.\] Using the relation $\introspect_{y} \approx \sample^{\bob }_{y}$ that we proved earlier together with the approximate commutations in \Cref{lem:approximate-commutation}, we obtain 
\begin{align}\label{eq:R-S-approximation}
\sum_a \rho(\introspect_{x,a,y,\cdot}\paren{\sample^{\bob }_{y} \,\, \introspect^{\alice}_{x,a} \,\, \sample^{\bob }_{y}}) \approx \sum_a \rho(\introspect_{x,a,y,\cdot}\paren{\introspect_{y} \,\, \introspect^{\alice}_{x,a} \,\, \introspect_{y}}) =  \sum_a \rho(\introspect_{x,a,y,\cdot} \,\, \introspect^{\alice}_{x,a}) \approx 2^{-2\ell}.
\end{align}
Define positive semidefinite operators $R_a = \introspect_{x,a,y,\cdot}$ and $S_a  = \sample^{\bob }_{y} \introspect^{\alice}_{x,a}\sample^{\bob }_{y}$, and write
\begin{align*}
    \sum_{a} \| R_a - S_a \|_\rho ^2 &= \sum_{a} \rho(R_a^2 + S_a ^2 - 2R_a \, S_a )\\
                                &\leq \sum_{a} \rho(R_a + S_a  - 2R_a S_a )\\
                                &= \sum_{a} \rho(R_a) + \rho(S_a ) - 2\rho(R_a \, S_a )\\
                                &\leq 2(1+\delta(\eps))2^{-2\ell} - 2(1-\delta(\eps))2^{-2\ell}\\
                                &= \delta(\eps).
\end{align*}
The first inequality follows from the fact that $R_a,S_a$ are positive semidefinite with operator norm $\leq 1$. The last inequality follows from $\rho(\sum_a R_a S_a ) \approx 2^{-2\ell}$ which we proved in (\ref{eq:R-S-approximation}) and the following two calculations
\begin{align*}
    \rho(\sum_a R_a) &= \rho(\introspect_{x,y}) \approx 2^{-2\ell},\\
    \rho(\sum_a S_a ) &= \rho(\sample^{\bob }_{y} \introspect^{\alice}_{x} \,\, \sample^{\bob }_{y}) = \rho(\introspect^{\alice}_{x} \,\, \sample^{\bob }_{y}) \approx \rho(\sample^{\alice }_{x} \,\, \sample^{\bob }_{y}) \approx 2^{-2\ell}.
\end{align*}
We conclude that $\introspect_{x,a,y,\cdot} \approx \sample^{\bob }_{y} \,\, \introspect^{\alice}_{x,a}\ \,\, \sample^{\bob }_{y}\approx \introspect^{\alice}_{x,a} \,\, \sample^{\bob }_{y}$. By a similar argument we get that $$\introspect_{x,\cdot,y,b} \approx \introspect^{\bob}_{y,b} \,\, \sample^{\alice }_x.$$ Putting these two together $$\introspect_{x,a,y,b} = \introspect_{x,a,y,\cdot} \,\, \introspect_{x,\cdot,y,b} \approx \introspect^{\alice}_{x,a} \,\, \sample^\bob_y \,\, \introspect^{\bob}_{y,b} \,\, \sample^\alice_x \approx \introspect^{\alice}_{x,a} \,\, \sample^\alice_x \,\, \introspect^{\bob}_{y,b} \,\,\sample^\bob_y =  \introspect^{\alice}_{x,a} \,\, \introspect^\alice_x \,\, \introspect^{\bob}_{y,b} \,\, \introspect^\bob_y = \introspect^{\alice}_{x,a} \,\, \introspect^{\bob}_{y,b}.$$
\end{proof}

 We first sketch a proof of \Cref{prop:introspection-soundness}. The key step is to establish that, in any strategy that wins with high probability in $G^\intro$, when players $A$ and $B$ receive questions $\introspect_\alice$ and $\introspect_\bob$, respectively, their answers $(x_\alice,a_\alice)$ and $(x_\bob,a_\bob)$ are such that $(x_\alice,x_\bob)$ is uniformly distributed in $\cal{X}\times \cal{X}$ and $a_\alice$ has no dependence on $x_{B}$ and similarly $a_\bob$ has no dependence on $x_\alice$. In other words players introspectively asked themselves a uniformly random question $(x_\alice,x_\bob)$ and produced answers $(a_\alice,a_\bob)$ as they would have answered if they received question $(x_\alice,x_\bob)$ in the original game. 

In \Cref{lem:approximate-commutation}, we proved that $\introspect^{W}_x \approx \sample^{W}_x$. This relation implies that on question $\introspect_W$ the player effectively obtains $x_W$ part of the answer by measuring $\{\sample^{W}_x\}$. So, by the rigidity properties of the Question Sampling game, we get that $(x_a,x_b)$ is sampled (almost) uniformly at random from $\cal{X}\times \cal{X}$. We also showed in \Cref{lem:approximately-uniform-sampling} that $(x_a,x_b)$ in answer to question $\introspect$ are also distributed (almost) uniformly. From the rigidity properties of the Question Sampling game, measurements $\sample^{W}$ and $\erase^{W}$ (approximately) anticommute while they both (approximately) commute with measurements $\sample^{\cw}$ and $\erase^{\cw}$. Additionally we saw in \Cref{lem:approximate-commutation} that $\introspect^W$ commutes with both $\sample^\cw$ and $\erase^\cw$. These relationships intuitively imply that the Hilbert space $\hilb$ can be (approximately) divided into a tensor product $\hilb_\alice\otimes \hilb_\bob \otimes \hilb_G$ of three Hilbert spaces such that the players measurements for special questions $\sample_W$ and $\erase_W$ are forced to act as identity on $\hilb_\cw$. Furthermore, the commutation of $\introspect^W$ with $\sample^\cw$ and $\erase^\cw$ implies that operators $\introspect^W$ act trivially on the register $\hilb_\cw$. Now since $x_\cw$ is obtained by a measurement on $\hilb_\cw$ we conclude that $a_W$ has no dependence on $x_\cw$.

Putting these together, we get that the player with question $\introspect_W$ produces $x_W$ via a measurement on $\hilb_W$, then produces $a_W$ with a measurement that depends on $x_W$ and has a nontrivial support only on the game register $\hilb_G$. In other words $\introspect^{W}_{x,a} = \sample^{W}_x\otimes N^x_a$ for some $N^x_a$ that acts as identity on $\hilb_\cw$. We can now let $\{N^x_a\}$ be the measurements in a strategy in the original game $G$ and show that its value is large. In what follows we make this argument precise.

\begin{proof}[Proof of \Cref{prop:introspection-soundness}]
Let $\strategy^\intro = (\rho, \{P^{q}\}_{q\in \cal{Q}^\intro})$ be a synchronous strategy for $G^\intro$ that has value at least $1 - \eps$. Let $\hat{\hilb},\Pi, \hat{\algebra}, \sigma$ be as defined in \Cref{lem:new-algebra}. 

For every $W \in \{\alice,\bob\}$, $x \in \cal{X}$ and $a \in \cal{A}$ define the operator
\begin{align*}
    \tilde{W}_a^x \coloneqq O^W_x \,\, I^W_{x,a} \,\, O^W_x~.
\end{align*}
Note that for every $W\in \{A,B\}$ and $x\in \cal{X}$ the operators $\{\tilde{W}_a^x\}_{a\in \cal{A}}$ are pairwise orthogonal projections. For every $x\in \cal{X}$ define the \emph{leftover} operator
\begin{align*}
	\tilde{W}_\perp^x \coloneqq \id - \sum_{a \in \cal{A}} \tilde{W}^x_a~.
\end{align*}
Let $\tilde{\cal{A}} = \cal{A}\cup \{\perp\}$ denote the expanded answer set. Then $\{\tilde{W}_a^x\}_{a\in \tilde{\cal{A}}}$ is a projective measurement for every $W\in \{\alice,\bob\}, x \in \cal{X}$. 

Now for every $x \in \cal{X}, a\in\tilde{\cal{A}}$ define 
\begin{align*}
	\hat{W}^x_a \coloneqq \Pi \,\, \tilde{W}^x_a \,\, \Pi~.
\end{align*}
These are clearly positive semidefinite operators and 
\[
	\sum_{a\in \tilde{A}} \hat{W}_a^x = \Pi \, \Big( \sum_{a\in \tilde{A}} \tilde{W}_a^x \Big ) \, \Pi= \Pi^2 = \Pi~.
\]
Since $\Pi$ is projection onto $\hat{\hilb}$, the set of operators $\{\hat{W}_a^x\}_{a\in \tilde{A}}$ are POVMs on $\hat{\hilb}$ for every $x$. 

Our first goal is to show that for every $x,y\in \cal{X},a,b \in \cal{A}$ it holds that
\begin{align}\label{eq:hat-operators-to-original-operators}
\rho(\hat{A}_a^x \,\, \hat{B}_b^y) \approx \rho (I^\alice_{x,a} \,\, I^\bob_{y,b}).
\end{align}
We achieve this by repeatedly applying \Cref{prop:exchange-operators}. First recall from \Cref{cor:entanglement-bound-question-sampling} that $\Pi \approx \sample^\alice_0 \sample^\bob_0$. Here we use $0$ as a shorthand notation for $0^\ell$. So we have
\begin{align*}
    \rho\paren{\hat{A}_a^x \,\, \hat{B}_b^y} &= \rho \paren{\Pi \,\, \tilde{A}_a^x \,\, \Pi \,\, \tilde{B}_b^y\,\,  \Pi}\\
    &\approx \rho\paren{\sample^\alice_0 \,\, \tilde{A}_a^x \,\, \sample^\alice_0 \,\, \sample^\bob_0 \,\, \tilde{B}_b^y \,\, \sample^\bob_0 },
\end{align*}
where we used \Cref{thm:rigidity-of-question-sampling} which states that $\sample^{\alice }_{0}$ and $\sample^{\bob }_{0}$ approximately commute. We continue by expanding $\tilde{A}_a^x$ and $\tilde{B}_a^x$ to obtain
\begin{align*}
    \rho\paren{\sample^\alice_0 \,\, \tilde{A}_a^x \,\, \sample^\alice_0 \,\, \sample^\bob_0 \,\, \tilde{B}_b^y \,\, \sample^\bob_0 } &=\rho\paren{\sample^\alice_0 \, \paren{O^\alice_x \,\, I^\alice_{x,a} \,\, O^\alice_x} \,\, \sample^\alice_0 \,\, \sample^\bob_0 \,\, \paren{O^\bob_y \,\, I^\bob_{y,b} \,\, O^\bob_y} \,\, \sample^\bob_0}\\
    &\approx \rho\paren{ \paren{O^{\alice}_{x} \,\, \sample^\alice_x \,\, I^\alice_{x,a} \,\, \sample^\alice_x \,\, O^\alice_x } \,\,   \paren{O^{\bob}_{y} \,\, \sample^\bob_y \,\, I^\bob_{y,b} \,\, \sample^\bob_y \,\, O^\bob_y }}
\end{align*}
where in the last line, we used \Cref{thm:rigidity-of-question-sampling} which states that $\sample^{W}_{0}\,O^{W}_{x} \approx O^{W}_{x}\,\sample^{W}_{x}$. By \Cref{lem:approximate-commutation} we have $\introspect^{W}_x \approx \sample^W_x$ so 
\begin{align*}
    \rho\paren{ \paren{O^{\alice}_{x} \,\, \sample^\alice_x \,\, I^\alice_{x,a} \,\, \sample^\alice_x \,\, O^\alice_x } \,\,   \paren{O^{\bob}_{y} \,\, \sample^\bob_y \,\, I^\bob_{y,b} \,\, \sample^\bob_y \,\, O^\bob_y }}
    &\approx \rho\paren{\paren{O^{\alice }_{x} \,\, \introspect^{\alice}_{x} \,\, \introspect^{\alice}_{x,a} \,\, \introspect^{\alice}_{x}O^{\alice }_{x}}   \paren{O^{\bob }_{y} \,\,\introspect^{\bob}_{y} \,\,\introspect^{\bob}_{y,b}\,\, \introspect^{\bob}_{y}\,\, O^{\bob }_{y}}}\\
    &\approx \rho\paren{\paren{O^{\alice }_{x} \,\,\introspect^{\alice}_{x,a}\,\, O^{\alice }_{x} }  \paren{O^{\bob }_{y}\,\, \introspect^{\bob}_{y,b}\,\, O^{\bob }_{y}}}
\end{align*}
where in the last line we used that $\introspect^{W}_x = \sum_a \introspect^{W}_{x,a}$ and that $\introspect^{W}_{x,a}$ are projections. Now using \Cref{lem:approximate-commutation} again, we know that erasure observables $O^W$ approximately commute with $\introspect^\cw$ projections. We also know that erasure observables $O^{\alice }$ and $O^{\bob }$ approximately commute. So we continue as follows
\begin{align*}
    \rho\paren{\paren{O^{\alice }_{x} \,\,\introspect^{\alice}_{x,a}\,\, O^{\alice }_{x} }  \paren{O^{\bob }_{y}\,\, \introspect^{\bob}_{y,b}\,\, O^{\bob }_{y}}}
    &\approx \rho\paren{O^{\bob }_{y} \,\, O^{\alice }_{x} \,\, \introspect^{\alice}_{x,a}   \introspect^{\bob}_{y,b} \,\, O^{\alice }_{x} \,\, O^{\bob }_{y}}\\
    &\approx \rho\paren{\paren{O^{\bob }_{y}}^2 \,\, \paren{O^{\alice }_{x}}^2 \,\, \introspect^{\alice}_{x,a}   \,\, \introspect^{\bob}_{y,b}}\\
    &= \rho\paren{\introspect^{\alice}_{x,a} \,\,  \introspect^{\bob}_{y,b}}.
\end{align*}
This completes the proof of \Cref{eq:hat-operators-to-original-operators}.

Our next goal is to show that POVMs $\{\hat{W}_a^x\}_a$ are close to being projective measurements. To this end, we first show that for any $x\in \cal{X}$ and $a,b\in \cal{A}$ 
\begin{align}\label{eq:kronecker}
    \hat{W}_a^x\hat{W}_b^x \approx \hat{W}_a^x\mathbf{1}_{a=b}
\end{align}
where $\mathbf{1}_{a=b}$ is the indicator variable for the equality $a=b$. First expanding according to the definitions
\begin{align*}
    \hat{W}_a^x\hat{W}_b^x &= \Pi \,\, O^{W}_x \,\, \introspect^{W}_{x,a} \,\, O^{W}_x \,\, \Pi \,\, O^{W}_x \,\, \introspect^{W}_{x,b} \,\, O^{W}_x \,\,\Pi\\
    &\approx \Pi \,\, O^{W}_x \,\, \introspect^{W}_{x,a} \,\, O^{W}_x \paren{\sample^{\cw}_{0}\,\, \sample^{W}_{0}\,\,\sample^{\cw}_{0}} O^{W}_x \,\, \introspect^{W}_{x,b} \,\, O^{W}_x  \,\, \Pi
\end{align*}
where in the last line we used the fact that $\Pi \approx \sample^{\cw}_{0}\,\,\sample^{W}_{0}\,\,\sample^{\cw}_{0}$ by \Cref{cor:entanglement-bound-question-sampling}. Now sampling projections $\sample^\cw$ commute with erasure observables $O^W$ and Introspection projections $\introspect^W$ so
\begin{align*}
    \Pi \,\, O^{W}_x \,\, \introspect^{W}_{x,a} \,\, O^{W}_x \paren{\sample^{\cw}_{0}\,\, \sample^{W}_{0}\,\,\sample^{\cw}_{0}} O^{W}_x \,\, \introspect^{W}_{x,b} \,\, O^{W}_x  \,\, \Pi &\approx \Pi \,\, \sample^{\cw}_{0} \,\, O^{W}_x  \,\, \introspect^{W}_{x,a} \,\, O^{W}_x \,\, \sample^{W}_{0} \,\, O^{W}_x \,\, \introspect^{W}_{x,b}  \,\, O^{W}_x \,\, \sample^{\cw}_{0} \,\, \Pi\\
    &\approx \Pi \,\, O^{W}_x \,\, \introspect^{W}_{x,a} \,\, O^{W}_x \,\, \sample^{W}_{0} \,\, O^{W}_x \,\, \introspect^{W}_{x,b} \,\, O^{W}_x \,\, \Pi
\end{align*}
where in the last line we use the fact that $\Pi \approx \sample^{\cw}_{0}\,\,\sample^{W}_{0}\,\,\sample^{\cw}_{0}$, and hence $\Pi \,\, \sample^{\cw}_{0} \approx \Pi \approx \sample^{\cw}_{0} \,\, \Pi$. Now moving $\sample^W_0$ passed $O^W_x$ using the relation $O^W_x \,\, \sample^W_0 \approx \sample^W_x \,\, O^W_x$, and then using the fact that $(O^W_x)^2 =I$ (as $O^W_x$ is an observable), we get
\begin{align*}
    \Pi \,\, O^{W}_x \,\, \introspect^{W}_{x,a} \,\, O^{W}_x \,\, \sample^{W}_{0} \,\, O^{W}_x \,\, \introspect^{W}_{x,b} \,\, O^{W}_x \,\, \Pi  &\approx \Pi \,\, O^{W}_x \,\, \introspect^{W}_{x,a} \,\, \sample^{W}_{x} \,\, (O^{W}_x)^2 \,\, \introspect^{W}_{x,b} \,\, O^{W}_x \,\, \Pi\\
    &= \Pi \,\, O^{W}_x \,\, \introspect^{W}_{x,a} \,\, \sample^{W}_{x} \,\, \introspect^{W}_{x,b} \,\, O^{W}_x \,\, \Pi
\end{align*}
Now substituting $\introspect^W_x$ in place of $\sample^W_x$ we get
\begin{align*}
    \Pi \,\, O^{W}_x \,\, \introspect^{W}_{x,a} \,\, \sample^{W}_{x} \,\, \introspect^{W}_{x,b} \,\, O^{W}_x \,\, \Pi
    &\approx \Pi \,\, O^{W}_x \,\, \introspect^{W}_{x,a} \,\, \introspect^{W}_{x} \,\, \introspect^{W}_{x,b} \,\, O^{W}_x \,\, \Pi\\
    &\approx \,\, \Pi \,\, O^{W}_x \,\, \introspect^{W}_{x,a} \,\, \introspect^{W}_{x,b} \,\, O^{W}_x \,\, \Pi\\
    &= \hat{W}_a^x\,\, \delta_{a,b},
\end{align*}
where in the last line we used the fact that $\introspect^{W}_{x,a}$ and $\introspect^{W}_{x,b}$ are orthogonal projections when $a\neq b$. This completes the proof of \Cref{eq:kronecker}. From this, we immediately obtain that $\paren{\hat{W}_\perp^x}^2 \approx \hat{W}_\perp^x$ also. So we established that $$\paren{\hat{W}_a^x}^2 \approx \hat{W}_a^x$$ for all $x \in \cal{X}$ and $a \in \tilde{\cal{A}}$. Using \Cref{prop:exchange-operators}, this in turn implies that $$\rho((\hat{W}_a^x)^2) \approx \rho(\hat{W}_a^x)$$ for all $a \in \tilde{\cal{A}}$. By definition of $\sigma$ it is also true that $$\sigma((\hat{W}_a^x)^2) \approx \sigma(\hat{W}_a^x).$$ 
So far we established that $\hat{W}^x_a$, as operators in $\hat{\algebra}$ acting on $\hat{\hilb}$, are close to projections. So applying \Cref{lem:projectivization}, for every $W\in \{A,B\}$ and $x \in \cal{X}$, there exists a projective measurement $\{W_a^x\}_a \subset \hat{\algebra}$ that is close to $\{\hat{W}_a^x\}_a$. 

Our final goal is to build a strategy for $G$ using these hard-earned projective measurements $\{{A}^{x}\}$ and $\{{B}^{y}\}$. On our way, we first need to relate $\{{A}_{a}^{x}\}_a$ and $\{B_{b}^{y}\}_b$ to the original measurements $\introspect^A_{x,a}$ and $\introspect^B_{y,b}$. For every $x,y \in \cal{X},a,b \in \cal{A}$, we can write
\begin{align*}
    \sigma\paren{{A}_{a}^{x}{B}_{b}^{y}} &\approx  \sigma\paren{{\hat{A}}_{a}^{x}{\hat{B}}_{b}^{y}}= \frac{\rho\paren{\hat{A}_{a}^{x}\hat{B}_{b}^{y}}}{\rho(\Pi)}\approx \frac{\rho\paren{\hat{A}_{a}^{x}\hat{B}_{b}^{y}}}{2^{-2\ell}}\approx \frac{\rho\paren{\introspect^{\alice}_{x,a} \introspect^{\bob}_{y,b}}}{2^{-2\ell}}.
\end{align*}
From this and \Cref{lem:approximately-uniform-sampling}, if $x,y$ is nontrivial in $G$, it holds that \[\frac{1}{2^{2\ell}}\sigma\paren{{A}_{a}^{x}{B}_{b}^{y}}\approx \rho\paren{\introspect_{x,a,y,b}}.\]
Therefore summing over all nontrivial question pairs, we have
\begin{align*}
    \sum_{\substack{x,y\\\text{nontrivial}}} \frac{1}{2^{2\ell}} \sum_{a,b \in \cal{A}} D(x,y,a,b) \sigma\paren{{A}_{a}^{x}{B}_{b}^{y}} &\approx \sum_{\substack{x,y\\\text{nontrivial}}} \sum_{a,b \in \cal{A}} D(x,y,a,b) \rho\paren{\introspect_{x,a,y,b}}.
\end{align*}
A similar approximate identity holds when summing over trivial question pairs, that is
\[\sum_{\substack{x,y\\\text{trivial}}} \frac{1}{2^{2\ell}} \sum_{a,b \in \cal{A}} D(x,y,a,b) \sigma\paren{{A}_{a}^{x}\,\, {B}_{b}^{y}}  \approx \sum_{\substack{x,y\\\text{trivial}}} \sum_{a,b \in \cal{A}} D(x,y,a,b)\rho\paren{\introspect_{x,a,y,b}}.\]
Let us see why this is true. First using the fact that $D(x,y,a,b) = 1$ for all $a,b$ and trivial question pair $x,y$, we can write
\begin{align*}
    \sum_{\substack{x,y\\\text{trivial}}} \frac{1}{2^{2\ell}} \sum_{a,b \in \cal{A}} D(x,y,a,b) \sigma\paren{{A}_{a}^{x}\,\,{B}_{b}^{y}}  =  \sum_{\substack{x,y\\\text{trivial}}} \frac{1}{2^{2\ell}} \sum_{a,b \in \cal{A}} \sigma\paren{{A}_{a}^{x}\,\, {B}_{b}^{y}} =  \sum_{\substack{x,y\\\text{trivial}}} \frac{1}{2^{2\ell}}
\end{align*}
where in the last equality we used the fact that $\sum_{a,b} {A}_{a}^{x}\,\, {B}_{b}^{y} = I_{\hat{\hilb}}$. Luckily, we also know that $\rho(\introspect_{x,y}) \approx \frac{1}{2^{2\ell}}$ by \Cref{lem:approximately-uniform-sampling}, and thus
\begin{align*}
    \sum_{\substack{x,y\\\text{trivial}}} \frac{1}{2^{2\ell}}
    &\approx  \sum_{\substack{x,y\\\text{trivial}}} \rho(I_{x,y})\\
    &= \sum_{\substack{x,y\\\text{trivial}}} \sum_{a,b \in \cal{A}} \rho\paren{\introspect_{x,a,y,b}}\\
    &= \sum_{\substack{x,y\\\text{trivial}}} \sum_{a,b \in \cal{A}} D(x,y,a,b)\rho\paren{\introspect_{x,a,y,b}}
\end{align*}
where in the last line we again used the fact that $D(x,y,a,b) = 1$ for all $a,b$ and trivial question pair $x,y$. 

So overall we established that
\[\sum_{x,y} \frac{1}{2^{2\ell}} \sum_{a,b \in \cal{A}} D(x,y,a,b) \sigma\paren{{A}_{a}^{x}\,\,{B}_{b}^{y}} \approx \sum_{x,y} \sum_{a,b \in \cal{A}} D(x,y,a,b) \rho\paren{\introspect_{x,a,y,b}}.\]
The right-hand-side is an upper bound on the probability of winning of $\strategy^\intro$ conditioned on the event that one of the players received the Introspection question $\introspect$. This probability must be at least $1-\delta(\eps)$ by a simple averaging argument. So we have
\begin{align}\label{eq:winning-probability-of-non-synchronous-strategy}
\sum_{x,y} \frac{1}{2^{2\ell}} \sum_{a,b \in \cal{A}} D(x,y,a,b)\sigma\paren{{A}_{a}^{x}\,\,{B}_{b}^{y}} = 1-\delta(\eps).
\end{align}
To summarize, at a high level, we constructed a set of operators $A_a^x$ and $B_b^y$ that together resemble a strategy for $G$ albeit with two sets of measurement operators instead of one. It remains to show that we can turn this into a synchronous strategy. From \Cref{eq:winning-probability-of-non-synchronous-strategy}, for every $x\in \cal{X}$ it must be that 
\[
    \sum_{a,b \in \cal{A}} D(x,x,a,b) \sigma\paren{{A}_{a}^{x}\,\,{B}_{b}^{x}} = 1 - \delta(\eps). 
\]
Since $G$ is synchronous, we have $D(x,x,a,b) = 0$ whenever $a\neq b$. Therefore 
\[\sum_{a\in \cal{A}} \sigma\paren{{A}_{a}^{x}\,\,{B}_{a}^{x}} = 1 - \delta(\eps)\]
or equivalently that $A_a^x \simeq B_a^x$ for every $x\in \cal{X}$. Therefore by \Cref{lem:consistency-consequences}, it holds that $A_a^x \approx B_a^x$ for every $x\in \cal{X}$. Therefore $\sigma(A^x_a \,\, B^y_b) \approx \sigma(A^x_a \,\, A^y_b)$. Using this approximation in (\ref{eq:winning-probability-of-non-synchronous-strategy}) we conclude that 
\begin{align}\label{eq:value-of-derived-strategy2}
    \sum_{x,y \in \cal{X}} \frac{1}{2^{2\ell}} \sum_{a,b \in \cal{A}} D(x,y,a,b) \sigma\paren{{A}_{a}^{x}\,\,{A}_{b}^{y}} = 1 - \delta(\eps).
\end{align}
Now we reduced to one set of measurement operators $A_a^x$ that more closely resemble a synchronous strategy for $G$.
Unfortunately we are not quite there as the set of operators $\{A_a^x\}_{a\in \cal{A}}$ is not a projective measurement if $A^x_{\perp} \neq 0$. We can resolve this issue by defining projective measurements $\{N_a^x\}_{a\in \cal{A}}$ for every $x$ such that $N_{a^\ast}^x = A_{a^\ast}^x + A_\perp^x$ for some special element $a^\ast\in \cal{A}$ and $N_a^x = A_a^x$ for all $a \neq a^\ast$. Now $\strategy = (\sigma, \{N^x\}_{x\in \cal{X}})$ is a synchronous strategy and is such that $\sigma(N^x_a N^y_b) \geq \sigma(A^x_a A^y_b)$. So by (\ref{eq:value-of-derived-strategy2}), we have
\begin{align*}
    \val(G,\strategy)=\sum_{x,y \in \cal{X}} \frac{1}{2^{2\ell}} \sum_{a,b \in \cal{A}} D(x,y,a,b) \sigma\paren{{N}_{a}^{x}{N}_{b}^{y}} = 1-\delta(\eps).
\end{align*}
So for all sufficiently small $\eps$, if there exists a strategy $\strategy^\intro$ with value at least $1-\eps$, we showed the existence of a strategy for $G$ with value $1-\delta(\eps)$. This in turn implies that for all $t \in \{q,co\}$
\[\val_t^s(G^\intro) = 1 \implies \val_t^s(G) = 1.\]

Next we prove the inequality
\[\mathcal{E} (G^{\intro},1) \geq \max \left \{ \mathcal{E} (G,1) , 2^{2\ell} \right\}.\]
Suppose the finite dimensional strategy $\strategy^\intro = (\rho,\{P^q\}_{q\in\cal{Q}^\intro})$ defined over a Hilbert space $\hilb$ has value $1$. Then since the strategy restricted to the Question Sampling game also wins with probability $1$, from \Cref{cor:entanglement-bound-question-sampling}, we get that the dimension of $\hilb$ is at least $2^{2\ell}$.

It remains to show that $\mathcal{E} (G^{\intro},1) \geq \mathcal{E} (G,1)$. Consider the finite-dimensional strategy $\strategy = (\sigma,\{N^x_a\})$ constructed above for the original game $G$. The inequality now follows from the fact that the strategy $\strategy$ is over the Hilbert space $\hat{\hilb}$ defined in \Cref{lem:new-algebra} which is a subspace of $\hilb$.

\end{proof}

\subsection{Proof of \Cref{thm:question-reduction}}
\label{sec:question-reduction-thm-proof}

From \Cref{def:verifiers}, we can let $G_n = (\cal{X}_n,\cal{A}_n,D_n)$ where $\cal{X}_n = \{0,1\}^{\ell_n}$ for some polynomial-time computable function $\ell_n$ of $n$. As we indicated in \Cref{def:verifiers}, the decider and checker Turing machines discard any string that comes after the $\ell_n$th bit in their second and third input tapes. By assumption, for all sufficiently large $n$, we have $\ell_n \leq n^\alpha$, so from our previous statement, we can simply assume that $\ell_n = n^\alpha$. We design the algorithm $\alg{QuestionReduction}_\alpha$ so that $G_n^\intro$ is the Introspection game $(G_n)^\intro$ as defined in \Cref{sec:introspection-definition}. From the definition of Introspection, it is straightforward to see that a polynomial-time algorithm exists that computes a description of $\verifier^\intro = (D^\intro,C^\intro)$ from a description of $\verifier = (D,C)$. The question length of $G_n^\intro$ is $\poly(\alpha,\log n)$ by the definition of the Introspection game and the assumption that $\ell_n  = n^\alpha$.

Given a pair of questions in $G^\intro_n$, if they are both Question Sampling questions, then they are a nontrivial question pair in the Introspection game if and only if they are a nontrivial question pair in the Question Sampling game. If questions are both among special questions $$\sample_\alice ,\erase_\alice ,\introspect_\alice, \introspect_\alice \sample_\bob  ,\introspect_\alice \erase_\bob ,\sample_\bob ,\erase_\bob , \introspect_\bob, \introspect_\bob \sample_\alice ,\introspect_\bob \erase_\alice ,$$ then the pair is nontrivial if they are connected by an edge or a self-loop in \Cref{fig:graph-of-introspection-game}. Since this graph has constant size, this can be decided in $O(1)$. If one question is a Question Sampling question that is not any of $\sample_\alice ,\sample_\bob ,\erase_\alice ,\erase_\bob $ and the other is a special Introspection game question $$\introspect_\alice, \introspect_\alice \sample_\bob  ,\introspect_\alice \erase_\bob , \introspect_\bob, \introspect_\bob \sample_\alice ,\introspect_\bob \erase_\alice ,$$ then the pair is trivial. Therefore the complexity of deciding if a pair is trivial in $G^\intro_n$ is asymptotically the same as the complexity of deciding if a pair is trivial in $\qs_{n^\alpha}$ which is $\poly(\alpha,\log n)$ (see \Cref{tab:question-sampling}).

Next we bound the complexity of $D^\intro(n)$. The bit length of questions in the Introspection game $G_n^\intro$ is $\poly(\alpha,\log n)$. The answer length of $G^\intro_n$ is $n^\alpha$ (as the answer length of $G_n$ is bounded by $\TIME_D(n)$). So the decider can compute in time $\poly(n^\alpha)$ whether the answer format of $G^\intro_n$ is respected. The decider, by simulating $D(n)$ and $C(n)$, can compute in time $\poly(|D|,|C|,\alpha,n^\alpha)$ whether a give quadruple $(q,r,\hat{a},\hat{b})$ is an accepting quadruple in $G_n^\intro$ according to \Cref{tab:introspection}. 

The completeness, soundness, and the dimension bound follow immediately from Propositions \ref{prop:introspection-completeness} and \ref{prop:introspection-soundness}.

\newpage
\section{Answer Reduction}
\label{sec:answer-reduction}

In this section we present the answer reduction transformation, whose properties are given by the following Theorem.
\begin{theorem}[Answer Reduction]
\label{thm:answer-reduction}
    For all $\beta \in \N$ there exists a polynomial-time algorithm $\alg{AnswerReduction}_{\beta}$ that takes as input a pair of Turing machines $(D,C)$ and outputs a pair of Turing machines $(D^\ans,C^\ans)$ such that the following holds. If $\verifier = (D,C)$ is a verifier for a sequence of games $\UGS_\verifier = (G_n)_{n \in \N}$ and $n_0 \in \N$ is an integer such that for all $n \geq n_0$,
        \begin{align*}
        	&\text{The questions of $G_n$ have length at most $\log^\beta(n)$, } \\
    &\TIME_{C}(n) = \log^\beta n~, \text{ and} \\    	
	&\TIME_{D}(n) \leq n^\beta
		\end{align*}
	then the output $\verifier^\ans = (D^\ans,C^\ans)$ is a verifier for a sequence of games $\UGS_{\verifier^\ans} = (G^\ans_n)_{n \in \N}$ with the following properties. There exists $\gamma = \poly(\beta)$ and $n_0^\ans = \poly(\gamma^\gamma,n_0)$ such that for all $n \geq n_0^\ans$,
    \begin{enumerate}
        \item (Complexity bounds) 
        \begin{gather*}
        	\TIME_{D^{\ans}}(n) = \log^\gamma n \\
			\TIME_{C^{\ans}}(n) = \log^\gamma n~.
		\end{gather*}
        \item (Completeness) For all oracularizable synchronous strategies $\strategy$ for $G_n$, there exists an oracularizable synchronous strategy $\strategy^\ans$ for $G_n^\ans$ such that
        \[
            \val(G^{\ans}_n,\strategy^\ans) \geq \frac{1}{2} + \frac{1}{2}\val(G_n,\strategy).
        \]
        Furthermore, if $\strategy$ is finite-dimensional, then so is $\strategy^\ans$.
        \item (Soundness) For all $t \in \{q,co\}$ we have 
        \[
        \val_t^s(G_n) < 1 \Longrightarrow \val_t^s(G^{\ans}_n) < 1~.
       	\]
		\item (Entanglement bound) 
		\[
		\mathcal{E}(G_n^\ans,1) \geq \mathcal{E}(G_n,1)~.
		\]
    \end{enumerate}
\end{theorem}
Intuitively, the answer reduction transformation transforms a sequence of games $(G_1,G_2,\ldots)$ to a sequence $(G_1^\ans,G_2^\ans,\ldots)$ such that the time complexity of the ``answer reduced'' game $G_n^\ans$ (in terms of computing its decision predicate) is \emph{polylogarithmic} in the time complexity $T(n)$ of the ``original game'' $G_n$, and \emph{polynomial} in the question length $Q(n)$ of $G_n$. The reason this transformation is called ``answer reduction'' is as follows. Suppose the original game $G_n$ already has polylogarithmic-length questions (i.e. $Q(n) \leq \poly(\log T(n))$), but the answer lengths are, say, $\Omega(T(n))$; this will be the case when we apply answer reduction to the introspection games from the previous section. The resulting game $G_n^\ans$ then has time complexity $\poly(\log T(n))$ and in particular both the question and answer lengths of $G_n^\ans$ are at most $\poly(\log T(n))$. 

We describe and analyze the answer reduction transformation $G \mapsto G^\ans$ for a single game (rather than a sequence), and then prove \Cref{thm:answer-reduction} in \Cref{sec:answer-reduction-proof}.

\subsection{Overview}
Let $Q,T \in \N$ be integers and let $G = (\cal{X},\cal{A},D)$ be a synchronous game where $\cal{X} = \{0,1\}^Q$ and $\cal{A} = \{0,1\}^T$, and $\TIME_D \leq T$ (meaning that on all inputs $D$ halts within $T$ timesteps). We can assume via padding that all questions have the same length, and all the answers have the same length.

\paragraph{Oracularization.} We first give an overview of a transformation on $G$ called \emph{oracularization}. This produces the following game $G^\orac$. The verifier may send a player either a question $x \in \cal{X}$ or a pair of questions $(x,y) \in \cal{X}^2$; thus the question alphabet is $\cal{X} \cup \cal{X}^2$. When a player receives a single question $x$ we call them an \emph{isolated player} and its question an \emph{isolated question}. When a player receives a pair $(x,y)$ we call them an \emph{oracle player} and its question an \emph{oracle question}. 

If both players receive the same question (either isolated or oracle), then they must return the same answer. If one player receives an oracle question $(x,y) \in \cal{X}^2$ that is nontrivial for the original game $G$ and the other receives an isolated question $x$ (resp. receives $y$), then the players win if the oracle player responds with an answer pair $(a,b) \in \cal{A}^2$ such that $D(x,y,a,b) = 1$ and the isolated player responds with answer $a$ (resp. responds with answer $b$). All other question combinations are considered trivial for $G^\orac$, and the players automatically win in those cases. 

Intuitively, in the oracularization of $G$ an oracle player must ``simulate'' the behavior of the two players in $G$, and the isolated player (who only receives half of the oracle question) is used to check that the oracle player's answers $(a,b)$ are produced in a way that $a$ only depends on $x$ and $b$ only depends on $y$.

\paragraph{Answer Reduction.} We now give a high-level overview of the \emph{answer-reduced} game $G^\ans = (\cal{X}^\ans,\cal{A}^\ans,D^\ans)$.  The questions of $G^\ans$ are of the form $(g,p)$, where $g$ is a \emph{game question} and $p$ is a \emph{proof question}. The game question $g$, intuitively, is meant to indicate a question from the original game $G$. However, in the answer reduction transformation, the game questions $g$ come from the oracularization $G^\orac$ of $G$.

In the oracularized game $G^\orac$, the players are supposed to respond with either an answer from $\cal{A}$ or from $\cal{A}^2$, depending on whether they received an isolated or oracle question. In the answer reduced game $G^\ans$, however, the players do not respond with a ``full-sized'' answer in $\cal{A} \cup \cal{A}^2$. Instead, the verifier expects that the oracle players will generate a \emph{proof} $\pi$ that they can produce answers $(a,b) \in \cal{A}^2$ that satisfies the decision predicate of the game $G$, and furthermore these answers can be produced in a way such that $a$ only depends on $x$ and $b$ only depends on $y$. The verifier does not examine this purported proof $\pi$ in its entirety but instead uses the proof question $p$ to query it in a constant number of locations.

The main point is this: now the players only have to respond with a constant number of bits corresponding to the proof locations queried, rather than with a symbol from the set $\cal{A} \cup \cal{A}^2$ (whose size we think of as growing to infinity). To ensure that the players' answers to the local queries are consistent with a global proof string $\pi$, and that the purported answers $(a,b)$ (which are included in $\pi$) was generated ``honestly'' (e.g., $a$ does not depend on $x$), the verifier performs cross-checks between the two players. Before describing the format of the proof questions, we first explain in detail what a proof is supposed to look like.

The starting point is the well-known \emph{Cook-Levin reduction} from classical computer science: this is an efficient transformation that maps Turing machines $M$ to 3SAT formulas $\varphi_M$ such that there is an input $w$ (called the \emph{witness}) where $M(w) = 1$ if and only if $\varphi_M$ is satisfiable. Furthermore, it is well-known~\cite[Chapter~20]{papadimitriou1994computational} that the clauses of the SAT formula $\varphi_M$ can be computed extremely efficiently -- in fact, in time that is \emph{logarithmic} in the size of the entire SAT formula (if we treat the description length of $M$ as a constant):

\begin{theorem}[Cook-Levin Theorem]
\label{thm:cook-levin}
For all $1$-input Turing machines $M$ and integers $R,T \in \N$, there exists a 3SAT formula $\varphi(M,T,R)$ (called a \emph{Cook-Levin SAT formula}) with $L = \poly(|M|,T,R)$ variables, such that
\begin{itemize}
	\item For all $w \in \{0,1\}^R$ such that $M(w)$ accepts within $T$ time steps, there exists a unique satisfying assignment $\pi$ for the formula $\varphi(M,T,R)$, and furthermore $\pi_{\leq R}$ (the first $R$ bits of $\pi$) is $w$, and
	\item For all satisfying assignments $\pi$ for the formula $\varphi(M,T,R)$, the Turing machine $M$ accepts $\pi_{\leq R}$ within $T$ time steps.	
\end{itemize}
Furthermore, there exists a polynomial-time algorithm $\alg{CookLevin}$ that takes as input a tuple $(M,T,R,i,j,k)$ where $R,T,i,j,k$ are integers written in binary, and outputs the literals of the clause(s) of $\varphi(M,T,R)$ that contains the $i$-th, $j$-th, and $k$-th variables (or outputs a null symbol if no such clause exists).
\end{theorem}
We note that while the algorithm $\alg{CookLevin}$ runs in polynomial time in the length of its input, it runs in \emph{logarithmic} time in the number of variables of the Cook-Levin SAT formula $\varphi(M,T,R)$. This is because the length of the input tuple $(M,T,R,i,j,k)$ is $O(|M| + \log T + \log R + \log i + \log j + \log k)$, and since the variable indices $i,j,k$ are at most $\poly(|M|,T,R)$, the time complexity of the algorithm $\alg{CookLevin}$ is at most $\poly(|M|,\log T,\log R)$.

The verifier in the answer-reduced game $G^\ans$ expects an oracle player who received game question pair $g = (x,y)$ to compute a string $\pi$ satisfying the following:
\begin{enumerate}
	\item $\pi$ is a satisfying assignment for the Cook-Levin SAT formula $\varphi(D_{x,y},T,2T)$ where $D_{x,y}$ is the $1$-input Turing machine that on input $(a,b) \in \{0,1\}^{2T}$ executes the Turing machine $D$ on input $(x,y,a,b)$, and
	\item $\pi$ is composed of three strings $(a,b,\pi') \in \{0,1\}^{T} \times \{0,1\}^{T} \times \{0,1\}^L$ where \\ $L = \poly( |D_{x,y}|,T) = \poly(|D|, Q, T)$. Here we used that the description length $|D_{x,y}| = O(|D| + |x| + |y|) = \poly(|D|,Q)$. 
\end{enumerate}
Henceforth we shall abbreviate the Cook-Levin formula $\varphi(D_{x,y},T,2T)$ as $\varphi_{x,y}$. 

The verifier asks proof questions $p$ in order to ascertain whether it is possible for an oracle player to generate a proof $\pi$ satisfying these conditions. This requires the verifier to ask proof questions to both oracle players \emph{and} isolated players. Oracle players (who get game question pair $g = (x,y)$) can get asked to provide:
\begin{itemize}
	\item A single bit $\pi_i$ of the proof $\pi$, or
	\item A triple of bits $(\pi_i,\pi_j,\pi_k)$ from the proof $\pi$ (which may not necessarily correspond to a clause in $\varphi_{x,y}$). \end{itemize}
An isolated player (who gets a single question $x$ or $y$) is asked to provide a pair of bits $(a_i,a_j)$ of their purported answer $a \in \{0,1\}^{T}$. 

Thus the proof questions are sampled from the set $[L] \cup [L]^2 \cup [L]^3$. Thus the question and answer sets for $G^\ans$ are
\[
	\cal{X}^\ans = \cal{X}^\orac \times ([L] \cup [L]^2 \cup [L]^3) \qquad \qquad \cal{A}^\ans = \{0,1\} \cup \{0,1\}^2 \cup \{0,1\}^3 
\]
where $\cal{X}^\orac = \cal{X} \cup \cal{X}^2$ is the question alphabet for the oracularized game $G^\orac$. 

Since the player answers $(a,b)$ are supposed to be embedded into a proof $\pi$, we use the following mapping to translate between indexing into answer $a$ or $b$ versus indexing into the proof $\pi$: given an index $i \in [T]$, the $i$-th bit of the \emph{first} answer $a$ (corresponding to the \emph{first} question $x$) is mapped to index $\eta(i) = i$ of the proof $\pi$, and the $i$-th bit of the \emph{second} answer $b$ (corresponding to the \emph{second} question $y$) is mapped to index $\lambda(i) = T + i$ of $\pi$.

\subsection{The answer-reduced decision procedure}
\label{sec:ar-decider}

We now formally specify the decision procedure $D^\ans$. On input $(\hat{x},\hat{y},\hat{a},\hat{b})$, it checks if $(\hat{x},\hat{y})$ (resp. $(\hat{y},\hat{x})$) is one of the nontrivial question pairs of $G^\ans$, which are presented in \Cref{tab:answer-reduction}. If so, then it accepts if and only if the answers $(\hat{a},\hat{b})$ (resp. $(\hat{b},\hat{a})$) satisfy the corresponding winning condition. Otherwise, if $(\hat{x},\hat{y})$ is a trivial question, the verifier automatically accepts.

\begin{table}[H]
\centering
\begin{tabularx}{\textwidth}{ L{8cm} L{8cm} } 
 \toprule
 \textbf{Nontrivial Question Pair $(\hat{x},\hat{y})$} & \textbf{Winning Condition on Answers $(\hat{a},\hat{b})$} \\
 \midrule
  $\hat{x} = \hat{y}$ &  $ \hat{a}=\hat{b}$ \\
 \midrule 
 $\hat{x} = ((x,y),i)$ where $(x,y)$ is nontrivial for $G$ & $(s_j,s_k,s_\ell)$ satisfies clause(s) specified by  \\
  $\hat{y} = ((x,y),(j,k,\ell))$ where $i \in \{j,k,\ell\}$ & $\alg{CookLevin}(D_{x,y},T,2T,j,k,\ell)$ and $r_i = s_i$, where $\hat{a} = r_i \in \{0,1\}$, $\hat{b} = (s_j,s_k,s_\ell) \in \{0,1\}^3$  \\
 \midrule
 $\hat{x} = ((x,y),i)$ where $(x,y)$ is nontrivial for $G$ & $r_i = a_{\eta^{-1}(i)}$   \\
 $\hat{y} = (x,(j,k))$ where $i \in \{\eta(j),\eta(k)\}$ & where $\hat{a} = r_i \in \{0,1\}$, $\hat{b} = (a_j,a_k) \in \{0,1\}^2$ \\  
 \midrule
	$\hat{x} = ((x,y),i)$ where $(x,y)$ is nontrivial for $G$ & $r_i = b_{\lambda^{-1}(i)}$  \\
	$\hat{y} = (y,(j,k))$ where  $i \in \{\lambda(j),\lambda(k)\}$ & where $\hat{a} = r_i \in \{0,1\}$, $\hat{b} = (b_j,b_k) \in \{0,1\}^2$\\ 
 \bottomrule  
\end{tabularx}
\caption{The nontrivial question pairs and winning conditions for the game $G^\ans$.}
\label{tab:answer-reduction}
\end{table}

\Cref{tab:answer-reduction} should be read as follows. In the second row, for example, the nontrivial question pair is where $\hat{x} = (g_1,p_1)$ where $g_1 = g_2 = (x,y) \in \cal{X}^2$ where $(x,y)$ is nontrivial for $G$, $p_1 = i$ for some $i \in [L]$, and $p_2 = (j,k,\ell) \in [L]^3$ such that $i \in \{j,k,\ell\}$. The answer $\hat{a}$ is expected to be a single bit $r_i$ and $\hat{b}$ is expected to be a triple of bits $(s_j,s_k,s_\ell)$; otherwise the verifier rejects. The verifier then checks that $r_i = s_i$ (i.e. the first player's assignment to the $i$-th variable of the proof is the same as the second player's assignment to the $i$-th variable), and the second player's assignment $(s_j,s_k,s_\ell)$ satisfies the clause of $\varphi_{x,y}$ that involves the triple of variables $(j,k,\ell)$. If there is no clause, then the verifier accepts any assignment to those variables.

\subsection{Completeness of answer reduction}

We now prove the completeness property of the answer reduction transformation. Similarly to \Cref{sec:question-reduction}, the completeness property implies that the value of $G^\ans$ is lower bounded by the value of $G$.

\begin{proposition}
\label{prop:answer-reduction-completeness}
For all oracularizable synchronous strategies $\strategy$ for $G$, there exists an oracularizable synchronous strategy $\strategy^\ans$ for $G^\ans$ such that 
	\[
		\val(G^{\ans}_n,\strategy^\ans) \geq \frac{1}{2} + \frac{1}{2}\val(G_n,\strategy)~.
	\] 
	Furthermore, if $\strategy$ is finite-dimensional then so is $\strategy^\ans$.	
\end{proposition}

\begin{proof}

Let $\strategy = (\tau,\{M^x\})$ be a tracial synchronous strategy for $G$ that commutes on the set of nontrivial questions of $G$.  We now define a tracial strategy $\strategy^\ans = (\tau,\{N^x\})$ for $G^\ans$. Before doing so, we define some intermediate measurements. Let $\cal{X}$ and $\cal{A}$ denote the question and answer sets of $G$, respectively. For all $x,y \in \cal{X}, a,b \in \cal{A}$:
\begin{itemize}
    \item $N^{x,y}_{a,b} = \left\{
	\begin{array}{ll}
		M^x_a \,\, M^y_b  & \mbox{if $(x,y)$ is a nontrivial question for $G$}  \\
		\id & \mbox{if $(x,y)$ is a trivial question for $G$ and $a = b = 0$} \\
		0 & \mbox{otherwise}
	\end{array}
    \right.$
    \item $N^{x}_a = M^x_a$.
\end{itemize}
The POVM $N^{x}$ is projective because $M^x$ is projective. Note that whenever $(x,y)$ is a nontrivial question of $G$, the projectors $M^x_a$ and $M^y_b$ commute, so $N^{x,y}$ is always projective. 

Now we define the measurements for $\strategy^\ans$:
\begin{enumerate}
    \item $N^{x,j,k} = N^{x}_{[a \mapsto (a_j,a_k)]}$
\item $N^{x,y,i} = N^{x,y}_{[(a,b) \mapsto \pi_i]}$
    \item $N^{x,y,i,j,k} = N^{x,y}_{[(a,b) \mapsto (\pi_i,\pi_j,\pi_k)]}$
\end{enumerate}
where here $\pi$ denotes the unique satisfying assignment to the Cook-Levin SAT formula $\varphi_{x,y}$ such that $\pi = (a,b,w)$ for some string $w$.

We now verify that the strategy $\strategy^\ans$ satisfies the desired properties: it is synchronous because the measurements are all projective. It commutes on the nontrivial questions of $G^\ans$, as seen by the following case analysis: letting $\hat{x} = (g_1,p_1)$ and $\hat{y} = (g_2,p_2)$,
\begin{enumerate}
	\item If $\hat{x} = \hat{y}$, then clearly the measurements $N^{\hat{x}}$ and $N^{\hat{y}}$ commute with each other because they are the same measurement. 
	\item If $g_1 = g_2 = (x,y)$, $p_1 = i$, and $p_2 = (j,k,\ell)$, then $N^{\hat{x}}$ and $N^{\hat{y}}$ are marginalizations of the same projective measurement $\{N^{x,y}\}$, and thus $N^{\hat{x}}$, $N^{\hat{y}}$ commute with each other. 
	\item If $g_1 = (x,y)$, $p_1 = i$, $g_2 =x$ (or $g_2 = y$) and $p_2 = (j,k)$, then either $(x,y)$ is a trivial question for $G$ (in which case $N^{\hat{x}}$ is the identity measurement, which commutes with everything), or $(x,y)$ is a nontrivial question, in which case $N^{\hat{x}}$ is a marginalization of the product $M^x_a M^y_b$, whereas $N^{\hat{y}}$ is a marginalization of $M^x_a$ (resp. $M^y_b$), which commutes with $M^y_b$ (resp. $M^x_a$).
\end{enumerate}

Clearly, the dimensionality of $\strategy^\ans$ is the same as the dimension of $\strategy$. 

Finally, we can evaluate the winning probability of $\strategy^\ans$ as follows: let $\gamma$ denote the probability that at least one of the players that receives a question $(g,p)$ where $g = (x,y)$ with $(x,y)$ nontrivial for $G$. If neither player receives such a game question, then either their question pair $(\hat{x},\hat{y})$ is trivial for $G^\ans$ (in which case the players win automatically), or $\hat{x} = \hat{y}$ (in which case the players win because their strategy is synchronous). 

Suppose one of the players (say, the first player) receiving such question pair $\hat{x} = (g,p)$. Intuitively, this oracle player will simultaneously measure $M^x$ and $M^y$ to obtain answers $(a,b)$. Since $x$ an $y$ are drawn uniformly at random, the probability that $D(x,y,a,b) = 1$ is exactly $\omega(G, \strategy)$. Suppose $(a,b)$ are winning answers. Then the oracle player can compute a satisfying assignment $\pi = (a,b,w)$ for the Cook-Levin formula $\varphi_{x,y}$ -- this uses the assumption that $\TIME_D \leq T$. Furthermore, the second player, no matter what question $\hat{y}$ they receive, they will be able to obtain perfectly consistent answers (if they receive game question $(x,y)$, then they can obtain the same proof $\pi = (a,b,w)$; if they receive game questions $x$ or $y$, they will obtain the same answers $a$ or $b$, respectively). Thus the success probability of the strategy $\strategy^\ans$ overall is at least
\[
	\val(G^\ans,\strategy^\ans) \geq (1 - \gamma) + \gamma \,\, \val(G,\strategy)\;.
\]
Since $\gamma \leq 1/2$, the Proposition follows. \end{proof}

\subsection{Soundness of answer reduction}

\begin{proposition}
\label{prop:answer-reduction-soundness}
For all $t \in \{q,co\}$, $\val_t^s(G) < 1 \Longrightarrow \val_t^s(G^\ans) < 1$.
\end{proposition}

\begin{proof}
Let $\strategy^\ans = (\tau,\{N^\hx \})$ be a tracial synchronous strategy for $G^\ans$ that has value $1 - \eps$. 
Our goal will be to construct measurements $\{M^x_a\}$ and $\{M^{x,y}_\pi\}$ that produce entire answer strings and entire proof strings, respectively. They will be constructed from the $N^{x,y,i}$ and $N^{x,j,k}$ measurements which only provide ``local'' views of purported answer and purported proof strings. In order to ``paste'' these ``local'' views together into consistent ``global'' views, we will need to establish pairwise consistency conditions between the measurement operators of the strategy $\strategy^\ans$.

From the condition that the strategy $\strategy^\ans$ has value $1 - \eps$, we obtain the following consistency conditions pointwise over all $x,y \in \cal{X}$ and $i,j,k,\ell \in [L]$: \begin{itemize}
    \item $N^{x,y,i}_r \simeq N^{x,y,j,k,\ell}_{[(s_j,s_k,s_\ell) \mapsto s_i \mid r]}$ whenever $i \in \{j,k,\ell\}$,
    \item $N^{x,y,\eta(j)}_r \simeq N^{x,j,k}_{[(a_j,a_k) \mapsto a_j \mid r]}$ and $N^{x,y,\eta(k)}_r \simeq N^{x,j,k}_{[(a_j,a_k) \mapsto a_k \mid r]}$
	\item $N^{x,y,\lambda(j)}_r \simeq N^{y,j,k}_{[(a_j,a_k) \mapsto a_j \mid r]}$ and $N^{x,y,\lambda(k)}_r \simeq N^{y,j,k}_{[(a_j,a_k) \mapsto a_k \mid r]}$
\end{itemize}
In other words, the assignments to variables that are in common to both players' questions are approximately consistent. Here and throughout this proof, all approximations ``$\simeq$'' and ``$\approx$'' implicitly hide some error function $\delta(\eps)$ that goes to $0$ as $\eps \to 0$. Furthermore, the error function will generally be different each time the ``$\simeq$'' or ``$\approx$'' notation is used. (See \Cref{subsec:asymptotics} for a more in-depth discussion of approximations and asymptotics).

We first prove a utility lemma, which will be used repeatedly throughout the analysis of soundness:
\begin{lemma}
\label{lem:ar-utility}
Let $t \in \N$ and let $A = \{A_{r} \}$ denote a projective measurement with outcomes in $\cal{R}^t$. For $i \in [t]$, let $B^i = \{B_r^i\}$ be a POVM with outcomes in $\cal{R}$. Suppose that for all $i \in [t]$,
\begin{align*}
	A_{[r \mapsto r_i \mid c]} \simeq_\delta B_{c}^i
\end{align*}
where the answer summation is over $c \in \cal{R}$. Then for all permutations $\sigma \in S_t$, we have that
\[
A_r \approx_{t\sqrt{2\delta}} B^{\sigma(1)}_{r_{\sigma(1)}} \cdot B^{\sigma(2)}_{r_{\sigma(2)}} \cdots B^{\sigma(t)}_{r_{\sigma(t)}} \;.
\]
In other words, the measurement $\{A_r\}$ is $t\sqrt{2\delta}$-close to the product of the $\{B^i_{r_i}\}$, in any order. Furthermore, 
\[
B^{\sigma(1)}_{r_{\sigma(1)}} \cdot B^{\sigma(2)}_{r_{\sigma(2)}} \cdots B^{\sigma(t)}_{r_{\sigma(t)}} \approx_{2t\sqrt{2\delta}} B^{\rho(1)}_{r_{\rho(1)}} \cdot B^{\rho(2)}_{r_{\rho(2)}} \cdots B^{\rho(t)}_{r_{\rho(t)}}
\]
for all permutations $\rho,\sigma \in S_t$.
\end{lemma}
\begin{proof}
We first argue that
\[
	A_r \approx_{t\sqrt{2\delta}} B^1_{r_1} \cdot B^2_{r_2} \cdots B^t_{r_t} \;.
\]
Using \Cref{lem:consistency-consequences} we get that for all $i \in [t]$,
\begin{gather}
\label{eq:ar-utility-0}
	A_{[r \mapsto r_i \mid c]} \approx_{\sqrt{2\delta}} B_r^i\;.
\end{gather}
Using \Cref{lem:add-a-measurement} we can right-multiply \Cref{eq:ar-utility-0} for $i = 1$ by the measurement $A_{[r \mapsto r_2 : d]}$ to deduce
\begin{equation}
\label{eq:ar-utility-2}
A_{[r \mapsto r_1]} \cdot A_{[r \mapsto r_2]} \approx_{\sqrt{2\delta}} B^1_{r_1} \cdot A_{[r \mapsto r_2]}
\end{equation}

Using using \Cref{lem:add-a-measurement} again we get that the right hand side of \Cref{eq:ar-utility-2} is $\sqrt{2\delta}$-close to $B^1_{r_1} \cdot B^2_{r_2}$, and therefore via the triangle inequality we get
\[
A_{[r \mapsto r_1]} \cdot A_{[r \mapsto r_2]} \approx_{2\sqrt{2\delta}} B^1_{r_1} \cdot B^2_{r_2}.
\]
Notice that since $A$ is projective, we have
\[
A_{[r \mapsto r_1]} \cdot A_{[r \mapsto r_2]} = A_{[r \mapsto (r_1,r_2)]}
\]
Thus $A_{[r \mapsto (r_1,r_2)]} \approx_{2\sqrt{2\delta}} B^1_{r_1} \cdot B^2_{r_2}$. By repeatedly using \Cref{lem:add-a-measurement}, we deduce that 
\[
A_r \approx_{t\sqrt{2\delta}} B^1_{r_1} \cdot B^2_{r_2} \cdots B^t_{r_t}
\]
as desired. The same argument holds with any other ordering of the $B^i$'s. 

The ``Furthermore'' part of the lemma then follows from the triangle inequality.
\end{proof}

\paragraph{Constructing the $M^x_a$ measurements.} The first step is to show that, for fixed $x,y$, the $\{ N^{x,y,i} \}$ measurements approximately commute. 

Fix $i,j \in [T]$. Using \Cref{lem:ar-utility} with $A = N^{x,i,j}$, $B^1 = N^{x,y,\eta(i)}$ and $B^2 = N^{x,y,\eta(j)}$, we get
\begin{equation}
\label{eq:ar-0}
	N^{x,y,\eta(j)}_s \cdot  N^{x,y,\eta(i)}_r  \approx  N^{x,y,\eta(i)}_r  \cdot N^{x,y,\eta(j)}_s \;.
\end{equation}

The next step is to deduce that the marginalizations of the $N^{x,i,j}$ measurements commute. Since $N^{x,y,\eta(i)}_r \approx N^{x,i,k}_{[(a_i,a_k) \mapsto a_i \mid r]}$ and $N^{x,y,\eta(j)}_s \approx N^{x,j,k}_{[(a_j,a_k) \mapsto a_j \mid s]}$ for all $k \in [T]$. Thus, using \Cref{lem:add-a-measurement} twice we get
\[
N^{x,y,\eta(j)}_s \cdot  N^{x,y,\eta(i)}_r  \approx N^{x,y,\eta(j)}_s \cdot  N^{x,i,k}_{[(a_i,a_k) \mapsto a_i \mid r]} \approx N^{x,j,k}_{[(a_j,a_k) \mapsto a_j \mid s]} \cdot  N^{x,i,k}_{[(a_i,a_k) \mapsto a_i \mid r]} 
\]
and similarly we get
\[
N^{x,y,\eta(i)}_r  \cdot N^{x,y,\eta(j)}_s \approx N^{x,i,k}_{[(a_i,a_k) \mapsto a_i \mid r]} \cdot N^{x,j,k}_{[(a_j,a_k) \mapsto a_j \mid s]} \;.
\]
Using the triangle inequality and \Cref{eq:ar-0}, we get for all $x \in \cal{X}$ and $i, j, k \in [T]$,
\[
N^{x,j,k}_{[(a_j,a_k) \mapsto a_j \mid s]} \cdot  N^{x,i,k}_{[(a_i,a_k) \mapsto a_i \mid r]}  \approx N^{x,i,k}_{[(a_i,a_k) \mapsto a_i \mid r]} \cdot N^{x,j,k}_{[(a_j,a_k) \mapsto a_j \mid s]} 
\]
Fix an arbitrary $k \in [T]$ and define
\[
	N^{x,i}_r = N^{x,i,k}_{[(a_i,a_k) \mapsto a_i \mid r]} \;.
\]
Fix an $x \in \cal{X}$. We invoke the Pasting Lemma (\Cref{lem:pasting}) on the set of measurements $\{ N^{x,i} \}_{i \in [T]}$, and obtain a projective measurement $M^x = \{M^x_a\}$ with outcomes in $\{0,1\}^{T}$ such that for all $i \in [T]$,
\[
	M^x_{[a \mapsto a_i \mid r]} \approx N^{x,i}_r~.
\]
Furthermore, by the triangle inequality, for all $y \in \cal{X}$ we have that
\begin{equation}
\label{eq:ar-1}
M^x_{[a \mapsto a_i \mid r]} \approx N^{x,y,\eta(i)}_r \;.
\end{equation}
Via the same arguments as above we have that $N^{x,i}_r \approx N^{y,x,\lambda(i)}_r$, which means that
\[
M^x_{[a \mapsto a_i \mid r]} \approx N^{y,x,\lambda(i)}_r \;.
\]

\paragraph{Constructing the $M^{x,y}_\pi$ measurements.} Fix $x,y \in \cal{X}$ and $i,j,k \in [L]$. Using \Cref{lem:ar-utility} with $A = N^{x,y,i,j,k}$, $B^1 = N^{x,y,i}$, $B^2 = N^{x,y,j}$, and $B^3 = N^{x,y,k}$ we get that the product of $N^{x,y,i}_r$, $N^{x,y,j}_s$, and $N^{x,y,k}_t$ (using any ordering) is close to $N^{x,y,i,j,k}$. 

In particular, we have
\[
	N^{x,y,i}_r \cdot N^{x,y,j}_s \approx N^{x,y,j}_s \cdot N^{x,y,i}_r \;.
\]
Using the Pasting Lemma on the set of measurements $\{N^{x,y,i}\}$, we obtain a projective measurement $M^{x,y} = \{M^{x,y}_\pi\}$ with outcomes in $\{0,1\}^R$ (i.e. proof strings) such that 
\[
	M^{x,y}_{[\pi \mapsto \pi_i \mid r]} \approx N^{x,y,i}_r~.
\]

Using \Cref{lem:add-a-measurement} repeatedly, we get that for all $i,j,k \in [L]$,
\begin{align*}
M^{x,y}_{[\pi \mapsto \pi_i \mid r]} \cdot M^{x,y}_{[\pi \mapsto \pi_j \mid s]} \cdot M^{x,y}_{[\pi \mapsto \pi_k \mid t]} &\approx N^{x,y,i}_r \cdot M^{x,y}_{[\pi \mapsto \pi_j \mid s]} \cdot M^{x,y}_{[\pi \mapsto \pi_k \mid t]} \\
&\approx N^{x,y,i}_r \cdot N^{x,y,j}_s \cdot M^{x,y}_{[\pi \mapsto \pi_k \mid t]}  \\
&\approx N^{x,y,i}_r \cdot N^{x,y,j}_s \cdot N^{x,y,k}_t \\
&\approx N^{x,y,i,j,k}_{r,s,t}
\end{align*}
where the last approximation follows from our earlier application of \Cref{lem:ar-utility}. Since $M^{x,y}_\pi$ is projective, we have that
\begin{equation}
\label{eq:ar-1-5}
M^{x,y}_{[\pi \mapsto (\pi_i,\pi_j,\pi_k) \mid (r,s,t)]} \approx N^{x,y,i,j,k}_{r,s,t} \;.
\end{equation}

We now relate the $M^{x,y}$ measurements to the $M^x$ measurements constructed previously. Using the triangle inequality with \Cref{eq:ar-1} we get for all $x,y \in \cal{X}$ and $j \in [T]$,
\begin{equation}
\label{eq:ar-2}
M^{x,y}_{[\pi \mapsto \pi_{\eta(j)} \mid r]} \approx M^x_{[a \mapsto a_j\mid r]}
\end{equation}
and similarly
\begin{equation}
\label{eq:ar-3}
M^{x,y}_{[\pi \mapsto \pi_{\lambda(j)} \mid r]} \approx M^y_{[a \mapsto a_j\mid r]} \;.
\end{equation}

Before proceeding we prove a utility lemma that allows us to argue that if all the marginalizations of projective measurements are close, then the original measurements must be close. 
\begin{lemma}
\label{lem:ar-utility2}
Let $A$ and $B$ be projective measurements with outcomes in $\{0,1\}^K$ such that for all $i \in [K]$, we have $A_{[r \mapsto r_i]} \approx_\kappa B_{[r \mapsto r_i]}$. Then 
\[
	A_r \approx_{K \kappa} B_r \;.
\]
\end{lemma}
\begin{proof}
	We prove this inductively on the prefix length of $r$. For the base case $t = 1$, we have that $A_{[r \mapsto r_1]} \approx_\kappa B_{[r \mapsto r_1]}$ by assumption. Let the inductive hypothesis be that for some $t \geq 1$, $A_{[r \mapsto r_{\leq t}]} \approx_{t\kappa} B_{[r \mapsto r_{\leq t}]}$ where $r_{\leq t}$ denotes the first $t$ bits of $r$. Then using \Cref{lem:add-a-measurement} twice, we get that
	\[
	A_{[r \mapsto r_{\leq t}]} \cdot A_{[r \mapsto r_{t+1}]} \approx_{t\kappa} B_{[r \mapsto r_{\leq t}]} \cdot A_{[r \mapsto r_{t+1}]} \approx_{\kappa} B_{[r \mapsto r_{\leq t}]} \cdot B_{[r \mapsto r_{t+1}]}
	\]
	which, via the triangle inequality, implies that
	\[
	A_{[r \mapsto r_{\leq t+1}]} \approx_{t\kappa} B_{[r \mapsto r_{\leq t+1}]}
	\]
	where we used the fact that the $A$ and $B$ measurements are projective. By induction, this statement is true for all $t$, and since $A_{[r \mapsto r_{\leq K}]} = A_r$ and $B_{[r \mapsto r_{\leq K}]} = B_r$, we conclude the proof.
\end{proof}

Applying \Cref{lem:ar-utility2} to \Cref{eq:ar-2,eq:ar-3} and interpreting the outcome of the $M^{x,y}$ measurement as a triple $(a,b,w) \in \{0,1\}^{T} \times \{0,1\}^{T} \times \{0,1\}^L$, we get
\begin{gather}
\label{eq:ar-4}
	M^{x,y}_{[(a,b,w) \mapsto a]} \approx M^x_a  \\
\label{eq:ar-5}
	M^{x,y}_{[(a,b,w) \mapsto b]} \approx M^y_b	 \;.
\end{gather}
Using \Cref{lem:add-a-measurement} several times with \Cref{eq:ar-4,eq:ar-5} we get
\begin{align*}
M^{x,y}_{[(a,b,w) \mapsto a]} \cdot M^{x,y}_{[(a,b,w) \mapsto b]} \cdot M^{x,y}_{[(a,b,w) \mapsto a]} &\approx M^x_a  \cdot M^{x,y}_{[(a,b,w) \mapsto b]} \cdot M^{x,y}_{[(a,b,w) \mapsto a]} \\
&\approx M^x_a \cdot M^y_b \cdot M^{x,y}_{[(a,b,w) \mapsto a]} \\
&\approx M^x_a \cdot M^y_b \cdot M^x_a
\end{align*}
and thus
\begin{equation}
\label{eq:ar-6}
	M^{x,y}_{[(a,b,w) \mapsto (a,b)]} \approx M^x_a \cdot M^y_b \cdot M^x_a\;.
\end{equation}

\paragraph{Evaluating the probability of success of the $M^x$ measurements.} Define the tracial synchronous strategy $\strategy = (\tau,\{M^x\})$ for game $G$. Its success probability can be lower-bounded as follows:
\begin{align*}
	\omega(G,\strategy) &= \E_{x,y} \sum_{a,b} D(x,y,a,b) \cdot \tau(M^x_a \,\, M^y_b) \\
	&= \E_{x,y} \sum_{a,b} D(x,y,a,b) \cdot \tau(M^x_a \cdot M^y_b \cdot M^x_a) \\
	&= \E_{x,y} \sum_{a,b} D(x,y,a,b) \cdot \Big( \tau \Big (M^{x,y}_{[(a,b,w) \mapsto (a,b)]} \Big) +  \tau \Big(M^{x,y}_{[(a,b,w) \mapsto (a,b)]} - M^x_a \,\,M^y_b \,\,M^x_a \Big) \Big) \\
	&\geq \E_{x,y} \sum_{a,b} D(x,y,a,b) \cdot \tau \Big (M^{x,y}_{[(a,b,w) \mapsto (a,b)]} \Big) - \E_{x,y} \sum_{a,b} \Big |  \tau\Big(M^{x,y}_{[(a,b,w) \mapsto (a,b)]} - M^x_a\,\, M^y_b \,\,M^x_a\Big) \Big| 
\end{align*}

We bound the second term first. From \Cref{lem:consistency-consequences} applied to \Cref{eq:ar-6} we get that $M^{x,y}_{[(a,b,w) \mapsto (a,b)]} \simeq_{\delta} M^x_a \cdot M^y_b \cdot M^x_a$ for some proper error function $\delta = \delta(\eps)$. We then apply \Cref{lem:consistency-to-prob-closeness} to get that
\[
\E_{x,y} \sum_{a,b} \Big |  \tau\Big(M^{x,y}_{[(a,b,w) \mapsto (a,b)]} - M^x_a \,\,M^y_b\,\, M^x_a\Big) \Big|  \leq 2 \delta~.
\]

Next, we evaluate
\begin{align*}
&\E_{x,y} \sum_{a,b} D(x,y,a,b) \cdot \tau \Big (M^{x,y}_{[(a,b,w) \mapsto (a,b)]} \Big) \\
&= \E_{x,y} \sum_{a,b,w } D(x,y,a,b) \cdot  \tau \Big (M^{x,y}_{a,b,w} \Big) \\
&= \E_{x,y} \sum_{a,b,w } \mathbf{1}[\text{$\exists w'$ : $(a,b,w')$ satisfies $\varphi_{x,y}$}] \cdot  \tau \Big (M^{x,y}_{a,b,w} \Big) \\
&\geq \E_{x,y} \sum_{a,b,w } \mathbf{1}[\text{$(a,b,w)$ satisfies $\varphi_{x,y}$}] \cdot  \tau \Big (M^{x,y}_{a,b,w} \Big) \\
&= 1 - \E_{x,y} \sum_{a,b,w} \mathbf{1}[\text{$(a,b,w)$ does not satisfy $\varphi_{x,y}$}] \cdot  \tau \Big (M^{x,y}_{a,b,w} \Big)
\end{align*}
where in the second line we use the conclusion of \Cref{thm:cook-levin} that since $\TIME_D \leq T$, we have $D(x,y,a,b) = 1$ if and only if there exists a satisfying assignment $(a,b,w')$ for the Cook-Levin formula $\varphi_{x,y}$.

Via the union bound, the probability that $\pi = (a,b,w)$ does not satisfy $\varphi_{x,y}$ is at most the sum, over all $i,j,k \in [L]$, that $(\pi_i,\pi_j,\pi_k)$ does not satisfy a clause in $\varphi_{x,y}$ (if there exists such a clause). Thus we have
\[
\E_{x,y} \sum_{a,b,w} \mathbf{1}[\text{$(a,b,w)$ unsat. $\varphi_{x,y}$}] \cdot  \tau \Big (M^{x,y}_{a,b,w} \Big) \leq \E_{x,y} \sum_{i,j,k} \sum_{\pi} \mathbf{1}[\text{$(\pi_i,\pi_j,\pi_k)$ unsat. $\varphi_{x,y}$}] \cdot \tau \Big (M^{x,y}_\pi \Big)
\]
We can now relate this quantity to the success probability of $\strategy^\ans$ in the answer-reduced game $G^\ans$. Let $\theta$ denote the probability that one of the players receives a question $\hat{x} = (g,p)$ of the form $g = (x,y)$ and $p = (i,j,k)$, and the other player receives a question $\hat{y} = (g',p')$ of the form $g' = x$ and $p \in \{i,j,k\}$. In this situation, by the design of the decider (see \Cref{sec:ar-decider}), the verifier checks whether the player who got question $\hat{x}$ responds with proof bits $(\pi_i,\pi_j,\pi_k)$ that satisfy a corresponding clause in $\varphi_{x,y}$. Thus, since the overall success probability of the strategy $\strategy^\ans$ in the game $G^\ans$ is at least $1 - \eps$, it must be that conditioned on a player receiving question of the form $\hat{x} = (x,y,i,j,k)$, their answer does not satisfies a corresponding clause in the formula $\varphi_{x,y}$ (if one exists) with probability at most $\eps/\theta$. In other words:
\[
	\E_{x,y,i,j,k} \sum_{\pi_i,\pi_j,\pi_k} \mathbf{1}[\text{$(\pi_i,\pi_j,\pi_k)$ unsat. $\varphi_{x,y}$}] \cdot \tau( N^{x,y,i,j,k}_{\pi_i,\pi_j,\pi_k}) \leq \eps/\theta.
\]
Multiplying both sides by $L^3$, we get that 
\[
\E_{x,y} \sum_{i,j,k} \sum_{\pi_i,\pi_j,\pi_k} \mathbf{1}[\text{$(\pi_i,\pi_j,\pi_k)$ unsat. $\varphi_{x,y}$}] \cdot \tau( N^{x,y,i,j,k}_{\pi_i,\pi_j,\pi_k}) \leq L^3 \eps/\theta \;.
\]
Using \Cref{lem:consistency-consequences} with \Cref{eq:ar-1-5}, we get that for every $i,j,k \in [L]$ and on average over $x,y$, 
\[
M^{x,y}_{[\pi \mapsto (\pi_i,\pi_j,\pi_k) \mid r,s,t]} \simeq_{\nu} N^{x,y,i,j,k}_{r,s,t}
\]
for some proper error function $\nu = \nu(\eps)$. Then using \Cref{lem:consistency-to-prob-closeness} we get that
\[
\E_{x,y} \sum_{r,s,t} \Big | \tau \Big ( M^{x,y}_{[\pi \mapsto (\pi_i,\pi_j,\pi_k) \mid r,s,t]} - N^{x,y,i,j,k}_{r,s,t} \Big) \Big| \leq 2\nu
\]
for every $i,j,k \in [L]$. 
Putting everything together, we find that
\begin{align*}
	&\E_{x,y} \sum_{i,j,k} \sum_{\pi} \mathbf{1}[\text{$(\pi_i,\pi_j,\pi_k)$ unsat. $\varphi_{x,y}$}] \cdot \tau \Big (M^{x,y}_\pi \Big) \\
	&\leq \E_{x,y} \sum_{i,j,k} \sum_{\pi_i,\pi_j,\pi_k} \mathbf{1}[\text{$(\pi_i,\pi_j,\pi_k)$ unsat. $\varphi_{x,y}$}] \cdot \tau( N^{x,y,i,j,k}_{\pi_i,\pi_j,\pi_k}) + 2\nu \\
	&\leq L^3 \Big( \frac{\eps}{\theta} + 2\nu \Big)\;.
\end{align*}
Let $\zeta = L^3 \Big( \frac{\eps}{\theta} + 2\nu \Big) + 2\delta$. Then we deduce that
\[
	\omega(G,\strategy) \geq 1 - \zeta.
\]
Since $\delta,\nu$ are proper error functions of $\eps$, so is $\zeta$. Thus $\zeta \to 0$ as $\eps \to 0$. Furthermore, the strategy $\strategy$ is finite-dimensional if and only if $\strategy^\ans$ is finite-dimensional. Thus, suppose that $\val^s_t(G^\ans) = 1$ for $t = q$ (resp. for $t = co$). This implies that there is a sequence of finite-dimensional (resp. commuting operator) strategies $\strategy^\ans$ such that $\val(G^\ans,\strategy^\ans)$ approaches $1$. This in turn implies the existence of a sequence of finite-dimensional (resp. commuting operator) strategies $\strategy$ such that $\val(G,\strategy)$ approaches $1$, and thus $\val^s_t(G) = 1$. Taking the contrapositive, we conclude that
\[
	\val^s_t(G) < 1 \Longrightarrow \val^s_t(G^\ans) < 1~.
\]
This finishes the proof of the Proposition.
\end{proof}

\subsection{Proof of \Cref{thm:answer-reduction}}
\label{sec:answer-reduction-proof}

We now prove the main result of this section, \Cref{thm:answer-reduction}. Fix $\beta \in \N$. The algorithm $\alg{AnswerReduction}_{\beta}$, on input $(D,C)$ where $D$ is a $5$-input Turing machine and $C$ is a $3$-input Turing machine, computes the descriptions of $5$-input and $3$-input Turing machines $D^\ans,C^\ans$ respectively as follows. Let $Q(n) = \log^\beta n$ and $T(n) = n^\beta$.

\medskip

\noindent \underline{Question checker $C^\ans$}. At a high level, the Turing machine $C^\ans$, on input $(n,\hat{x},\hat{y})$ checks whether the question pair $(\hat{x},\hat{y})$ is nontrivial according to \Cref{tab:answer-reduction}, where ``$G$'' in the table is supposed to be the $n$-th game $G_n$ of the sequence specified by the verifier $\verifier = (D,C)$, ``$D_{x,y}$'' in the table is supposed to be the Turing machine $D_{n,x,y}$ which on input $(a,b)$ outputs $D(n,x,y,a,b)$, and ``$T$'' in the table is supposed to be $T(n)$.

In order to compute whether $(\hat{x},\hat{y})$ (or $(\hat{y},\hat{x})$) is one of the question pairs specified by \Cref{tab:answer-reduction}, the Turing machine $C^\ans$ has to compute the question lengths of the $n$-th answer-reduced game $G^\ans$: it computes $L_n$, the number of variables of a Cook-Levin formula corresponding to a Turing machine with description length $|D| + O(\log n) + 2Q(n)$. (This is the description length of a Turing machine $D_{n,x,y}$, which is $D$ with $(n,x,y)$ ``hardwired'' into it.) It then checks whether $\hat{x},\hat{y}$ are (binary encodings of) elements of $(\{0,1\}^{Q(n)} \cup \{0,1\}^{2Q(n)}) \times ([L_n] \cup [L_n]^2 \cup [L_n]^3)$, which is the question alphabet of $G_n^\ans$. It not, then it outputs $0$. At this point, the Turing machine $C^\ans$ has ensured that $(\hat{x},\hat{y})$ is a properly-formatted question pair in the $n$-th answer-reduced game $G_n^\ans$.

The Turing machine $C^\ans$ then attempts to parse $(\hat{x},\hat{y})$ or $(\hat{y},\hat{x})$ as one of the combinations specified in \Cref{tab:answer-reduction} and outputs $1$ if there is a match; otherwise it outputs $0$. To determine whether $(x,y) \in (\{0,1\}^{Q(n)})^2$ is nontrivial for $G_n$, it computes whether $C(n,x,y) = 1$. This concludes the description of $C^\ans$.

\medskip

\noindent \underline{Decider $D^\ans$}. The Turing machine $D^\ans$ on input $(n,\hat{x},\hat{y},\hat{a},\hat{b})$ first computes $C^\ans(n,\hat{x},\hat{y})$. If the output is $0$ (i.e. the question pair $(\hat{x},\hat{y})$ is trivial), then the Turing machine $D^\ans$ accepts (i.e. outputs $1$). Otherwise, it continues. It computes $L_n$ just like with $C^\ans$, and then matches $(\hat{x},\hat{y})$ (resp. $(\hat{y},\hat{x})$) to one of the entries of the table. Since $C^\ans(n,\hat{x},\hat{y}) = 1$, there must be a match. The Turing machine $D^\ans$ then evaluates whether the winning conditions $(\hat{a},\hat{b})$ (resp. $(\hat{b},\hat{a})$) are satisfied according to \Cref{tab:answer-reduction}. If the winning conditions are satisfied, then $D^\ans$ outputs $1$ (accepts), otherwise it outputs $0$ (rejects).

\medskip
\medskip
Now assume the conditions of \Cref{thm:answer-reduction}; i.e., that $\verifier = (D,C)$ is a verifier for a sequence of games $\UGS_\verifier = (G_n)_{n \in \N}$ and
\begin{enumerate}
	\item The questions of $G_n$ have length at most $Q(n)$,
	\item $\TIME_C(n) \leq Q(n)$, and
	\item $\TIME_D(n) \leq T(n)$.
\end{enumerate}
Now we argue that the output $\verifier^\ans = (D^\ans,C^\ans)$ is a verifier for a sequence of games $\UGS_{\verifier^\ans} = (G_n^\ans)_{n \in \N}$ satisfying the conclusions of \Cref{thm:answer-reduction}.

\paragraph{Complexity of the question checker $C^\ans$.} 
The question checker $C^\ans$ for the answer-reduced game first has to compute $L_n$, the number of variables in the Cook-Levin formula corresponding to $D_{n,x,y}$. This requires computing the description length of $D_{n,x,y}$, where $x,y$ are questions in the original game $G_n$, which by assumption has length at most $Q(n)$. It then has to check that the questions $(\hat{x},\hat{y})$ are properly formatted questions from the question alphabet of $G_n^\ans$, which takes time $\poly(Q(n),\log L_n)$. Then, it has to determine whether $(\hat{x},\hat{y})$ matches one of the question pairs in \Cref{tab:answer-reduction}, which includes running the question checker $C$ for the original verifier $\verifier$. Thus overall we have $\TIME_{C^\ans}(n) \leq \poly(|D|,|C|,Q(n),\log T(n),\log n) = \poly(|D|,|C|,\beta, \log^\beta n)$.

\paragraph{Complexity of the decider $D^\ans$.} The time complexity of the answer-reduced verifier $D^\ans$ includes the complexity of computing the question checker $C^\ans(n,x,y)$ and computing the number of variables $L_n$. It also includes the complexity of computing a clause of the Cook-Levin formula $\varphi_{n,x,y}$, which involves invoking the algorithm $\alg{CookLevin}$ on the input $(D_{n,x,y},T(n),2T(n),i,j,k)$ for some variable indices $i,j,k \in [L_n]$, where $(x,y)$ are questions for the original game $G_n$ (which have length $Q(n)$ by assumption). Computing the description of $D_{n,x,y}$ takes time $\poly(|D|,|x|,|y|,\log n)$ because it involves ``hard-wiring'' the integer $n$ and strings $x$, $y$ into the description of $D$. Thus it takes at most $\poly(|D|,Q(n),\log T(n),\log n)$ to compute a clause. Computing the $\eta(\cdot)$ and $\lambda(\cdot)$ maps also take time at most $\poly(\log T(n))$ (because it requires computing $T(n)$). Thus, in total, the complexity of the answer-reduced verifier is $\poly(|D|,|C|,Q(n),\log T(n),\log n) = \poly(|D|,|C|,\beta, \log^\beta n)$.

\paragraph{Completeness and Soundness.} Completeness follows from \Cref{prop:answer-reduction-completeness}. Soundness follows from \Cref{prop:answer-reduction-soundness}.

\medskip

\noindent This completes the proof of \Cref{thm:answer-reduction}.

\newpage
\section{Compression of nonlocal games and their applications}
\label{sec:compression}

In this section we describe the compression theorems and some of their applications. 

\subsection{Gapless compression}
First we present the main technical result of this paper, which is a gapless compression theorem for both the quantum and commuting operator value of nonlocal games. This theorem statement is a formalization of \Cref{thm:gapless-comp-informal} from the introduction.

\begin{theorem}[Gapless compression of nonlocal games]
\label{thm:gapless-comp-formal}
	For all $\alpha \in \N$ there is a polynomial time algorithm $\alg{GaplessCompress}_{\alpha}$ that takes as input a pair of Turing machines $(D,C)$ and outputs a pair of Turing machines $(D',C')$ such that the following holds. If $\verifier =(D,C)$ is a verifier for a sequence of games $\UGS_{\verifier} = (G_n)_{n \in \N}$ and $n_0 \in \N$ is an integer such that for all $n \geq n_0$,
	\begin{equation}
	\label{eq:gapless-time-bound}
    	\max \Big \{ \TIME_C(n), \TIME_D(n) \Big \} \leq n^\alpha~,
	\end{equation}
    then $\verifier' = (D',C')$ is a verifier for a sequence of games $\UGS_{\verifier'} = (G_n')_{n \in \N}$ with the following properties. There exist an integer $\gamma = \poly(\alpha)$ and $n_0' = \poly(\gamma,n_0)$ such that for all $n \geq n_0'$,
    \begin{enumerate}
        \item (Complexity bounds)
        \begin{gather*}
			\max \left \{ \TIME_{C'}(n), \TIME_{D'}(n) \right \} \leq \log^\gamma n~.
		\end{gather*}
        \item (Completeness) For all oracularizable synchronous strategies $\strategy$ for $G_n$, there exists an oracularizable synchronous strategy $\strategy'$ for $G_n'$ such that
        \[
            \val(G_n',\strategy') \geq \frac{1}{2} + \frac{1}{2} \val(G_n,\strategy)~.
        \]
        Furthermore, if $\strategy$ is finite dimensional, so is $\strategy'$. 
        \item (Soundness) For all $t \in \{q,co\}$ we have 
        \[\val_t^s(G_n) < 1 \Longrightarrow \val_t^s(G_n') < 1~.\]
        \item (Entanglement bound)
        \[
        	\mathcal{E} (G_n',1) \geq \max \left \{ \mathcal{E} (G_n,1) , 2^{2n} \right \}~.
        \]
    \end{enumerate}
\end{theorem}

We prove this by combining the question reduction and answer reduction algorithms of \Cref{sec:question-reduction,sec:answer-reduction}. The algorithm $\alg{GaplessCompress}_\alpha$ is presented below. The parameter $\beta$ in \Cref{alg:gapless-compress} is defined to be the same $\beta = \poly(\alpha)$ from \Cref{thm:question-reduction}. 

\vspace{10pt}
\IncMargin{1em}
\begin{algorithm}[H]
\DontPrintSemicolon

\textbf{Input}: $D,C$.

Compute $(D^{\intro}, C^{\intro}) = \alg{QuestionReduction}_{\alpha}(D,C)$. 

Compute $(D', C') = \alg{AnswerReduction}_{\beta}(D^{\intro}, C^{\intro})$. 

Return $(D', C')$.

\caption{$\alg{GaplessCompress}_\alpha$} 
\label{alg:gapless-compress}
\end{algorithm}\DecMargin{1em}
\vspace{10pt}

\begin{proof}
First, it is clear that $\alg{GaplessCompress}_\alpha$ runs in polynomial time in the description length of the input $(D,C)$, because the algorithm $\alg{QuestionReduction}_\alpha$ runs in time $\poly(|D|,|C|)$ and the algorithm $\alg{AnswerReduction}_{\beta}$ runs in time $\poly(|D^\intro|,|C^\intro|) = \poly(|D|,|C|)$. This last equality uses that $\max \{ |D^\intro|, |C^\intro| \} \leq \poly(|D|,|C|)$ because the running time of $\alg{QuestionReduction}_{\alpha}$ is an upper bound on the length of the descriptions of $D^\intro$ and $C^\intro$. 

Next, suppose that $\verifier = (D,C)$ is such that the time bound of~\eqref{eq:gapless-time-bound} is satisfied. Then, the complexity bounds on $(D^\intro,C^\intro)$ given by the conclusion of \Cref{thm:question-reduction} are exactly those that satisfy the conditions of \Cref{thm:answer-reduction}. Thus, the output $(D',C')$ of $\alg{AnswerReduction}_\beta(D^\intro,C^\intro)$ satisfy the conclusions of \Cref{thm:answer-reduction} (with $\gamma = \poly(\beta) = \poly(\alpha)$) and thus this establishes the desired complexity bounds on the output verifier $\verifier'$. 

Define the integers $\beta = \poly(\alpha)$,$n_0^\intro = \poly(\beta,n_0)$ as given by \Cref{thm:question-reduction}. Then, define the integers $\gamma = \poly(\beta)$, $n_0^\ans = \poly(\gamma^\gamma,n_0^\intro) = \poly(\alpha^\alpha,n_0)$ as given by \Cref{thm:answer-reduction}. Define $n_0' = n_0^\ans$. We assume that $n \geq n_0'$ implies that $n \geq \max \{ n_0, n_0^\intro \}$ as well.

We now establish the completeness property of $\verifier'$. Let $n \geq n_0'$. Let $\strategy$ be an oracularizable synchronous strategy for $G_n$. By the completeness of Question Reduction, this implies there is an oracularizable synchronous strategy $\strategy^{intro}$ for $G_n^{\intro}$ such that
\[
	\val(G^{\intro}_n,\strategy^{intro}) \geq \val(G_n,\strategy)~.
\]
Then, by the completeness of Answer Reduction, there is an oracularizable synchronous strategy $\strategy'$ for $G_n'$ such that
\[
\val(G_n',\strategy') \geq \frac{1}{2} + \frac{1}{2}\val(G^{\intro}_n,\strategy^{intro}) \geq \frac{1}{2} + \frac{1}{2}\val(G_n,\strategy)~.
\]
Furthermore, if $\strategy$ is finite-dimensional, then so are $\strategy^\intro$ and $\strategy'$. 

We establish the soundness property of $\verifier'$ by combining the soundness guarantees of Question Reduction and Answer Reduction:
\[
    \val_t^s(G_{n}) < 1 \Longrightarrow \val_t^s(G_n^{\intro}) < 1 \Longrightarrow \val_t^s(G_n') < 1.
\]

Finally, we establish the entanglement bound property by combining the entanglement bounds from Question Reduction and Answer Reduction 
\[ 
\mathcal{E} (G_n',1) \geq \mathcal{E} (G_n^{intro},1) \geq \max \left \{ \mathcal{E} (G_n,1) , 2^{2n} \right \}~.
\]

\end{proof}

\subsection{Super compression} \label{subsec:super-comp}

The gapless compression procedure of \Cref{thm:gapless-comp-formal} transforms uniform sequences of games $(G_1,G_2,\ldots)$ to another uniform sequence $(G_1',G_2',\ldots)$ that is, in a sense, exponentially more efficient. Using this we prove a \emph{super compression} procedure, which transforms a sequence of games $(G_1,G_2,\ldots)$ into a \emph{single} game $G'$ such that $\val_t^s(G') = 1$ if and only if $\val_t^s(G_n) = 1$ for every sufficiently large $n$ and $t \in \{q, co \}$. 

\begin{theorem}[Super compression of nonlocal games]
\label{thm:super-comp}
	For all $\alpha \in \N$ there is a polynomial time algorithm $\alg{SuperCompress}_{\alpha}$ that takes as input a pair of Turing machines $(D,C)$ and outputs a pair of Turing machines $(D^\super,C^\super)$ such that the following holds. If $\verifier =(D,C)$ is a verifier for a sequence of games $\UGS_{\verifier} = (G_n)_{n \in \N}$ and $n_0 \in \N$ is an integer such that for all $n \geq n_0$,
	\begin{equation}
	\label{eq:super-comp-time-bound}
    	\max \Big \{ \TIME_C(n), \TIME_D(n) \Big \} \leq n^\alpha~,
	\end{equation}
    then $\verifier^\super = (D^\super,C^\super)$ is a verifier for a sequence of games $\UGS_{\verifier^\super} = (G_n^\super)_{n \in \N}$ such that there exist integers $\lambda = O(\alpha)$ and $\kappa = \poly(|D|,|C|,\alpha,n_0,\lambda^{\poly(\lambda)})$ and the $\kappa$-th game in the sequence, $G_\kappa^\super$, satisfies the following properties:
    \begin{enumerate}
        \item (Complexity bounds) \begin{gather*}
			\max \left \{ \TIME_{C^\super}(\kappa), \TIME_{D^\super}(\kappa) \right \} \leq \kappa^\lambda~.
		\end{gather*}
\item (Completeness for $t = q$) If for all $n \geq \kappa$ we have
        \[
        	\sup_{\text{finite-dim osync } \strategy_n} \val(G_n,\strategy_n) = 1
        \]
        where the supremum is over finite-dimensional oracularizable synchronous strategies $\strategy_n$, then $\val_q^s(G_\kappa^\super) = 1$. 
 		\item (Completeness for $t = co$) If for all $n \geq \kappa$, there exists an oracularizable synchronous strategy $\strategy_n$ for $G_n$ such that $\val(G_n,\strategy_n) = 1$, then $\val_{co}^s(G_\kappa^\super) = 1$. 
		      
        \item (Soundness) For all $t\in \{q,co\}$, if there exists an $n \geq \kappa$ such that $\val_t^s(G_n) < 1$, then $\val_t^s(G_\kappa^\super) < 1$.
        
        \item (Entanglement lower bound) There is no finite-dimensional strategy $\strategy^\super_\kappa$ such that $\val(G_\kappa^\super,\strategy^\super_\kappa) = 1$.
    \end{enumerate}
\end{theorem}

Note that, unlike \Cref{thm:gapless-comp-formal}, the conclusions of \Cref{thm:super-comp} pertain to a \emph{single} game in the output sequence $\UGS_{\verifier^\super} = (G_n^\super)_n$ of games, namely, $G_\kappa^\super$.

At a high level, the games $(G_n^\super)_n$ has the following structure: with probability $\frac{1}{2}$, the verifier in the game $G_n^\super$ plays the game $G_n$. With the other probability the verifier plays the game $G_{n+1}'$ where $(G_n')_n$ is the compression of $(G_n^\super)_n$ using $\alg{GaplessCompress}$ from \Cref{thm:gapless-comp-formal}. Note the self-referentiality! We now proceed with the proof.

\begin{proof}

Let $(D,C)$ be a pair of Turing machines and let $\alpha$ be such that \cref{eq:gapless-time-bound} is satisfied. We first define, for every integer $\lambda \in \N$, a pair of Turing machines $(D^\super_{\lambda},C^\super_{\lambda})$ whose descriptions are given below in \Cref{alg:d-super,alg:c-super}. We will then identify a special $\lambda^*$ and define the algorithm $\alg{SuperCompress}_{\alpha}$ to output the descriptions of $(D^\super_{\lambda^*},C^\super_{\lambda^*})$. 

Note that the descriptions of $D^\super_{\lambda},C^\super_{\lambda}$ are self-referential: they perform computations on \emph{their own descriptions}. It is possible to define Turing machines in this manner; one can appeal to either Kleene's Recursion Theorem/Roger's Fixed Point Theorem to argue that these descriptions are well-defined (see, e.g.\,~\cite[Chapter 14]{jones1997computability} for a modern explanation).
The description lengths of these Turing machines satisfy 
\[
	\max \{ |D^\super_{\lambda}|, |C^\super_{\lambda}| \} \leq \poly(\lambda,|D|,|C|)~.
\]

\vspace{10pt}
\IncMargin{1em}
\begin{algorithm}[H]
\DontPrintSemicolon

\textbf{Input}: $n, x, y, a,b$

If the following takes more than $n^\lambda$ steps, then automatically reject.

Parse $x = (t_x,\hat{x})$ and $y = (t_y,\hat{y})$, where $t_x,t_y \in \{0,1\}$. 

\If{$t_x = t_y = 0$}{ 
	If $D(n,\hat{x},\hat{y},a,b)$ accepts, then accept. Otherwise, reject.
}
\ElseIf{$t_x = t_y = 1$} {

Compute $(D',C') = \alg{GaplessCompress}_{\lambda}(D^\super_{\lambda},C^\super_{\lambda})$.

If $D'(n+1,\hat{x},\hat{y},a,b)$ accepts, then accept. Otherwise, reject.

} 

On all other inputs, accept.

\caption{Specification of Turing machine $D^\super_{\lambda}$.}
\label{alg:d-super}
\end{algorithm}\DecMargin{1em}
\vspace{10pt}

\vspace{10pt}
\IncMargin{1em}
\begin{algorithm}[H]
\DontPrintSemicolon

\textbf{Input}: $n, x, y$

If the following takes more than $n^\lambda$ steps, then automatically reject.

Parse $x = (t_x,\hat{x})$ and $y = (t_y,\hat{y})$, where $t_x,t_y \in \{0,1\}$. 

\If{$t_x = t_y = 0$}{ 
	Output $C(n,\hat{x},\hat{y})$.
}
\ElseIf{$t_x = t_y = 1$} {

Compute $(D',C') = \alg{GaplessCompress}_{\lambda}(D^\super_{\lambda},C^\super_{\lambda})$.

Output $C'(n+1,\hat{x},\hat{y})$.

} 

On all other inputs, output $1$.

\caption{Specification of Turing machine $C^\super_{\lambda}$.}
\label{alg:c-super}
\end{algorithm}\DecMargin{1em}
\vspace{10pt}

First, observe that by construction both $D^\super_{\lambda}$ and $C^\super_{\lambda}$, when given index $n$, run in time at most $n^\lambda$. Thus, $(D^\super_{\lambda},C^\super_{\lambda})$ satisfy the complexity conditions of \Cref{thm:gapless-comp-formal} for the algorithm $\alg{GaplessCompress}_{\lambda}$, and thus the output Turing machines $(D',C')$ satisfy the complexity bounds in the conclusion of $\alg{GaplessCompress}_{\lambda}$, namely, that there exists $\gamma = \poly(\lambda)$ such that for all $n \in \N$, 
\[
	\max \{ \TIME_{D'}(n), \TIME_{C'}(n) \} \leq \log^\gamma n~.
\]

The next claim shows that we can find an integer $\lambda^*$ such that for sufficiently large $n$, the Turing machines $D^\super_{\lambda^*},C^\super_{\lambda^*}$ never encounter the time-out.

\begin{claim}
There exist integers $\lambda^* = O(\alpha),\kappa = \poly(|D|,|C|,\alpha,n_0,\lambda^{\poly(\lambda)})$ such that for all $n \geq \kappa$, the Turing machines $D^\super_{\lambda^*},C^\super_{\lambda^*}$ when given index $n$ never reject due to exceeding the $n^{\lambda^*}$ time-out. 
\end{claim}
\begin{proof}

Next, the time complexity of $D^\super_{\lambda}$ (resp. $C^\super_\lambda$) \emph{without} the automatic $n^\lambda$ timeout is polynomial in the complexity of running the decider $D$/checker $C$, computing $\alg{GaplessCompress}_{\lambda}$, and running the decider $D'$ (resp. checker $C'$). By our assumptions on $(D,C)$, when $n \geq n_0$ we have the bounds from \cref{eq:gapless-time-bound}. The algorithm $\alg{GaplessCompress}_{\lambda}$ runs in time $\poly(|D^\super_{\lambda}|,|C^\super_{\lambda}|,\lambda) = \poly(|D|,|C|,\lambda)$. Putting this together with the complexity bounds on $D'$ (resp. $C'$), we have that the complexity of $D^\super_{\lambda}$ (resp. $C^\super_\lambda$), without the automatic timeout, is at most
\begin{equation}
\label{eq:super-comp-0b}
	\sigma (n^\alpha \cdot |D| \cdot |C| \cdot \lambda \cdot \log^\gamma n)^\sigma
\end{equation}
for all $n \geq n_0$, where $\sigma \in \N$ is some universal constant.

We can find integers $\lambda^*, \kappa \in \N$ such that each component of the expression in~\eqref{eq:super-comp-0b} is at most $n^{\lambda^*}$ for all $n \geq \kappa$. Namely:
\begin{itemize}
	\item By taking $\lambda^* \geq \sigma \cdot \alpha$ and $\kappa \geq \sigma$, we have that $\sigma n^{\alpha \cdot \sigma} \leq n^{\lambda^*}$ for all $n \geq \kappa$.
	\item By taking $\lambda^* \geq \sigma$ and $\kappa \geq |D| \cdot |C| \cdot \lambda^*$, we have that $(|D| \cdot |C| \cdot \lambda^*)^\sigma \leq n^{\lambda^*}$ for all $n \geq \kappa$.
	\item By taking $\lambda^* \geq 2$ and $\kappa \geq (\gamma \cdot \sigma)^{\gamma \cdot \sigma}$ where $\gamma = \poly(\lambda^*)$, we have that $\log^{\gamma \cdot \sigma}(n) \leq n^{\lambda^*}$ for all $n \geq \kappa$.
\end{itemize}
Putting everything together, by setting $\lambda^* = 2\sigma\alpha$ and $\kappa = \sigma \cdot \alpha \cdot |D| \cdot |C| \cdot \lambda^* \cdot  (\gamma \cdot \sigma)^{\gamma \cdot \sigma} \cdot n_0$, we get that the Turing machines $D^\super_{\lambda^*}$ and $C^\super_{\lambda^*}$ run in time that is less than $n^{\lambda^*}$ for all $n \geq \kappa$.

\end{proof}

We define the algorithm $\alg{SuperCompress}_\alpha$, on input $(D,C)$, to compute $\lambda^* = O(\alpha)$ and output the descriptions of $(D^\super_{\lambda^*},C^\super_{\lambda^*})$. The algorithm clearly runs in polynomial time. 

By construction the Turing machines $(D^\super_{\lambda^*},C^\super_{\lambda^*})$ satisfy the desired time complexity bound on index $n = \kappa$. What remains is to argue completeness and soundness. For notational simplicity we fix $\lambda^*$ and let $(D^\super,C^\super) = (D^\super_{\lambda^*},C^\super_{\lambda^*})$.

Fix $t \in \{q,co\}$. Since the Turing machines $D^\super,C^\super$ never reject due to the time-out, we have that the verifier in the game $G_n^\super$ automatically accepts with probability $\frac{1}{2}$ (when $t_x \neq t_y$), plays the game $G_n$ with probability $\frac{1}{4}$ (when $t_x = t_y = 0$), and plays the game $G_{n+1}'$ with probability $\frac{1}{4}$ (when $t_x = t_y = 1$) where $G_{n+1}'$ is the $(n+1)$-st game in the sequence of games output by $\alg{GaplessCompress}$ on input $(D^\super,C^\super)$. 

We first prove completeness for $t = q$. Suppose for all $n \geq \kappa$ we have
\begin{equation}
\label{eq:super-comp-0}
        	\sup_{\text{finite-dim osync } \strategy_n} \val(G_n,\strategy_n) = 1.
\end{equation}
Define
\[
	c_n = \sup_{\text{finite-dim osync } \strategy_n^\super} \val(G_n^\super,\strategy_n^\super)
\]
and define $c = \inf_{n \geq \kappa} c_n$. We aim to prove that $c = 1$; this would imply that $\val_q^s(G_n^\super) = 1$ for all $n \geq \kappa$. Suppose this were not true, so that $0 \leq c < 1$. We now show that $c_n \geq \frac{7 + c}{8} > c$ for all $n \geq \kappa$, which would contradict the fact that $c$ is the infimum of the sequence $(c_n)_{n \geq \kappa}$. 

For all $m \geq \kappa$, let: (a) $\strategy_m$ be a finite-dimensional oracularizable synchronous (``finite-dim osync'') strategy for $G_m$, (b) let $\strategy_m^\super$ denote a finite-dim osync strategy for $G_m^\super$ whose value is at least $c$, and (c) let $\strategy_m'$ denote the finite-dim osync strategy for $G_m'$, given by the completeness property of \Cref{thm:gapless-comp-formal}, whose value satisfies
\begin{equation}
\label{eq:super-comp-1}
	\val(G_m',\strategy_m') \geq \frac{1}{2} + \frac{1}{2} \val(G_m^\super,\strategy_m^\super) \geq \frac{1 + c}{2}~.
\end{equation}

We now construct, for all $n \geq \kappa$, a finite-dim osync strategy $\mathscr{T}_n$ for $G_n^\super$ that has value at least 
\begin{equation}
\label{eq:super-comp-2}
	\val(G_n^\super,\mathscr{T}_n) \geq \frac{1}{2} + \frac{1}{4} \val(G_n,\strategy_n) + \frac{1}{4} \val(G_{n+1}',\strategy_{n+1}') \geq \frac{5+c}{8} + \frac{1}{4} \val(G_n,\strategy_n)
\end{equation}
where the second inequality follows from \cref{eq:super-comp-1}. The strategy $\mathscr{T}_n$ is constructed as follows. Its tracial state is the tensor product of the tracial states from $\strategy_n$ and $\strategy_{n+1}'$; since both of these strategies are finite-dimensional so is the strategy $\mathscr{T}_n$. When a player gets question $x = (0,\hat{x})$, they perform the measurement corresponding to question $\hat{x}$ from the strategy $\strategy_n$. When a player gets question $x = (1,\hat{x})$, they perform the measurement corresponding to question $\hat{x}$ from the strategy $\strategy_n'$. Thus when both players get questions whose first bit is $0$, they are essentially playing the game $G_n$, and when they both get questions whose first bit is $1$, they are essentially playing the game $G_{n+1}'$. Taking the supremum of the right-hand side of \cref{eq:super-comp-2} over finite-dim osync strategies $\strategy_n$ for $G_n$ and using \cref{eq:super-comp-0}, we get that $c_n \geq \frac{7 + c}{8}$, which yields a contradiction as desired.

The proof of completeness for $t = co$ is virutally identical, except we consider all oracularizable synchronous strategies, not just finite-dimensional ones.

We now prove the soundness property. Let $t \in \{q,co\}$. Let $n^* \geq \kappa$ be such that $\val_t^s(G_{n^*}) < 1$. For all $m \geq \kappa$, by construction of the game $G^\super_m$ we have
\[
	\val_t^s(G^\super_m) = \frac{1}{2} + \frac{1}{4} \val_t^s(G_m) + \frac{1}{4} \val_t^s(G_{m+1}')~,
\]
so therefore $\val_t^s(G^\super_{n^*}) < 1$. By the soundness property of \Cref{thm:gapless-comp-formal}, this means that  $\val_t^s(G'_{n^*}) < 1$, and therefore $\val_t^s(G^\super_{n^*-1}) < 1$. This in turn implies that $\val_t^s(G^\super_{n^*-2}) < 1$, and so on, until we obtain $\val_t^s(G^\super_\kappa) < 1$, the desired conclusion.

Finally, we prove that there is no finite-dimensional perfect strategy for $G^\super_\kappa$. Suppose for contradiction that there a $d$-dimensional strategy $\strategy^\super_\kappa$ such that $\val(G^\super_\kappa,\strategy^\super_\kappa) = 1$. Then in particular it must give rise to a $d$-dimensional strategy $\strategy'_{\kappa+1}$ such that $\val(G'_{\kappa+1},\strategy'_{\kappa+1}) = 1$ (simply by taking the measurement operators corresponding to questions $x = (1,\hat{x})$). By the entanglement bound of \Cref{thm:gapless-comp-formal}, it must be that the dimension $d$ is at least $\mathcal{E}(G^\super_{\kappa + 1},1)$. If this quantity is infinite, then we arrive at a contradiction and are done. Otherwise, there is a $d$-dimensional perfect strategy $\strategy^\super_{\kappa+1}$ for $G^\super_{\kappa + 1}$. Again, this must imply a $d$-dimensional perfect strategy for $G'_{\kappa+2}$. Continuing in this fashion, we either obtain a contradiction or deduce the existence of a $d$-dimensional perfect strategy for $G'_m$ for all $m \geq \kappa$. On the other hand, the entanglement bound of \Cref{thm:gapless-comp-formal} also implies that $\mathcal{E}(G'_m,1) \geq 2^{2m}$. Thus, $d \geq 2^{2m}$ for all $m \geq \kappa$, contradicting the assumption that $d$ is finite.
\end{proof}

\subsection{$\Pi_1$-completeness of the exact $co$-value problem} \label{subsec:application_pi1}

As a warmup, we present an application of the super compression procedure to show that the exact $co$-value problem (i.e. determining whether $\val_{co}(G) = 1$) is complete for $\Pi_1$, also known as $\mathsf{coRE}$. This was first shown by Slofstra~\cite{slofstra_tsirelsons_problem_and_an_embedding_theorem} using very different techniques based on group theory.

\begin{theorem}
\label{thm:pi-1-completeness-of-exact-co}
The exact $co$-value problem is complete for $\Pi_1$.
\end{theorem}
\begin{proof}
The easy direction is that the exact $co$-value problem is contained in $\Pi_1$ because one can express it as a $\Pi_1$ sentence: for all nonlocal games $G$, $\val_{co}(G) = 1$ if and only if $\forall x \, \phi(x)$ where $\phi(x)$ is a computable predicate that is true when the $x$-th level of the semidefinite programming hierarchy of~\cite{navascues2008convergent,doherty2008quantum} computes an upper bound of $1$ on $\val_{co}(G)$. In other words, the best upper bound on the commuting operator value of $G$ computed by the $x$-th level of the hierarchy is $1$. If this is true for all $x$, then this implies that $\val_{co}(G) = 1$. On the other hand, if $\val_{co}(G) < 1$, then there exists a level $x$ such that $\phi(x)$ is false.

Now we turn to the other direction. To prove $\Pi_1$-hardness, we reduce an arbitrary $\Pi_1$ sentence $S = \forall x \, \phi(x)$ to a nonlocal game $G$ such that $S$ is true if and only if $\val_{co}(G) = 1$. 

Define the Turing machine $T_{\phi}$ that halts on the empty input if and only if the sentence $S$ is false:

\vspace{10pt}
\IncMargin{1em}
\begin{algorithm}[H]
\DontPrintSemicolon

\For{$x \in \{0,1\}^* $}{
    If $\phi(x)$ is false then halt.
}

\caption{Specification of $T_\phi$.}
\end{algorithm}\DecMargin{1em}
\vspace{10pt}

Next, define the sequence of games $\UGS_{\phi} = (G_n)_{n \in \N}$ with  verifier $\verifier = (D,C)$, where $C(x,y) = 1$ if and only if $x = y$, and where the decider $D$ is defined as follows:

\vspace{10pt}
\IncMargin{1em}
\begin{algorithm}[H]
\DontPrintSemicolon

\textbf{Input}: $n, x, y, a,b$

If $T_{\phi}$ halts in $n$ steps, reject.

If any of $x,y,a,b$ exceed $n$ bits, reject.

If $x = y$ but $a \neq b$, reject.

Otherwise, accept.

\caption{Specification of Turing machine $D$.}
\label{alg:d-super}
\end{algorithm}\DecMargin{1em}
\vspace{10pt}

Notice that $\max \{ \TIME_D(n),\TIME_c(n) \} \leq O(n)$, which is at most $n^2$ for sufficiently large $n$. Furthermore, $\val_{co}(G_n) = 1$ if and only if the Turing machine $T_\phi$ does not halt in $n$ steps. 
Furthermore, if $T_\phi$ does not halt in $n$ steps, then there exists an oracularizable synchronous (``osync'') strategy $\strategy_n$ such that $\val(G_n,\strategy_n) = 1$: the strategy is to output a fixed answer no matter what the question is. 

We apply super compression to the family of games $\UGS_\phi$: the output of $\alg{SuperCompress}_\alpha(D,C)$ where $\alpha = 2$ is a verifier $(D^\super,C^\super)$ for a sequence of games $\UGS^\super = (G_n^{\super})_{n \in \N}$ such that $\val_{co}^s(G^{\super}_\kappa) = 1$ if and only if there exists an osync value-$1$ strategy $\strategy_n$ for $G_n$, where $\kappa$ is defined as in \Cref{thm:super-comp}.

Thus if $S$ is true, then $T_\phi$ never halts, and there exists an osync strategy $\strategy_n$ such that $\val(G_n,\strategy_n) = 1$ for all $n \in \N$, and thus $\val_{co}^s(G^\super_\kappa) = 1$. On the other hand, if $S$ is false and $T_\phi$ does halt in some time $t$, then $\val_{co}^s(G_n) < 1$ for all $n \geq t$, which implies that $\val_{co}^s(G^\super_\kappa) < 1$. 

By~\cite{paulsen2016estimating}, since $G^\super_\kappa$ is a synchronous game, we have that $\val_{co}^s(G^\super_\kappa) = 1$ if and only if $\val_{co}(G^\super_\kappa) = 1$. This, combined with the fact that the mapping from the $\Pi_1$ sentence $S$ to the game $G^\super_\kappa$ is computable, implies that the exact $co$-value problem is $\Pi_1$-hard.

\end{proof}

Note that the exact same proof, considering $q$-type strategies rather than $co$-type strategies, shows that the exact $q$-value problem is hard for $\Pi_1$. While we improve this lower bound to $\Pi_2$ in the next section, we note that this directly implies that the set of quantum correlations is not closed, a result that was also established by Slofstra in~\cite{slofstra_set_of_quantum_correlations}.\footnote{Briefly, the set of quantum correlations on $n$ inputs and $k$ outputs, denoted by $C_q(n,k)$, is the (convex) set of all vectors $p_{xyab} \in \R^{n \times n \times k \times k}$ such that 
\[
	p_{xyab} = \bra{\psi} A^x_a \otimes B^y_b \ket{\psi}
\]
for some dimension $d$, some quantum state $\ket{\psi} \in \C^d \otimes \C^d$, and some POVMs $\{A^x_a\}, \{B^y_b\}$. } Again, the proof approaches are quite different: his proof uses techniques from approximate representation theory as well as group theory. 

\begin{corollary}[\cite{slofstra_set_of_quantum_correlations}]
The set of quantum correlations is not closed.
\end{corollary}
\begin{proof}

Let $S$ be a true $\Pi_1$ sentence. The construction of the game $G^\super_\kappa$ from $S$ in \Cref{thm:pi-1-completeness-of-exact-co}, by \Cref{thm:super-comp}, has the property that $\val_q(G^\super_\kappa) = 1$ but there is no finite-dimensional strategy $\strategy$ that actually achieves value $1$ in the game. 

\end{proof}

\subsection{$\Pi_2$-completeness of the exact $q$-value problem} \label{subsec:application_pi2}

We now prove the main result of this paper, which is the $\Pi_2$-completeness of the exact $q$-value problem. 
As explained in \Cref{sec:compression-intro}, we combine our gapless compression theorem with a consequence of the $\MIP^* = \RE$ theorem from~\cite{ji_mip_re}, which we state in the following theorem. In the theorem, nonlocal games $G$ are represented via an integer $n \in \N$, and a pair of Turing machines $(D,C)$ where $D$ represents the decider for $G$ (so is a $4$-input Turing machine) and $C$ represents the checker (so is a $2$-input Turing machine). The game $G$ is then defined to be $(\cal{X},\cal{A},D)$ where $\cal{X} = \cal{A} = \{0,1\}^n$. The checker $C$, on input $(x,y) \in \cal{X} \times \cal{X}$, indicates whether $(x,y)$ is trivial for $G$.

\begin{theorem}[\cite{ji_mip_re}]
\label{thm:halting-game}
There is a universal constant $\lambda_\Halt \in \N$ and algorithm $\alg{HaltingGame}$ that takes as input the description of a $\Sigma_1$ sentence $S$ and outputs a tuple $(D,C)$ for a nonlocal game $G$ such that
\begin{enumerate}
    \item (Completeness) If $S$ is true, then 
    \[\sup_{\text{finite-dim osync } \strategy} \val(G,\strategy) = 1.\]
    \item (Soundness) If $S$ is false, then 
    \[\val^s_q(G) < 1.\]
    \item (Complexity bounds) Letting $|S|$ denote the description length of the sentence $S$, we have \[\max \Big \{ \TIME_{C}, \TIME_{D}, \TIME_{\alg{HaltingGame}(S)} \Big \} \leq O(|S|^{\lambda_{\Halt}}) \]
    where $\TIME_C,\TIME_D$ denote the time complexities of $C,D$ (on any input), and $\TIME_{\alg{HaltingGame}(S)}$ denotes the time complexity of $\alg{HaltingGame}$ on input $S$.
\end{enumerate}
\end{theorem}

\begin{proof}
	This is a corollary of~\cite[Theorem 12.7]{ji_mip_re} which reduces the Halting problem to deciding whether the $q$-value of a nonlocal game is equal to $1$ or at most $1/2$. To obtain the present theorem, we first observe that every $\Sigma_1$ sentence $S = \exists x \, \phi(x)$ can be expressed as an equivalent instance of the Halting problem: define the Turing machine $M_S$ that on the empty input, starts looping over all $x$ and evaluates $\phi(x)$. If it finds an $x$ such that $\phi(x)$ is true, then it halts. Clearly $S$ is true if and only if $M_S$ halts. 
	
	The game $H$ corresponding to $M_S$ from~\cite[Theorem 12.7]{ji_mip_re} is synchronous and the decider complexity is at most some polynomial in the description length of $S$. However, the question distribution $\mu$ of the game $H$ is not uniform. Without loss of generality, assume that the question and answer sets of $H$ are represented by $n$-bit strings. Because the reduction from $M_S$ to $H$ is efficient, we have that $n = \poly(|S|)$. 
	
	The game $G$ that we construct will be $H$ but with a uniform distribution over all $n$-bit question pairs $(x,y)$. Whenever a sampled question pair $(x,y)$ is not in the support of $\mu$, the decider $D$ of $G$ will automatically accept (and thus $(x,y)$ is a trivial question). Otherwise, the decider from the game $H$ is invoked. The key thing to note is that $\val_q(H) = 1$ if and only if $\val_q(G) = 1$. Furthermore, since $G$ is a synchronous game (since $H$ is a synchronous game), it holds that $\val_q^s(G) = 1$ if and only if $\val_q(G) = 1$. 
	
	Finally, since determining the support of the question distribution of $H$ can be done in $\poly(|S|)$ time, we obtain a checker $C$ for the game $G$ that runs in $\poly(|S|)$ time. Thus, on input $S$, the algorithm $\alg{HaltingGame}$ can output the tuple $(D,C)$ which satisfies the conclusions of the theorem. \end{proof}

We break up the proof of the $\Pi_2$ completeness of the exact $q$-value problem into two parts. First we show hardness.

\begin{lemma}
The exact $q$-value problem is hard for $\Pi_2$. \end{lemma}
\begin{proof}

Fix a $\Pi_2$ sentence $S = \forall x \exists y \, \phi(x,y)$ where $\phi$ is a computable predicate. For every $n \in \N$ define the $\Sigma_1$ sentence 
\[
S_n = \exists y_1,\ldots,y_n \, \bigwedge_{i = 1}^n \phi(i,y_i).
\]
Thus the sentence $S$ is true if and only if the sentences $S_n$ are true for all $n \in \N$. Note that if $S_n$ is true then $S_i$ is true for all $i \leq n$ .

Using $\alg{HaltingGame}$ we construct the sequence of games $\UGS_{\phi} = (G_n)_{n \in \N}$ with verifier $\verifier = (D,C)$. Let
\[
	c_n = \sup_{\text{finite-dim osync } \strategy_n} \val(G_n,\strategy_n),
\]
then these games have the property that $c_n = 1$ if and only if the sentence $S_n$ is true.

\vspace{10pt}
\IncMargin{1em}
\begin{algorithm}[H]
\DontPrintSemicolon

\textbf{Input}: $n, x, y, a,b$

Compute the game decider and checker $(D_n,C_n)$ for $\alg{HaltingGame}(S_n)$.

If $D_n(x,y,a,b)$ accepts, then accept. 

Otherwise, reject.

\caption{Specification of Turing machine $D$.}
\end{algorithm}\DecMargin{1em}
\vspace{10pt}

\vspace{10pt}
\IncMargin{1em}
\begin{algorithm}[H]
\DontPrintSemicolon

\textbf{Input}: $n, x, y$

Compute the game decider and checker $(D_n,C_n)$ for $\alg{HaltingGame}(S_n)$.

Output $C_n(x,y)$.

\caption{Specification of Turing machine $C$.}
\end{algorithm}\DecMargin{1em}
\vspace{10pt}

For large enough $n$ the verifier is bounded by 
\[\max \Big \{ \TIME_C(n), \TIME_D(n) \Big \} \leq n^{\lambda_{\Halt}+1}\] since

\[\max \Big \{ \TIME_{C_n}, \TIME_{D_n}, \TIME_{\alg{HaltingGame}(S_n)} \Big \} \leq (n|S|)^{\lambda_{\Halt}}.\]

We apply super compression to the family of games $\UGS_\phi$: the output of $\alg{SuperCompress}_{\alpha}(D,C)$ where $\alpha = \lambda_{\Halt} + 1$ is a verifier $(D^\super,C^\super)$ for a sequence of games $\UGS^\super = (G_n^{\super})_{n \in \N}$ such that $\val_{q}^s(G^{\super}_\kappa) = 1$ if and only if $c_n = 1$ for all $n \geq \kappa$, where $\kappa$ is defined as in \Cref{thm:super-comp}.

Therefore, $\val_{q}^s(G^{\super}_\kappa) = 1$ if and only if the sentences $S_n$ are true for $n \geq \kappa$, which is equivalent to the $\Pi_2$ sentence $S$ being true. We have therefore reduced the problem of deciding an arbitrary $\Pi_2$ sentence to deciding the exact $q$-value problem.
\end{proof}

\noindent Finally, we argue that the exact $q$-value problem is contained in $\Pi_2$. 

\begin{lemma}
The exact $q$-value problem is in $\Pi_2$.
\end{lemma}
\begin{proof}
We will state the exact $q$-value problem as a $\Pi_2$ sentence. Fix a nonlocal game $G$ then we would like to decide if 
\[
\sup_{\text{finite-dim } \strategy} \val(G,\strategy) = 1.
\]

Let $\cal{S}_{\epsilon}^d$ be an $\epsilon$-net for quantum strategies of dimension $d \in \N$. This is a finite set, since strategies of a fixed dimension form a compact set \cite{compact}. Let $\cal{S}_\epsilon = \bigcup_{d \in \N}\cal{S}_\epsilon^d$. Then we can equivalently formulate the decision problem as
\[
\forall \epsilon \in (0,1] ~ \exists \strategy \in \cal{S}_\epsilon \text{ such that } \val(G,\strategy) > 1 - 2\epsilon.
\]
This in turn is equivalent to the $\Pi_2$ sentence
\[
\forall n \in \N ~ \exists \strategy \in \cal{S}_{\frac{1}{n}} \text{ such that } \val(G,\strategy) > 1 - \frac{2}{n}.
\]

\end{proof}

\noindent Putting the two together, we get:
\begin{theorem} \label{thm:pi-2-completeness-of-exact-q}
The exact $q$-value problem is complete for $\Pi_2$.
\end{theorem}

\subsection{Necessity of compression}
\label{sec:equivalence}
We will show how to compress nonlocal games given many-one reductions from arithmetical hierarchy classes to the corresponding $t$-value problems for $t \in \{q,co\}$. This shows that, in a certain sense, compression theorems are \emph{necessary} for proving the complexity lower bounds indicated in \Cref{fig:table}. In particular we construct \emph{super compression} procedures (procedures that map families of games to a single equivalent game).  

The following theorem was proved in~\cite{mousavi2020complexity}:

\begin{theorem} \label{thm:gap-preserving-supercompression}
Assume that the approximate $q$-value problem is $\Sigma_1$-hard. Then there exists a computable map $\alg{GapCompress}_{q}$ that takes in as input a description of a sequence of games  $\UGS = (G_n)_{n \in \N}$ and outputs the description of a single game $G'$ such that
 \begin{enumerate}
    \item $\val_q(G') = 1$ if $\val_q(G_n) = 1$ for some game $G_n \in \UGS$.
    \item $\val_q(G') < \frac{1}{2}$ if $\val_q(G_n) < \frac{1}{2}$ for every game $G_n \in \UGS$.
\end{enumerate}
\end{theorem}

Now we show that if the approximate $co$-value problem is $\Pi_1$ hard, then there exists a gap-preserving compression procedure for the commuting operator value of games. 

\begin{theorem} \label{thm:gap-amplifying-supercompression}
Assume that the approximate $co$-value problem is $\Pi_1$ hard. Then there exists a computable map $\alg{GapCompress}_{co}$ that takes in as input a description of a sequence of games  $\UGS = (G_n)_{n \in \N}$ and outputs the description of a single game $G'$ such that
\begin{enumerate}
    \item $\val_{co}(G') = 1$ if $\val_{co}(G_n) = 1$ for every game $G_n \in \UGS$.
    \item $\val_{co}(G') < \frac{1}{2}$, otherwise.
\end{enumerate}
\end{theorem}
\begin{proof}
Consider the following Turing machine $T^{co}_{\UGS}$: it interleaves running some number of levels of the NPA semidefinite programming hierarchy~\cite{navascues2008convergent} on each game $G_m$ in the sequence, trying to find a game $m$ for which $\val_{co}(G_m) < 1$. The completeness of the NPA hierarchy implies that if $\val_{co}(G_m) < 1$ for some $m$, then eventually a certificate will be found. Thus the Turing machine halts only if there exists $m$ such that $\val_{co}(G_m) < 1$.

\vspace{10pt}
\IncMargin{1em}
\begin{algorithm}[H]
\DontPrintSemicolon

\For{$n \in \N$}{
    \For{$m \in \{1, ..., n\}$}{
        Run the first $n$ levels of the NPA hierarchy for the game $G_m \in \UGS$.
        
        If there is a certificate that $\val_{co}(G_m) < 1$ then halt.
    }
}

\caption{Specification of $T^{co}_\UGS$}
\end{algorithm}\DecMargin{1em}
\vspace{10pt}

Consider the sentence $S$ defined as ``$\forall n \in \N,\,   T^{co}_{\UGS}$ does not halt in $n$ steps''. Note that $S$ is a $\Pi_1$ sentence, and since the approximate $co$-value problem is $\Pi_1$-hard, this means there is a corresponding game $G'$ computable from $S$ such that such that $\val_{co}(G') = 1$ if $T^{co}_{\UGS}$ never halts (i.e. $\val_{co}(G_m) = 1$ for all $m$), otherwise $\val_{co}(G') < \frac{1}{2}$.
\end{proof}

Next we show that $\Pi_1$-hardness of the \emph{exact} $co$-value problem implies a \emph{gapless} compression theorem for the commuting operator value of nonlocal games.

\begin{theorem} \label{thm:gapless-supercompression-necessity}
Assume that the exact $co$-value problem is $\Pi_1$ hard. Then there exists a computable map $\alg{GaplessCompress}_{co}$ that takes in as input a description of a sequence of games  $\UGS = (G_n)_{n \in \N}$ and outputs the description of a single game $G'$ such that $\val_{co}(G') = 1$ if and only if $\val_{co}(G_n) = 1$ for all $n \in \N$.
\end{theorem}
\begin{proof}
This follows exactly the same proof as above, except the reduction from the sentence $S$ to the game $G'$ is such that $\val_{co}(G_m) = 1$ for all $m$ if and only if $S$ is true if and only if $\val_{co}(G') = 1$. 
\end{proof}

Finally we prove that $\Pi_2$-hardness of the exact $q$-value problem implies a gapless compression theorem for the quantum value of nonlocal games.

\begin{theorem} \label{thm:gapless-supercompression}
Assume that the approximate $q$-value problem is $\Pi_2$-hard. Then there exists a computable map $\alg{GaplessCompress}_{q}$ that takes in as input a description of a sequence of games  $\UGS = (G_n)_{n \in \N}$ and outputs the description of a single game $G'$ such that $\val_q(G) = 1$ if and only if $\val_q(G_n) = 1$ for all $n \in \N$.
\end{theorem}
\begin{proof}

Consider the following Turing machine $T^q_\UGS$: it takes in as input a precision parameter $\eps$ and an integer $m$, and it searches for a finite-dimensional strategy $\strategy$ (specified with precision $\eps$) such that the game $G_m$ in the sequence $\UGS$ has $\val(G_m,\strategy) \geq 1 - 2\eps$. This can be done because given a dimension $d \in \N$ and a precision parameter $\eps$, there is an algorithm to exhaustively search over $\cal{S}_{d,\eps}$, an $\eps$-net over $d$-dimensional quantum strategies.

\vspace{10pt}
\IncMargin{1em}
\begin{algorithm}[H]
\DontPrintSemicolon

\textbf{Input}: $\epsilon, m$

\For{$d \in \N$}{
    If there exists a strategy $\strategy \in \cal{S}_{\epsilon}^d$, an $\epsilon$-net for quantum strategies of dimension $d$, such that $\val(G_m,\strategy) > 1 - 2\eps$, then halt.

}

\caption{Specification of $T^{q}_\UGS$}
\end{algorithm}\DecMargin{1em}
\vspace{10pt}

Note that if $\val_q(G_m) = 1$, then for all $\eps > 0$ there exists a finite-dimensional strategy that achieves value at least $1 - 2\eps$. On the other hand, if $\val_q(G_m) < 1$, then there exists an $\eps$ for which \emph{all} finite dimensional strategies have value at most $1 - 2\eps$. Thus $\val_q(G_m) = 1$ for all $m \in \N$ if and only if the following sentence $S$ is true: ``$\forall k ,m  \, \exists n  \, T^{q}_{\UGS}$ halts on input $\Big(\frac{1}{k}, m \Big)$ in $n$ steps''. Note that $S$ is a $\Pi_2$ sentence, and by our assumption there exists a nonlocal game $G'$ that is computable from $S$ such that $\val_q(G') = 1$ if and only if $\val_q(G_m) = 1$ for all $m \in \N$.

\end{proof} 
\appendix
\newpage

\section{The pasting lemma}
\label{sec:pasting}

We now prove \Cref{lem:pasting}, which is reproduced below for convenience. Recall that $\algebra$ is a von Neumann algebra with a normal tracial state $\tau$.

\pasting*

We introduce some notation. For every integer $k \geq 1$, vector $\vec{a} \in \cal{A}^k$, and operator index sequence $s \in [M]^k$, define the operator
\[
	P^s_{\vec{a}} = A^{(s_1)}_{\vec{a}_1} \cdot A^{(s_2)}_{\vec{a}_2} \cdots A^{(s_k)}_{\vec{a}_k}.
\]
Note that $P^s = \{ P^s_{\vec{a}} \}_{a \in \cal{A}^k}$ is a general set of operators (not necessarily a POVM, because the operators are not positive).

We first prove the following utility Lemma. We use the following notational convention: given two operator sets $C = \{C_a\}_{a \in \cal{A}}$ and $D = \{D_b\}_{b \in \cal{B}}$, we write $C \cdot D$ to denote the operator set $\{C_a \cdot D_b \}_{a \in \cal{A},b \in \cal{B}}$. 

\begin{lemma}
\label{lem:pasting-utility}
For integers $k \geq 1$, for all all sequences $s \in [M]^k$, for all $i \in [M]$, we have
\[
	\| P^s \cdot A^{(i)} - A^{(i)} \cdot P^s \|_\tau \leq k \eps
\]
\end{lemma}
\begin{proof}
	We prove this via induction on $k$. The base case for $k = 1$ follows from the assumption of the approximate commutativity of the $A^{(i)}$ measurements. Assuming the inductive hypothesis holds for some $k \geq 1$, we now prove it for $k + 1$: let $s \in [M]^k, t \in [M]$. We can treat $(s,t)$ as an operator index sequence of length $k+1$. Then for all $i \in [M]$, we have
	\begin{align}
		&\| P^{s,t} \cdot A^{(i)} - A^{(i)} \cdot P^{s,t} \|_\tau = \| P^s \cdot A^{(t)} \cdot A^{(i)} - A^{(i)} \cdot P^{s} \cdot A^{(t)} \|_\tau \notag \\
							&\qquad \qquad \leq \Big \| P^s \cdot \Big ( A^{(t)} \cdot A^{(i)} - A^{(i)} \cdot A^{(t)} \Big) \Big \|_\tau + \Big \| \Big(P^s \cdot A^{(i)} - A^{(i)} \cdot P^{s} \Big) \cdot A^{(t)} \Big \|_\tau \label{eq:pasting-utility-1}
	\end{align}
	where the inequality follows from the triangle inequality of the $\tau$-norm on operator sets (\Cref{lem:triangle-inequality-tau-norm}). 
	
	We can bound the first term as
	\[
		\Big \| P^s \cdot \Big ( A^{(t)} \cdot A^{(i)} - A^{(i)} \cdot A^{(t)} \Big) \Big \|_\tau = \Big \| A^{(t)} \cdot A^{(i)} - A^{(i)} \cdot A^{(t)} \Big \|_\tau \leq \eps~.
	\]
	The inequality follows from the almost-commutativity of the $A$'s, and the first equality is because
	\begin{align*}
&= \sum_{\substack{\vec{a} \in \cal{A}^k \\ b,c \in \cal{A}}} \trace{ \Big ( A^{(t)}_b \cdot A^{(i)}_c - A^{(i)}_c \cdot A^{(t)}_b \Big)^* (P^s_{\vec{a}})^* P^s_{\vec{a}} \Big ( A^{(t)}_b \cdot A^{(i)}_c - A^{(i)}_c \cdot A^{(t)}_b \Big) } \\
	&= \sum_{b,c \in \cal{A}} \trace{ \Big ( A^{(t)}_b \cdot A^{(i)}_c - A^{(i)}_c \cdot A^{(t)}_b \Big)^* \Big ( A^{(t)}_b \cdot A^{(i)}_c - A^{(i)}_c \cdot A^{(t)}_b \Big) }
\end{align*}
	where we used the fact that $\sum_{\vec{a} \in \cal{A}^k}(P^s_{\vec{a}})^* P^s_{\vec{a}} = \id$.
	
	The second term in~\eqref{eq:pasting-utility-1} can be similarly bounded as
	\[
	\Big \| \Big(P^s \cdot A^{(i)} - A^{(i)} \cdot P^{s} \Big) \cdot A^{(t)} \Big \|_\tau = \Big \| P^s \cdot A^{(i)} - A^{(i)} \cdot P^{s}  \Big \|_\tau \leq k \eps
	\] 
	by the inductive hypothesis. Thus we can bound~\eqref{eq:pasting-utility-1} by $(k+1)\eps$, completing the induction.
\end{proof}

	For the remainder of the proof let $k = M$. Let $s = (1,2,\ldots,M) \in [M]^k$ denote an operator index sequence. For all $\vec{a} \in \cal{A}^k$, define
	\[
		Q_{\vec{a}} = P^s_{\vec{a}} (P^s_{\vec{a}})^* \;.
	\]
	Note that $Q_{\vec{a}}$ is positive and furthermore $\{ Q_{\vec{a}} \}$ forms a POVM with outcomes in $\cal{A}^k$ (this uses the fact that the $A^{(i)}_a$ operators are projections).  
	
	We now calculate the closeness of $Q_{[\vec{a} \mapsto \vec{a}_i \mid b]}$ to the individual $A^{(i)}_b$'s:
	\begin{align*}
		\sum_{b \in \cal{A}} \| Q_{[\vec{a} \mapsto \vec{a}_i \mid b]} - A^{(i)}_b \|_\tau^2 &= \sum_{b \in \cal{A}} \tau \Big ( \Big (Q_{[\vec{a} \mapsto \vec{a}_i \mid b]} - A^{(i)}_b \Big )^2 \Big) \\
		&\leq 2 - 2 \sum_{b \in \cal{A}} \tau \Big ( Q_{[\vec{a} \mapsto \vec{a}_i \mid b]} A^{(i)}_b \Big) \\
		&= 2 - 2 \sum_{\vec{a}} \tau \Big ( Q_{\vec{a}} A^{(i)}_{\vec{a}_i} \Big)
\end{align*}
	
	We give a lower bound on the magnitude of the second term. Spliting the index sequence $s = (s_{<i},i,s_{>i})$  and answer tuples $\vec{a} = (\vec{a}_{<i},\vec{a}_i,\vec{a}_{>i})$, we get
		\begin{align*}
	\sum_{\vec{a}} \tau\Big ( Q_{\vec{a}} A^{(i)}_{\vec{a}_i} \Big) &= \sum_{\vec{a}} \tau\Big ( P_{\vec{a}_{< i}}^{s_{<i}} \cdot A^{(i)}_{\vec{a}_i} \cdot  P_{\vec{a}_{> i}}^{s_{>i}} \cdot (P_{\vec{a}_{> i}}^{s_{>i}})^* \cdot A^{(i)}_{\vec{a}_i} \cdot (P_{\vec{a}_{< i}}^{s_{<i}})^* \cdot A^{(i)}_{\vec{a}_i} \Big) \\
	&= \sum_{\vec{a}_{< i}, \vec{a}_i} \tau\Big ( P_{\vec{a}_{< i}}^{s_{<i}} \cdot A^{(i)}_{\vec{a}_i} \cdot (P_{\vec{a}_{< i}}^{s_{<i}})^* \cdot A^{(i)}_{\vec{a}_i} \Big) \\
	&= \sum_{\vec{a}_{< i}, \vec{a}_i} \tau\Big ( P_{\vec{a}_{< i}}^{s_{<i}} \cdot A^{(i)}_{\vec{a}_i} \cdot (P_{\vec{a}_{< i}}^{s_{<i}})^* \Big) + \tau\Big ( P_{\vec{a}_{< i}}^{s_{<i}} \cdot A^{(i)}_{\vec{a}_i} \cdot  \Big( (P_{\vec{a}_{< i}}^{s_{<i}})^* \cdot A^{(i)}_{\vec{a}_i} - A^{(i)}_{\vec{a}_i} \cdot (P_{\vec{a}_{< i}}^{s_{<i}})^* \Big) \Big) \\
	&= 1 + \sum_{\vec{a}_{< i}, \vec{a}_i} \tau\Big ( P_{\vec{a}_{< i}}^{s_{<i}} \cdot A^{(i)}_{\vec{a}_i} \cdot  \Big( (P_{\vec{a}_{< i}}^{s_{<i}})^* \cdot A^{(i)}_{\vec{a}_i} - A^{(i)}_{\vec{a}_i} \cdot (P_{\vec{a}_{< i}}^{s_{<i}})^* \Big) \Big)
	\end{align*}
	We can bound the magnitude of the second term using Cauchy-Schwarz:
	\begin{align*}
	&\left | \sum_{\vec{a}_{< i}, \vec{a}_i} \tau\Big ( P_{\vec{a}_{< i}}^{s_{<i}} \cdot A^{(i)}_{\vec{a}_i} \cdot  \Big( (P_{\vec{a}_{< i}}^{s_{<i}})^* \cdot A^{(i)}_{\vec{a}_i} - A^{(i)}_{\vec{a}_i} \cdot (P_{\vec{a}_{< i}}^{s_{<i}})^* \Big) \Big) \right | \\
	&\leq \sqrt{ \sum_{\vec{a}_{< i}, \vec{a}_i} \tau\Big ( \Big( P_{\vec{a}_{< i}}^{s_{<i}} \cdot A^{(i)}_{\vec{a}_i} - A^{(i)}_{\vec{a}_i} \cdot P_{\vec{a}_{< i}}^{s_{<i}} \Big)^* \Big(  P_{\vec{a}_{< i}}^{s_{<i}} \cdot A^{(i)}_{\vec{a}_i} - A^{(i)}_{\vec{a}_i} \cdot P_{\vec{a}_{< i}}^{s_{<i}}\Big)  \Big) }\cdot \sqrt{ \sum_{\vec{a}_{< i}, \vec{a}_i} \tau \Big( P_{\vec{a}_{< i}}^{s_{<i}} \cdot A^{(i)}_{\vec{a}_i} \cdot (P_{\vec{a}_{< i}}^{s_{<i}})^* \Big) } \\
	&\leq \sqrt{ \sum_{\vec{a}_{<i},\vec{a}_i} \left \| P_{\vec{a}_{< i}}^{s_{<i}} \cdot A^{(i)}_{\vec{a}_i} - A^{(i)}_{\vec{a}_i} \cdot P_{\vec{a}_{< i}}^{s_{<i}} \right \|_\tau^2} \\
	&\leq M \eps
	\end{align*}
	where the last inequality follows from \Cref{lem:pasting-utility}. Thus we deduce that 
	\begin{equation}
	\label{eq:pasting-2}
	\sqrt{\sum_{b \in \cal{A}} \| Q_{[\vec{a} \mapsto \vec{a}_i \mid b]} - A^{(i)}_b \|_\tau^2} \leq \sqrt{2M \eps}\;.
	\end{equation}
	
	Next we argue that the $Q_{\vec{a}}$ is ``almost projective''. Using that $\sum_{\vec{a}} \tau(Q_{\vec{a}}) = \sum_{\vec{a}} \tau(P_{\vec{a}}^s) = 1$, we get
	\begin{align*}
		\sum_{\vec{a}} \tau \Big ( Q_{\vec{a}} -  Q_{\vec{a}}^2 \Big) &= \sum_{\vec{a}} \tau \Big ( P_{\vec{a}}^s -  Q_{\vec{a}}^2 \Big) \\
		&= \sum_{\vec{a}} \tau \Big ( P_{\vec{a}}^s -  P_{\vec{a}}^s \cdot Q_{\vec{a}} \Big) + \tau( (P_{\vec{a}}^s - Q_{\vec{a}}) \cdot Q_{\vec{a}} ) \\
		&= \sum_{\vec{a}} \tau \Big ( P_{\vec{a}}^s -  P_{\vec{a}}^s \cdot (P_{\vec{a}}^s)^* \Big) + \tau( (P_{\vec{a}}^s - Q_{\vec{a}}) \cdot Q_{\vec{a}} ) + \tau( ( (P_{\vec{a}}^s)^* - Q_{\vec{a}}) \cdot P_{\vec{a}}^s ) \\
		&= \sum_{\vec{a}} \tau( (P_{\vec{a}}^s - Q_{\vec{a}}) \cdot Q_{\vec{a}} ) + \tau( ((P_{\vec{a}}^s)^* - Q_{\vec{a}}) \cdot P_{\vec{a}}^s )
	\end{align*}
	where in the last line we used that $P_{\vec{a}}^s \cdot (P_{\vec{a}}^s)^* = Q_{\vec{a}}$ and $\sum_{\vec{a}} \tau(Q_{\vec{a}}) = \sum_{\vec{a}} \tau(P_{\vec{a}}^s) = 1$. Using Cauchy-Schwarz and the fact that $\sum_{\vec{a}} \| P^s_{\vec{a}} \|_\tau^2$ and $\sum_{\vec{a}} \| Q_{\vec{a}} \|_\tau^2$ are most $1$, this last line is at most $2 \sqrt{ \sum_{\vec{a}} \| P_{\vec{a}}^s - Q_{\vec{a}} \|_\tau^2 }$. To bound this, we note that we can express
	$P^s_{\vec{a}}$ and $Q_{\vec{a}}$ as longer products 
	\[
		P^t_{\vec{b}} = P^{(s_1)}_{\vec{a}_1} \cdot P^{(s_1)}_{\vec{a}_1} \cdots P^{(s_k)}_{\vec{a}_k} \cdot P^{(s_k)}_{\vec{a}_k}~,
	\qquad \qquad
		P^u_{\vec{c}} = P^{(s_1)}_{\vec{a}_1} \cdots P^{(s_k)}_{\vec{a}_k} \cdots \cdot P^{(s_1)}_{\vec{a}_1} 
	\]
	where $t = (s_1,s_1,\ldots,s_k,s_k) \in [M]^{2k}$ and $u = (s_1,\ldots,s_k,s_k,\ldots,s_1)$, and $\vec{b} = (\vec{a}_1,\vec{a}_1,\ldots,\vec{a}_k,\vec{a}_k)$ and $\vec{c} = (\vec{a}_1,\ldots,\vec{a}_k,\vec{a}_k,\ldots,\vec{a}_1)$. In particular, let $\pi$ denote a permutation on $2k$ elements such that $\pi(\vec{b}) = \vec{c}$. 
	Thus
	\[
		\sqrt{ \sum_{\vec{a} \in \cal{A}^k} \| P_{\vec{a}}^s - Q_{\vec{a}} \|_\tau^2}  = \sqrt{ \sum_{\vec{a}\in \cal{A}^k} \left \| P^t_{\vec{b}} - P^u_{\vec{c}} \right \|_\tau^2 }\leq \sqrt{ \sum_{\vec{b}\in \cal{A}^{2k}} \left \| P^t_{\vec{b}} - P^u_{\pi(\vec{b})} \right \|_\tau^2 }
	\]
	Let $\pi'$ be a permutation that differs from $\pi$ by a swap of adjacent elements. Then 
	\[
	\sqrt{ \sum_{\vec{b}\in \cal{A}^{2k}} \left \| P^t_{\vec{b}} - P^u_{\pi(\vec{b})} \right \|_\tau^2} \leq \eps
	\]
	by our assumption on the almost-commutativity of the $A$'s. Since $\pi$ can be formed from the identity permutation by swapping at most $(2k)^2$ adjacent elements, by the triangle inequality we have that 
	\[
	\sqrt{ \sum_{\vec{b}\in \cal{A}^{2k}} \left \| P^t_{\vec{b}} - P^u_{\pi(\vec{b})} \right \|_\tau^2 } \leq 4k^2 \eps
	\]
	and therefore $\sum_{\vec{a}} \tau \Big ( Q_{\vec{a}} -  Q_{\vec{a}}^2 \Big) \leq 8 M^2 \eps$.

	Thus we can apply the Projectivization Lemma (\Cref{lem:projectivization}) to the POVM $\{Q_{\vec{a}} \}$ to obtain a projective measurement $R = \{ R_{\vec{a}} \}$ such that
	\[
		R_{\vec{a}} \approx_\eta Q_{\vec{a}}
	\]
	where $\eta = \delta_{proj}(8M^2 \eps)$ where $\delta_{proj}(\cdot)$ is the error function from the Projectivization Lemma. Using the fact that $R$ is projective, we get from \Cref{lem:closeness-to-consistency} that
	\[
		R_{\vec{a}} \simeq_\eta Q_{\vec{a}}.
	\]
	Using the Data Processing Lemma for consistency (\Cref{lem:data-processing}), we get that 
	\[
		R_{[\vec{a} \mapsto \vec{a}_i \mid b]} \simeq_\eta Q_{[\vec{a} \mapsto \vec{a}_i \mid b]} \;.
	\]
	Converting from consistency to closeness (\Cref{lem:consistency-consequences}) we get
	\[
		R_{[\vec{a} \mapsto \vec{a}_i \mid b]} \approx_{\sqrt{2\eta}} Q_{[\vec{a} \mapsto \vec{a}_i \mid b]}
	\]
	Finally, we get
	\begin{align*}
		\| R_{[\vec{a} \mapsto \vec{a}_i]} - A^{(i)} \|_\tau &\leq \left \| R_{[\vec{a} \mapsto \vec{a}_i]} -  Q_{[\vec{a} \mapsto \vec{a}_i]} \right \|_\tau + \left \| Q_{[\vec{a} \mapsto \vec{a}_i]} - A^{(i)} \right \|_\tau \\
		 &\leq \sqrt{2\eta} + \sqrt{2M\eps}  \;.
	\end{align*}
	Thus we get
	\[
		R_{[\vec{a} \mapsto \vec{a}_i \mid b]} \approx_{\sqrt{2\eta} + \sqrt{2M\eps}} A^{(i)}_b  \;.
	\]	
	Setting $\delta_{pasting}(M,\cal{A},\eps) = \sqrt{2\eta} + \sqrt{2M\eps}$ proves the Lemma.

\newpage

\section{Complexity of noncommutative polynomial optimization}
\label{app:polyopt}

For convenience we recall the general formulation of noncommutative polynomial optimization (ncPO for short): given Hermitian polynomials $p,q_1,\ldots,q_m$ in $2n$-noncommutative variables \\ $(x_1,\ldots,x_n,x_1^*,\ldots,x_n^*)$ over $\C$, compute the value of the following optimization program:
\begin{align*}
	\sup \qquad & \bra{\phi} p(X) \ket{\phi} \\
		\text{s.t.} \qquad &q_i(X) \succeq 0 \qquad \text{for $i=1,\ldots,m$}
\end{align*}
The supremum is over choices of tuples $(\cal{H},X,\phi)$ where $\cal{H}$ is a Hilbert space, $X$ is an $n$-tuple of bounded operators acting on $\cal{H}$, and $\ket{\phi}$ is a unit vector on $\cal{H}$. The notation $p(X)$ and $q_i(X)$ indicates that we evaluate each of the indeterminates $x_i$ with the operator $X_i$ and $x_i^*$ with the adjoint $X_i^*$, respectively. We consider two different variations of a ncPO program $P$; if we restrict the supremum to vary only over finite -- but unbounded -- dimensional Hilbert spaces then we call the program \emph{finite-dimensional} and let $\omega_{\mathrm{fin}}(P)$ denote the value of the program. Otherwise we call the program \emph{infinite-dimensional} and let $\omega_{\infty}(P)$ denote the value.

\begin{proposition} \label{prop:gamesdp}
Given a nonlocal game $G = (\cal{X},\cal{A},\mu,D)$ there exists a ncPO program $P$ where $\omega_{\mathrm{fin}}(P) = \omega_{q}(G)$ and $\omega_{\infty}(P) = \omega_{co}(G)$.
\end{proposition}
\begin{proof}
Define the following optimization problem $P$ over $4|\cal{X}||\cal{A}|$ variables $\{A^x_a\}, \{B^y_b\}, \{(A^x_a)^*\}, \{(B^y_b)^*\}$. The objective polynomial $p$ to be optimized is 
\[
p = \sum_{x,y \in \cal{X}} \sum_{a,b \in \cal{A}} \, \mu(x,y) \, A^x_a B^y_b \, D(x,y,a,b)~.
\]
To enforce that the operators $\{A^x_a\}, \{B^y_b\}$ correspond to POVMs, we add the constraints
\begin{enumerate}
	\item $A^x_a = (A^x_a)^*$, $B^y_b = (B^y_b)^*$ (i.e. the operators are self-adjoint);
	\item $A^x_a, B^y_b \succeq 0$ (i.e. operators are positive);
	\item $\sum_a A^x_a = \sum_b B^y_b = \id$ for all $x,y$ (i.e. operators form POVMs);
	\item $[A^x_a, B^y_b] = 0$ (i.e. Alice's and Bob's operators commute)~.
\end{enumerate}
It is easy to see that these constraints can be expressed as Hermitian polynomial inequalities. The value of this optimization problem corresponds exactly to the definition of $\omega_q$ (in the finite-dimensional case) and $\omega_{co}$ (in the infinite-dimensional case).
\end{proof}

\begin{theorem} \label{thm:ncPO-Sigma1}
Deciding if $\omega_{\mathrm{fin}}(P) \geq c$ or $\omega_{\mathrm{fin}}(P) \leq c - \eps$ for fixed $\eps > 0$ is complete for $\Sigma_1$.
\end{theorem}
\begin{proof}
$\Sigma_1$-hardness follows from Proposition \ref{prop:gamesdp} and the $\Sigma_1$-hardness of approximating $\omega_q$~\cite{ji_mip_re}.

To show that the problem is contained in $\Sigma_1$, we first argue that, when restricting the Hilbert space to have a \emph{fixed} dimension $d$, a ncPO program $P$ can be recast as a \emph{commutative} polynomial optimization problem $P_d$ over $\C$. Let $p$ denote the objective polynomial and let $q_1,\ldots,q_m$ denote the constraint polynomials. Let $x_1,\ldots,x_n$ (and $x_1^*,\ldots,x_n^*$) denote the indeterminates of the program. 

The optimization problem $P_d$ is defined as follows. To every noncommutative indeterminate $x_i$ we associate $d^2$ commutative indeterminates $x_i^{ab}$ for $1 \leq a,b \leq d$. Intuitively these indeterminates correspond to the entries of the $d \times d$ matrix that is supposed to be substituted in for $x_i$. We also introduce $d$ indeterminates $y_1,\ldots,y_d$ to represent the unit vector $\ket{\phi} \in \C^d$. 

The objective polynomial of $P_d$ is a polynomial $p_d$ that expresses the quantity \\ $\bra{\phi} p(x_1,\ldots,x_n,x_1^*,\ldots,x_n^*) \ket{\phi}$ when $\ket{\phi}$ and the indeterminates $x_i$ are substituted with the corresponding complex numbers. There are constraint polynomials in $P_d$ that encode the fact that the $x_i$ matrices are self-adjoint, and furthermore the vector $(y_1,\ldots,y_d)$ is a unit vector. To check the positivity constraints $q_i \succeq 0$ in $P$ we can instead check that all the leading principal minors of $q_i$ are positive. The order $k$ leading principal minor of a $d \times d$ matrix is the determinant of the submatrix obtained from deleting the last $d-k$ rows and columns of the matrix.

Thus, by construction, the value of $P_d$ is the value of $P$ when restricted to $d$-dimensional Hilbert spaces. We thus have $\omega_{\mathrm{fin}}(P) = \lim_{d \to \infty} \omega(P_d)$. Therefore, if $\omega_{\mathrm{fin}}(P) \geq c$ then there exists $d \in \N$ such that $c - \omega(P_d) < \epsilon$. Otherwise, if $\omega_{\mathrm{fin}}(P) \leq c - \eps$ then no such $d \in \N$ exists.

Therefore we have reduced the problem to deciding whether there exists a dimension $d$ such that $c - \omega(P_d) < \epsilon$. This corresponds to deciding a $\Sigma_1$ sentence as determining whether $c - \omega(P_d) \leq \eps$ is decidable, due to the decidability of the first order theory of the complex numbers~\cite{harrison-complex} (which is analogous to the decidability of the first order theory of the reals~\cite{decidabilityreals,existentialtheoryreal}).

\end{proof}

\begin{theorem}
Deciding if $\omega_{\mathrm{fin}}(P) \geq c$ is complete for $\Pi_2$.
\end{theorem}
\begin{proof}
$\Pi_2$-hardness follows from Proposition \ref{prop:gamesdp} and Theorem \ref{thm:pi-2-completeness-of-exact-q}. 

Furthermore, deciding if $\omega_{\mathrm{fin}}(P) \geq c$ is equivalent to deciding if $c - \omega_{\mathrm{fin}}(P) < \frac{1}{n}$ for every $n \in \N$. Therefore, following from Theorem \ref{thm:ncPO-Sigma1}, we can state the decision problem as a $\Pi_2$ sentence.
\end{proof}

\begin{theorem}
Deciding if $\omega_{\infty}(P) \geq c$ is complete for $\Pi_1$.
\end{theorem}
\begin{proof}
$\Pi_1$-hardness follows from Proposition \ref{prop:gamesdp} and Theorem \ref{thm:pi-1-completeness-of-exact-co}. The inclusion is due to \cite{navascuesconvergentncPO} where they construct a sequence of commutative polynomial optimization relaxations $\{P_i\}_{i \in \N}$ where their values converge to the value of a given ncPO. Then we can decide if $\omega_{\infty}(P) \geq c$ by the $\Pi_1$ sentence
\[
\forall i \in \N, \,\, \omega(P_i) \geq c
\]
where the $\omega(P_i)$'s converge from above to the the value of  $\omega_{\infty}(P)$.
\end{proof}

\newpage

\bibliographystyle{alpha}
\bibliography{bibliography.bib}

\end{document}